%% file: main.tex
\documentclass[11pt]{article}
\pdfoutput=1
\usepackage{graphicx,amsmath,amssymb,url,epstopdf,cite}
\usepackage[ruled,vlined]{algorithm2e}
\usepackage{hyperref}
\hypersetup{
    colorlinks=true, %
    linktoc=all,     %
    linkcolor=blue  %
}
\usepackage[dvipsnames]{xcolor}
\usepackage{todonotes}
\usepackage{caption}
\usepackage{subcaption}
\usepackage{enumitem}
\usepackage{booktabs}
\usepackage{multirow}
\usepackage{lipsum}

\graphicspath{{./figures/}{../figures/}}
\usepackage{fullpage}

\usepackage{amsthm}
\usepackage[titletoc]{appendix}

\usepackage{hyperref}
\hypersetup{
    colorlinks=true,
    linkcolor=violet,
    filecolor=violet,
    urlcolor=violet,
    citecolor =violet,
    pdfpagemode=FullScreen,
    }
\usepackage[capitalize]{cleveref}
\newtheorem{theorem}{Theorem}[section]

\newtheorem{proposition}[theorem]{Proposition}

\newcommand{\dgm}{\operatorname{dgm}}

\newcommand{\R}{\mathbb{R}}
\newcommand{\calI}{\mathcal{I}}
\newcommand{\e}{\varepsilon}
\renewcommand{\epsilon}{\varepsilon}

\title{\bf Robust Zero-crossings Detection in Noisy Signals using Topological Signal Processing}

\author{Sunia Tanweer\thanks{tanweer1@msu.edu} \and Firas A.~Khasawneh\thanks{khasawn3@msu.edu}  \and  Elizabeth Munch\thanks{ muncheli@msu.edu} }

\begin{document}
\maketitle

\par\noindent\rule{\textwidth}{0.4pt}
\input{./sections/sec-abstract}
\vspace{5mm}

\input{./sections/sec-intro}
\input{./sections/sec-persistent-homology}

\input{./sections/sec-numerical-computations}
\input{./sections/sec-results}

\input{./sections/sec-conclusion}
\section{Acknowledgements}
This material is based upon work supported by the Air Force Office of Scientific Research under award number FA9550-22-1-0007.

\bibliographystyle{unsrt}
\bibliography{bibliography}

\appendix
\begin{appendices}
    \newpage

\input{./sections/sec-appendixA}
    \newpage

\input{./sections/sec-appendixB}
    \newpage

\input{./sections/sec-appendixC}
    \newpage

\input{./sections/sec-appendixD}
\end{appendices}

\end{document}

%% file: sections/sec-abstract.tex
\begin{center}
\section*{Abstract}
\label{sec:abstract}    
\end{center}

In this article, we explore a novel application of zero-dimensional persistent homology from Topological Data Analysis (TDA) for bracketing zero-crossings of both one-dimensional continuous functions, and uniformly sampled time series. We present an algorithm and show its robustness in the presence of noise for a range of sampling frequencies. In comparison to state-of-the-art software-based methods for finding zeros of a time series, our method generally converges faster, provides higher accuracy, and is capable of finding all the roots in a given interval instead of converging only to one of them. We also present and compare options for automatically setting the persistence threshold parameter that influences the accurate bracketing of the roots.

%% file: sections/sec-intro.tex
\section{Introduction}
\label{sec:intro}

Zero--crossings are the locations where a continuous function changes its sign. Determining these zero-crossings is a classical but significant problem in many fields such as engineering, medicine, and physical sciences. Specifically, zero-crossings have been used for signals' frequency determination \cite{Duric2005, Djuric2008, Friedman1994}, estimation of muscle fatigue \cite{Masuda1982}, distinction between neutron and gamma \cite{Bayat2012}, detection of short circuit faults in induction motors \cite{Ukil2011}, DC motor speed control \cite{Pindoriya2016}, speed measurement of land vehicles \cite{Misans2012}, strain estimation for elastography \cite{Srinivasan2003}, and even for recognition of hand-written characters \cite{Raju2006}. 

Although the zero-crossings problem of finding where the function $f(x)=0$ is significant and has an extensive history, there have been a number of zero-bracketing methods proposed in literature each with one's own set of drawbacks. These methods are classified into open bracketing schemes---such as Newton-Raphson, fixed point, and secant methods---and closed bracketing schemes, such as bisection and regula-falsi \cite{Antia2002}. The most fundamental closed method of bracketing a zero is the bisection method which, albeit slow, promises convergence for a continuous function $f(x)$ given an initial interval $[a, b]$ such that $f(a)f(b) < 0$. Another popular bracketing method is the Regula Falsi \cite{Antia2002} which shows a faster convergence than bisection method, except in the fatal cases of a function with a flat or steep slope. Various improvements have been made using these algorithms as the foundation. One such modification is seen in Suhadolnik's~\cite{Suhadolnik2012} combined method of switching between Bisection and Regula Falsi for bracketing roots of nonlinear equations. In that method, Suhadolnik uses quadratic interpolation to fit the known two-points of the function and the estimate from Bisection/Regula Falsi on a parabola. Other such methods of root bracketing have been presented by Alojz~\cite{Suhadolnik2013}, Razbani~\cite{Razbani2015}, Kavvadias~\cite{Kavvadias2005}, Kodnyanko~\cite{Kodnyanko2021}, Badr et al.~\cite{Badr2022}, Hussein et al.~\cite{Hussein2022}, and Daponte~\cite{daponte1995}. Somewhat more exotic algorithms for zero-detection in non-linear systems have been devised by Sadrpour et al.~\cite{Sadrpour2013}, Fried~\cite{Fried2013}, Li et al.~\cite{Li2020} and Kim et al.~\cite{Kim2020}. While some of these traditional methods require the ability to determine the function's derivative, others demand an intelligent initial guess of the root for a reasonable convergence rate.  Regardless of their particular pros and cons, all of these algorithms require the expression of the function, and cannot provide a root in case of a discrete time series sampled from an unknown function. 

Multiple successful ventures in engineering have been made for capturing the zero-crossings from time series by leveraging hardware elements such as diodes, comparators and filters \cite{Wall2003, Sreenivasan1983}. However, it is not always practical to build an electronic circuit for finding zero-crossings. In contrast, Molinaro and Sergeyev~\cite{Molinaro2001} developed an algorithmic, zero-finding approach---similar to Daponte et al.~\cite{daponte1995, Daponte1996}---based on estimating the Lipchitz constants of the signal. The algorithm, although fast in computation, only provides an estimate for the first zero-crossing in the domain, missing all the remaining zeros of the function. Furthermore, the algorithm requires the left boundary of the domain to hold a positive function value. 

Our work bypasses these grave shortcomings by presenting a novel approach for bracketing the zeros of a time series using 0-dimensional persistence, a tool from applied topology. The algorithm is fast, provides higher accuracy than comparable methods, and is capable of bracketing all the zeros of well-behaved signals. 

%% file: sections/sec-persistent-homology.tex
\section{Persistent Homology}
\label{sec:PH}

In this work, we will utilize the ideas of persistent homology, although it will be in an exceptionally simple case: namely 0-dimensional persistence for a point cloud in $\R$.
For this reason, we will focus only on this case and leave the interested reader to explore generalizations to higher dimensions \cite{Oudot2017,Dey2021,Munch2017}. 

The main idea of persistent homology is to encode the changing structure of a changing topological space.
In our case, the data will be a collection of points $P \subset \R$.
We can think of expanding intervals centered at each point $(a-\e/2,a+\e/2)$ for $a \in P$, and watching how the coverage of the real line changes as $\e$ is increased.
In particular, we are interested in the values of $\e$ when these intervals merge together to decrease the number of connected components.
The simplicity of $\R$ means that if we sort the values of $P = \{ a_1 < a_2 < \cdots < a_n\}$, connected components will merge at the values $a_{i+1} - a_{i}$.
Thus, we use the set $\dgm(P) = \{ a_{i+1} - a_{i} \mid i = 1,\cdots, n-1 \}$ to represent the changing connected components of the set of points.
We apologize to the informed reader since we are calling something a persistence diagram that does not quite fit with the literature. 
However, we feel we can be absolved as this is the collection of death times of the 0-dimensional persistence diagram in the true sense, and in this setting, all birth times are 0.

\begin{figure}[!t]
\centering
\includegraphics[width = .75\textwidth]{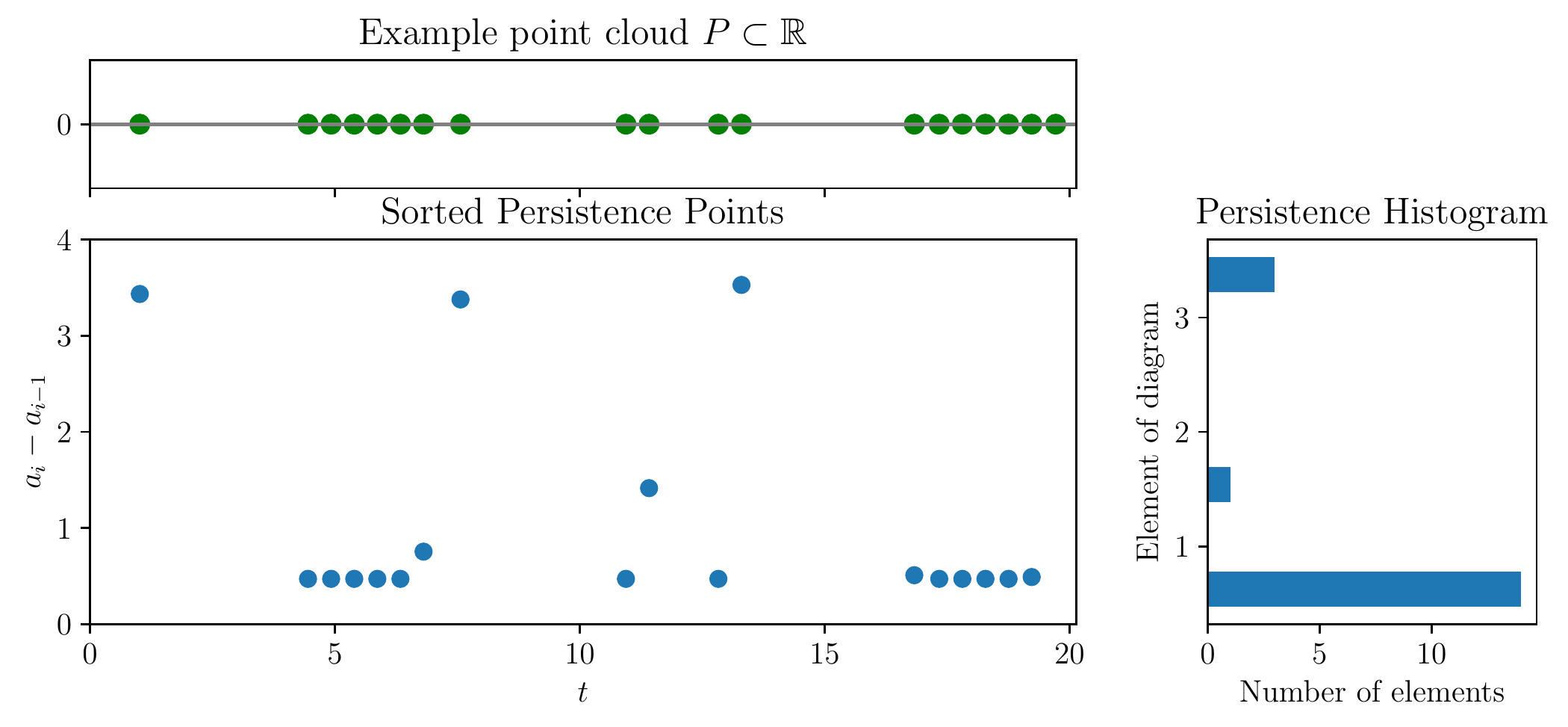}
\caption{An example point cloud in $\R$ is given at the top of the figure. The persistence diagram points $(a_i-a_{i-1})$ are given at the bottom left, drawn at the location of the $a_{i-1}$ value.  Note that high value points in this diagram occur at splits in the point cloud.
Then at right we show the histogram of (death-time) points in the persistence diagram.}
\label{fig:persistence1D}
\end{figure}
Consider the example of Fig.~\ref{fig:persistence1D}. 
An example point cloud $P = \{ a_1,\cdots,a_n\}$ is given in the top row. 
We can compute $\dgm(P)$, and visualize it one of two ways. 
First, because $\R$ has a natural total order, we can sort the points by the coordinate used to create it, thus giving the sorted points in the bottom left figure. 
If the sorted information is not necessary, we can also visualize this information in a persistence histogram, seen on the right panel.
We will often be interested in high persistence points; that is, the outlier values in $\dgm(P)$ above some threshold $\mu$. 
High persistence points in the sorted diagram correspond to locations in the point cloud where there is a wide space between adjacent points in the original point cloud.
We can use the histogram to determine $\mu$, looking for a large jump. \\
In the example of Fig.~\ref{fig:persistence1D},  a choice of threshold anywhere around $\mu = 2$ results in a labeling of three points as high persistence.
Three spaces between locations (generally, $k$ points in the diagram) in the point cloud corresponds to four (generally, $k+1$) clusters in the original data.
In the next section, we will use this idea of a persistence diagram on two 1-dimensional point clouds extracted from a signal to bound locations where we expect a zero crossing.

%% file: sections/sec-numerical-computations.tex
\section{Methods}
\label{sec:methods}
This section describes the computation method and provides an algorithm based on zero-dimensional persistent homology for detecting zero-crossings in a  noisy signal.
Assume we are given a discrete time series $x(t_0), \ x\left(t_{1}\right), \ \cdots,  x\left(t_{N}\right)$ with times $\left\{t_{0}<t_{1}<\cdots<t_{N}\right\}$.
We are interested in two point clouds in $\R$: the times when $x(t)$ is positive and negative respectively.
For technical reasons, we add the endpoints to each set.
Denote these
$P=\{t_i \mid f(t_i) > 0 \} \cup \{ t_0, t_N\} = \left\{p_{1}  < \cdots < p_{\ell}\right\}$
and
$Q=\{t_i \mid f(t_i) < 0 \} \cup \{ t_0, t_N\} = \left\{q_{1} < \cdots < q_{m}\right\}$.
We then compute the two diagrams as described in Sec.~\ref{sec:PH},
$\dgm(P)=\left\{p_{i}-p_{i-1} \mid i=1, \cdots, \ell-1\right\}$ and
$\dgm(Q)=\left\{q_{i}-q_{i-1} \mid i=1, \cdots, m-1\right\}$.
See Fig.~\ref{fig:pulse_with_pq} for a visualization of the notation.

\begin{figure}%
\centering
\includegraphics[width=3.0in]{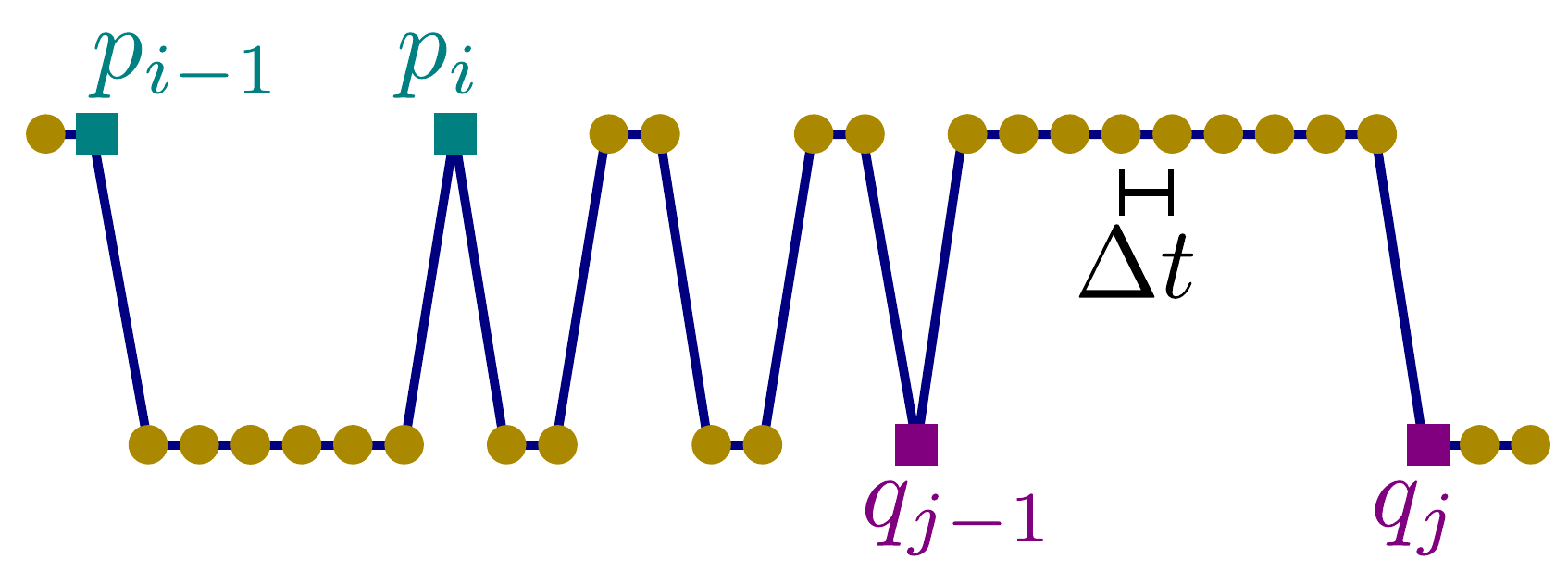}
\caption{A pulse signal showing the $p$ and $q$ notation. The $P$ points are those on the top row; the $Q$ on the bottom. For a threshold $\mu = 3\Delta t$, we see entries $(p_{i-1},p_i)$ from $\dgm P$ and $(q_{j-1},q_j)$ from $\dgm Q$.  
}
\label{fig:pulse_with_pq}
\end{figure}

Fix a threshold $\mu$.
Based on this threshold, we can extract the values in $\dgm(P)$ and $\dgm(Q)$ which are above $\mu$.
Namely, let $\dgm(P)_\mu = \{ a \in \dgm(P) \geq \mu\}$ and define $\dgm(Q)_\mu$ similarly.
With this setup, we are interested in the region of $\R$ remaining after we remove the intervals
\begin{equation*}
  \calI_P ={}
  \{ (p_{i-1},p_i) \mid  p_{i}-p_{i-1} \in \dgm(P)_\mu\}
\end{equation*}
and
\begin{equation*}
    \calI_Q =
  \{ (q_{i-1},q_i) \mid  q_{i}-q_{i-1} \in \dgm(Q)_\mu\}.
\end{equation*}
In particular, after sorting and interleaving these intervals, the remaining portion of the line can be written as
\begin{equation*}
  \R \setminus \bigcup_{I \in \calI_P \cup \calI_Q} I = \bigcup_{j} [r_{2j-1},r_{2j}]
\end{equation*}
where each $r_{2j-1}$ (resp.~$r_{2j}$) is the right (resp.~left) endpoint of an interval in either $\calI_P$ or $\calI_Q$.
We can further assume that these endpoints are sorted, so denote the endpoint set as $R = \{ r_1 < \cdots < r_{2i}<r_{2i+1}<\cdots < r_{2K} \}$.

The regions in $\calI_P$ and $\calI_Q$ are the portions of $\R$ where we are highly certain the function stays with the same sign, thus the remaining portion of the line is the locations where we expect to find our zero-crossings. 
For this reason, our algorithm returns the intervals $[r_{2i-1},r_{2i}]$ as the potential locations for crossings. 
Further, if the two $r$ points come one from $P$ and one from $Q$, we are quite certain there is a zero crossing between thanks to the intermediate value theorem. 

\begin{proposition}
  If the values are sampled from a continuous function $f$, then there is at least one zero-crossing in any interval with different endpoint types.
\end{proposition}

\begin{proof}
If the interval in question is $[a,b]$, without loss of generality we can assume that $a = f(t_a)$ is the higher value of a $P$-interval, and $b = f(t_b)$ is the lower bound of a $Q$-interval.
By definition of the $P$ and $Q$ points, this means that $a>0$ and $b<0$.
So, by the IVT, there is a $t \in [t_a,t_b]$ for which $f(a) = 0$.
\end{proof}

Less can be said in the case where the endpoints are of the same type.
In fact, the likely scenarios in simple cases are 0 or 2 zero-crossings, but any even number is possible.
For this reason, our algorithm (\Cref{alg:zero-cross}) returns a list of intervals potentially containing zero crossings, with an uncertainty flag for intervals either bounded on either side by the same type of endpoint, or coming from the boundary.

Once these intervals are found, one can choose any number of methods for coming up with a potential value for the crossing, such as bisecting each of the intervals or linearly interpolating them to find the root estimate. 
For our computations, we have used the average of each returned interval's bounds to estimate the root location.  
The full approach is summarized in \Cref{alg:zero-cross}.
\begin{algorithm}
 \caption{ Persistence algorithm for bounding zero-crossings.}
\label{alg:zero-cross}
\textbf{Data: } A time series $x\left(t_{1}\right), \ \cdots,  x\left(t_{N}\right)$ with times $\left\{t_{1}<\cdots<t_{N}\right\}$ and Persistence threshold $\mu \geq 0.$\\
\textbf{Result:} Intervals containing potential zero-crossings of the signal $x(t)$.
\\
Calculate
  $P=\left\{p_{i} \mid \text{sgn}(x\left(p_{i}\right)) >0\right\} \cup \{t_0, t_N \} $
and $Q=\left\{q_{i} \mid \text{sgn}(x\left(q_{i}\right)) < 0 \right\}\cup \{t_0, t_N \}$.\\
Sort the lists, denote as
    $P=\left\{p_{1}  < \cdots < p_{\ell}\right\}$
    and
    $Q=\left\{q_{1} < \cdots < q_{m}\right\}$. \\
Keep the pairs which correspond to persistence above $\mu$ in each diagram\\
  \qquad $\mathcal{I}_P=\left\{(p_i,p_{i+1}) \mid p_{i+1}-p_{i}\geq \mu,\,  i=1, \cdots, \ell-1\right\}$ and\\
  \qquad $\mathcal{I}_Q=\left\{(q_i,q_{i+1}) \mid q_{i+1}-q_{i}  \geq \mu,\, i=1, \cdots, m-1\right\}$.\\
Interleave and sort the end points of the intervals \\
\qquad $R = \{ r_0 < \cdots < r_{2i}<r_{2i+1}<\cdots < r_{2K-1} \} $\\
Augment $R$ with the endpoint values if they are not already included\\ %
\qquad $R = \{ t_0 = r_{-1} \leq r_0 < r_1 < \cdots < r_{2i}<r_{2i+1}<\cdots < r_{2K-1} \leq r_{2K} = t_N \} $\\ 
\textbf{return} Intervals 
$$\{ [r_{2i-1},r_{2i}] \mid r_{2i-1}\neq r_{2i},\;  i = 0, \cdots, K \}$$
 labeled with an uncertainty flag if the endpoints did not come from different types of intervals.
\end{algorithm}
\subsection{Example}
\begin{figure}[!t]
\centering
\includegraphics[width = .75\textwidth]{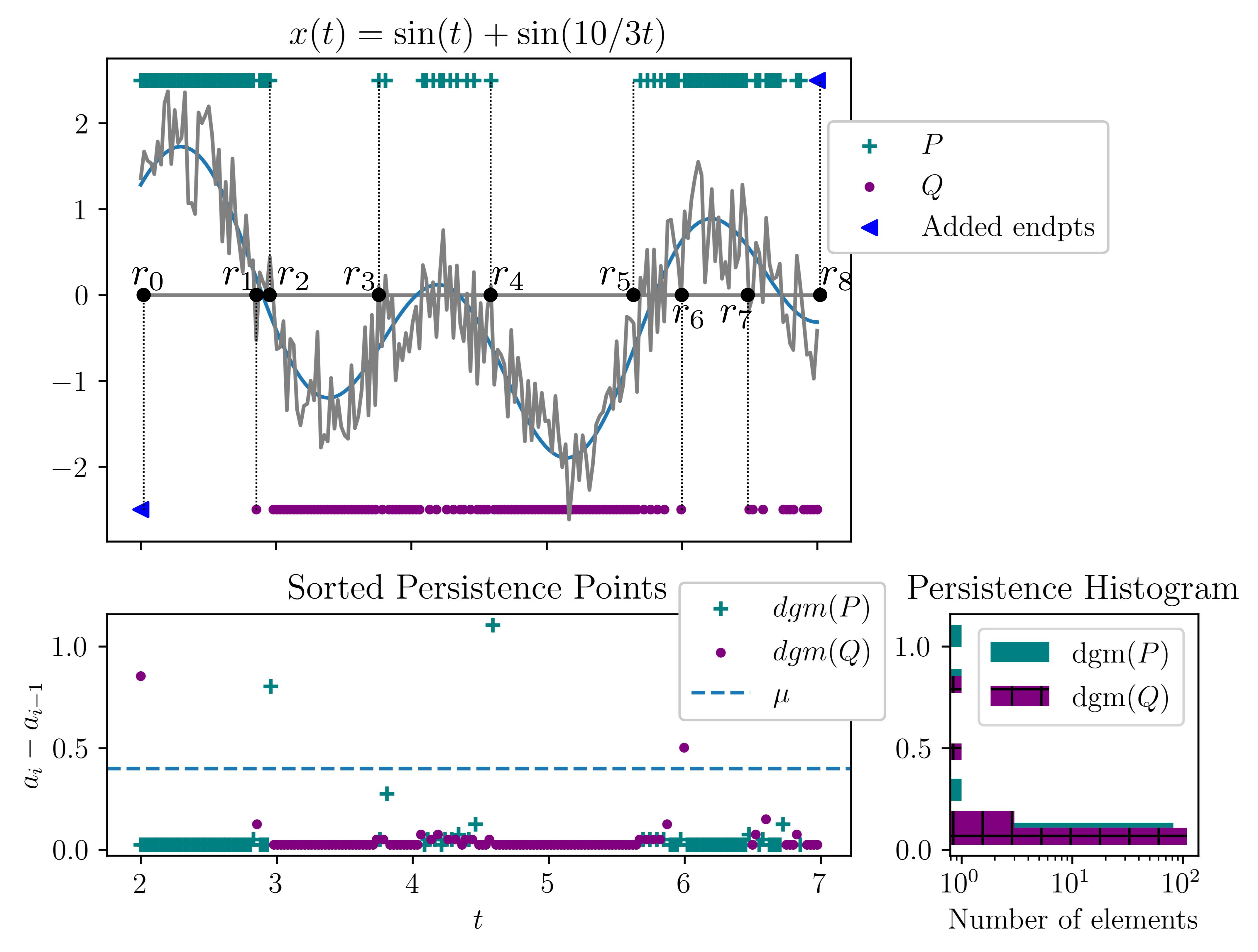}
\caption{An example using $x(t)=\sin (t)+\sin (10 / 3 t)$. The intervals returned by the algorithm are $[r_1,r_2]$ and $[r_5,r_6]$ (without a uncertainty flag); and  $[r_3,r_4]$ and $[r_7,r_8]$ (with an uncertainty flag). On the right, the persistence diagram is shown. The threshold was chosen to be $\mu = 0.4$, but different choices of $\mu$ will result in different zero-crossing intervals.}
  \label{fig:BasicSineExample}
\end{figure}
We demonstrate the algorithm on a multi frequency sine wave given by
\begin{equation}
\centering
\label{eq:demo}
x(t)=\sin (t)+\sin (10t / 3).
\end{equation}
shown in Fig.~\ref{fig:BasicSineExample}.
We include additive white noise in Eq.~\eqref{eq:demo}, and convert the resulting signal to a pulse wave. 
The locations of the $t$ values in $P$ and $Q$ are shown above and below the original signal. 
Note that the blue triangles at $t_{0}$ and $t_{N}$ are added to $P$ and $Q$ after the fact to deal with the boundary.

Below the signal, we have the full sorted persistence diagram for the example.
For $\mu = 0.4$ as represented by the dashed line, the resulting intervals are
\begin{align*}
  \dgm_\mu(P) =  \{(r_2,r_3), (r_4,r_5) \} \\
  \dgm_\mu(Q) =  \{(r_0,r_1), (r_6,r_7) \}
\end{align*}
and we can augment $R$ to include the endpoints $R = \{r_{-1} =r_0 < \cdots < r_7 \leq r_8\} $.
The intervals returned by the algorithm without an uncertainty flag, i.e., intervals whose end points are not both in $P$ or $Q$, are $[r_1,r_2]$ and $[r_5,r_6]$.
The remaining intervals returned with the uncertainty are $[r_3,r_4]$ and $[r_7,r_8]$.

Depending on the choice of threshold $\mu$, the algorithm will return different collections of intervals.
If we choose a slightly higher threshold of $\mu = 0.6$, the points $r_6$ and $r_7$ would not be included as potential bounds for zero crossings and so the interval $[r_5,r_8]$ would be provided instead with an uncertainty flag. 

\subsection[Choosing the threshold mu]{Choosing the threshold {$\mu$}}
\label{sec:PT}
\Cref{alg:zero-cross} utilizes the idea of the points bracketing a root having a higher persistence, a visual quantification of which is possible by plotting both $\dgm(P)$ and $\dgm(Q)$ together.
The algorithm presents the use of a number $\mu$ for setting this persistence threshold, but as seen in the examples, the resulting brackets are quite dependent on the choice of $\mu$.

If we are in particularly restrictive settings, there are some cases where determining outliers are particularly simple. 
The first would be a signal with no noise, for which the time difference between two uniformly sampled points, $1\cdot dt$, is an error-free persistence threshold. 
The second case pertains to the availability of information on the number of roots in the interval. 
For $n$ roots in the interval, $n+1$ highest points from the persistence diagram represent the required brackets.

However, when we do not have such nice input data, we must find other methods of setting $\mu$ mathematically. 
Owing to the low frequency of appearance and quantifiable difference in high persistence points in these settings, they may be treated as outliers in the data-set. 
Consequently, statistical and machine learning outlier detection methods can be applied to identify them. 
Here, we give two options for detecting the high persistence points which will be used in our experiments in the next section. 
\begin{figure}%
\centering
\includegraphics[width=.25\textwidth]{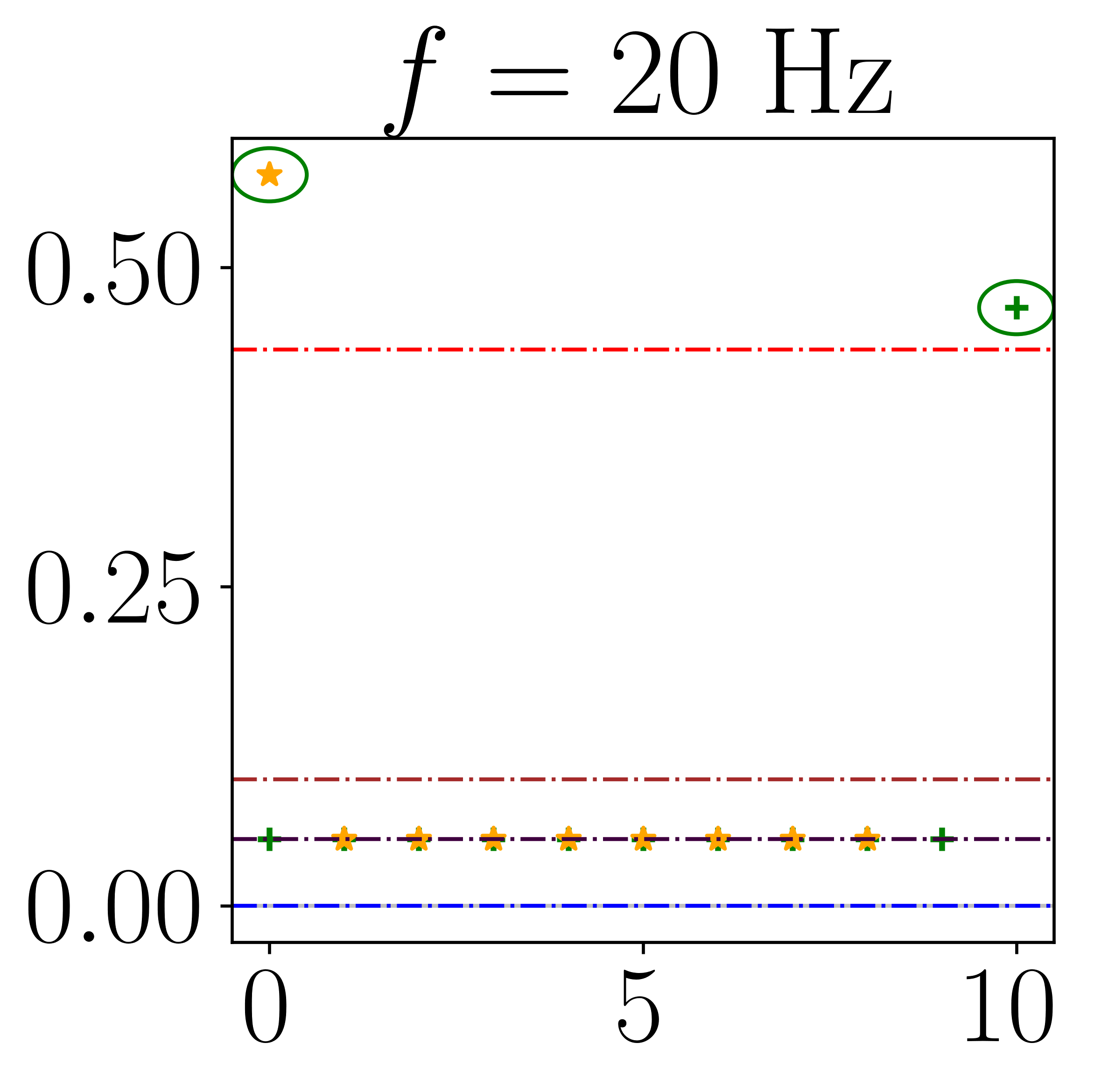}
\hfil
\includegraphics[width=.25\textwidth]{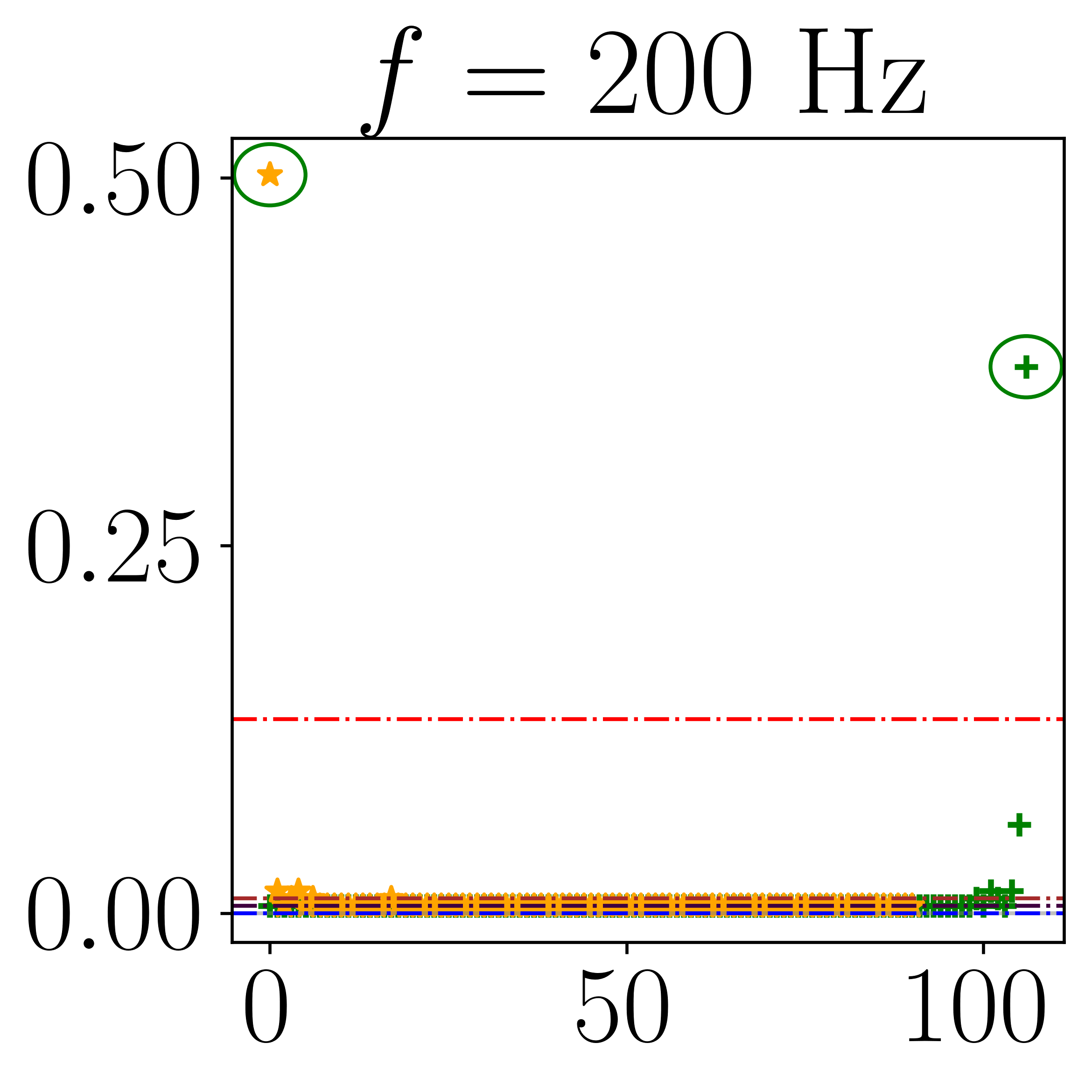}
\hfil
\includegraphics[width=.25\textwidth]{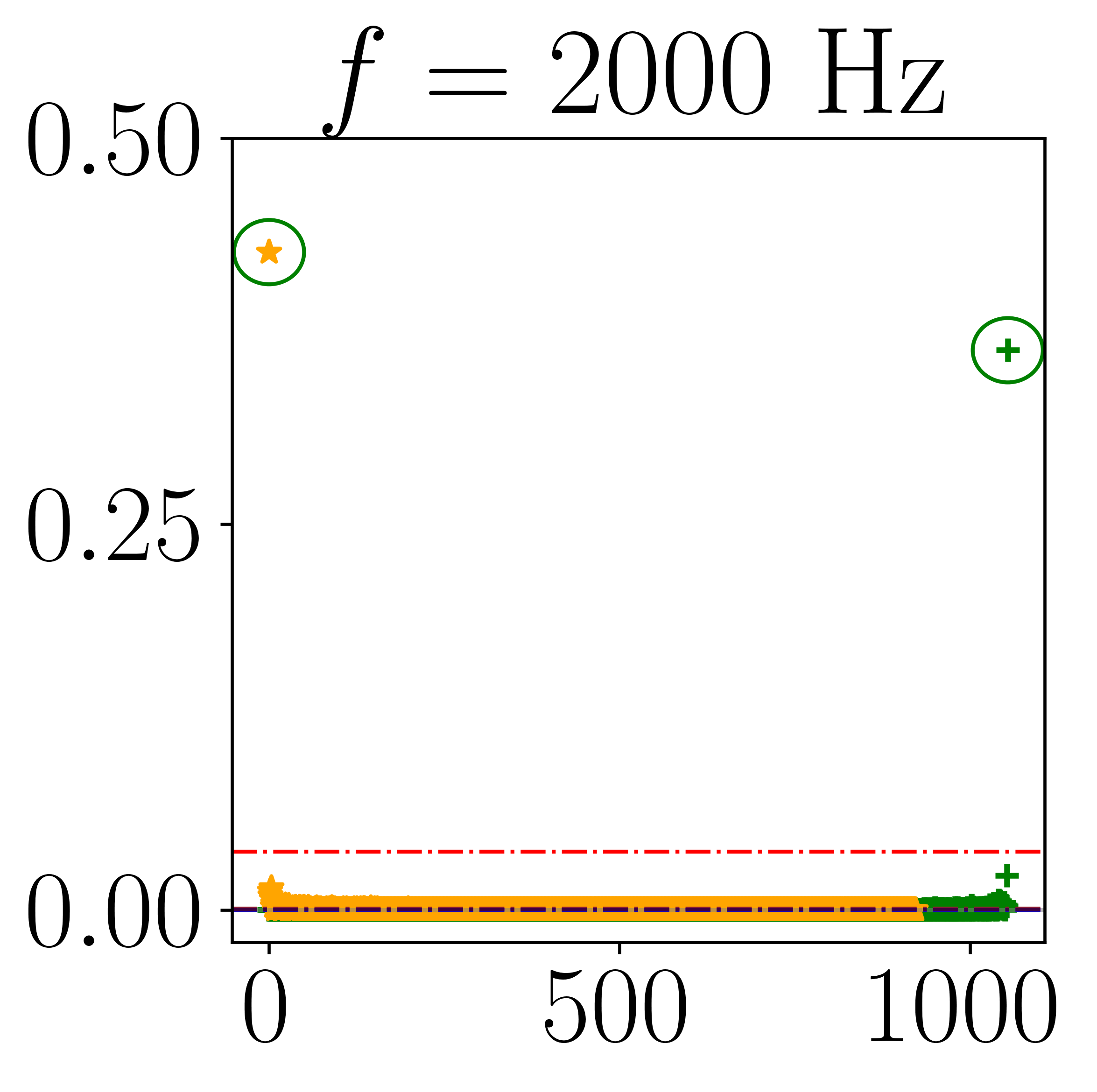}
\hfil
\raisebox{0.4em}{\includegraphics[width=.2\textwidth]{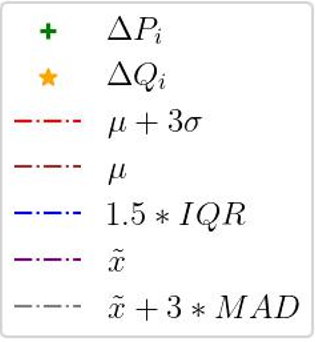}}
\caption{Sorted persistence points for $x_{9}$ (using SNR $=15$ dB) with statistical measures for outlier detection plotted.}
\label{fig:statistics}
\end{figure}
The simplest method for detecting outliers is with statistical measures. 
We show this method on the example of diagrams computed from signal $x_{9}$ in Table~\ref{tab:func}.
Fig.~\ref{fig:statistics} shows 
mean $\mu$, 
z-score of $3$ ($\mu+3\sigma$) \cite{Saleem2021, Mowbray2018}, 
$1.5$ times Interquartile Range (IQR) \cite{Saleem2021, Mowbray2018}, 
median $\tilde{x}$ 
and the Median Absolute Deviation (MAD) \cite{Saleem2021} - all of which have been used as cut-offs for finding outliers in a data-set.
We can see from this example that  $\mu + 3\sigma$ is a  reliable estimate of the persistence threshold.

Finally, we explore the application of Machine Learning algorithm Isolation Forest \cite{Pedregosa2011} on sorted persistence points of $x_{4}$ of Table \ref{tab:func}. Encircled are the points returned by the technique while the points required for true roots are the first six from the top. Fig.~\ref{fig:ml_iforest} shows Isolation Forest to be suitable for low sampling frequencies at all SNRs, but requires a cleaner signal (higher SNR) as sampling frequency is increased. For example, the figure shows the technique failing for frequency of 5000 Hz at 40 dB noise, but manages to find the desired points for 55 dB at the same frequency.\\
\begin{figure}%
\centering
\includegraphics[width=.235\textwidth]{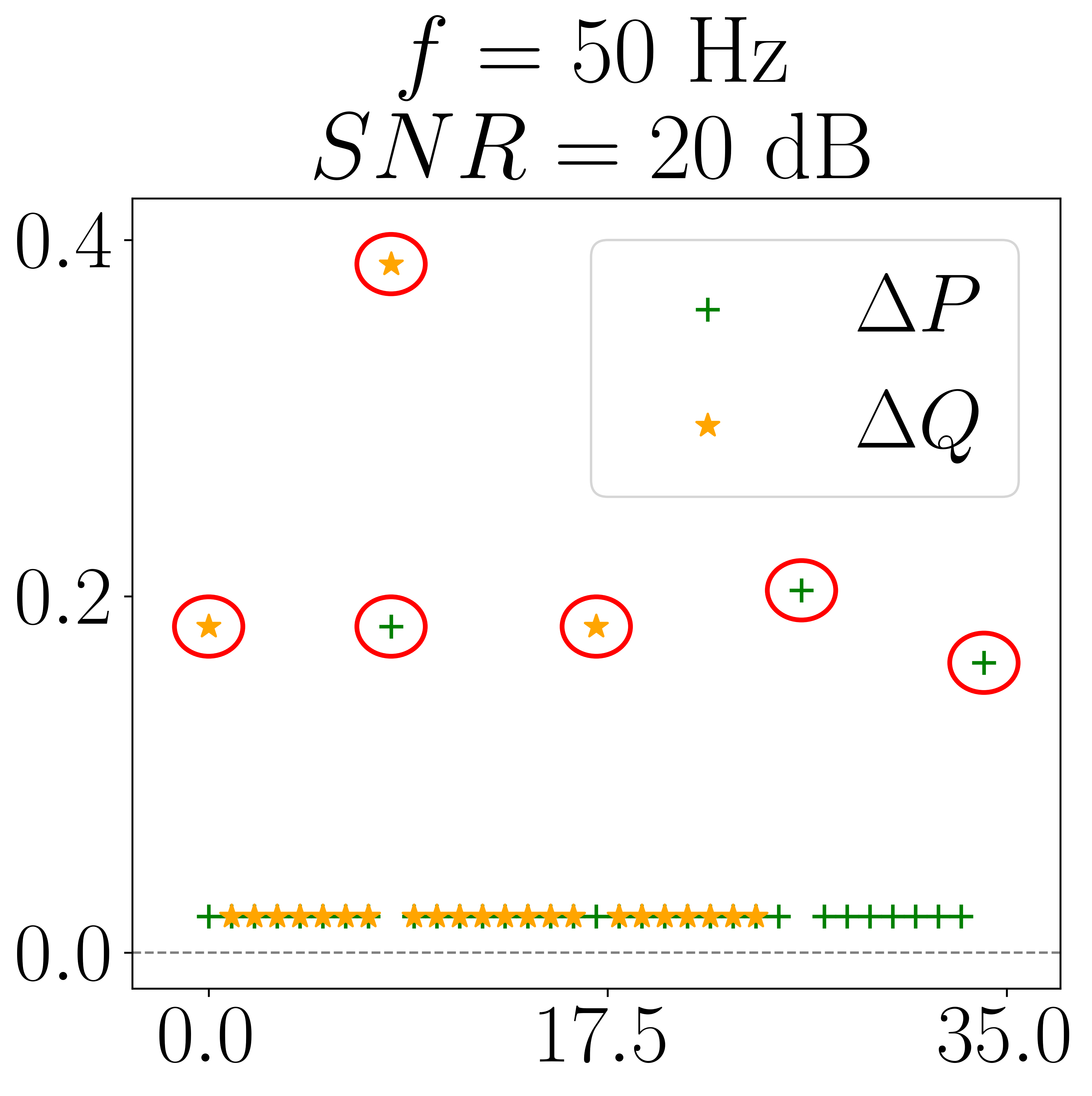}
\includegraphics[width=.235\textwidth]{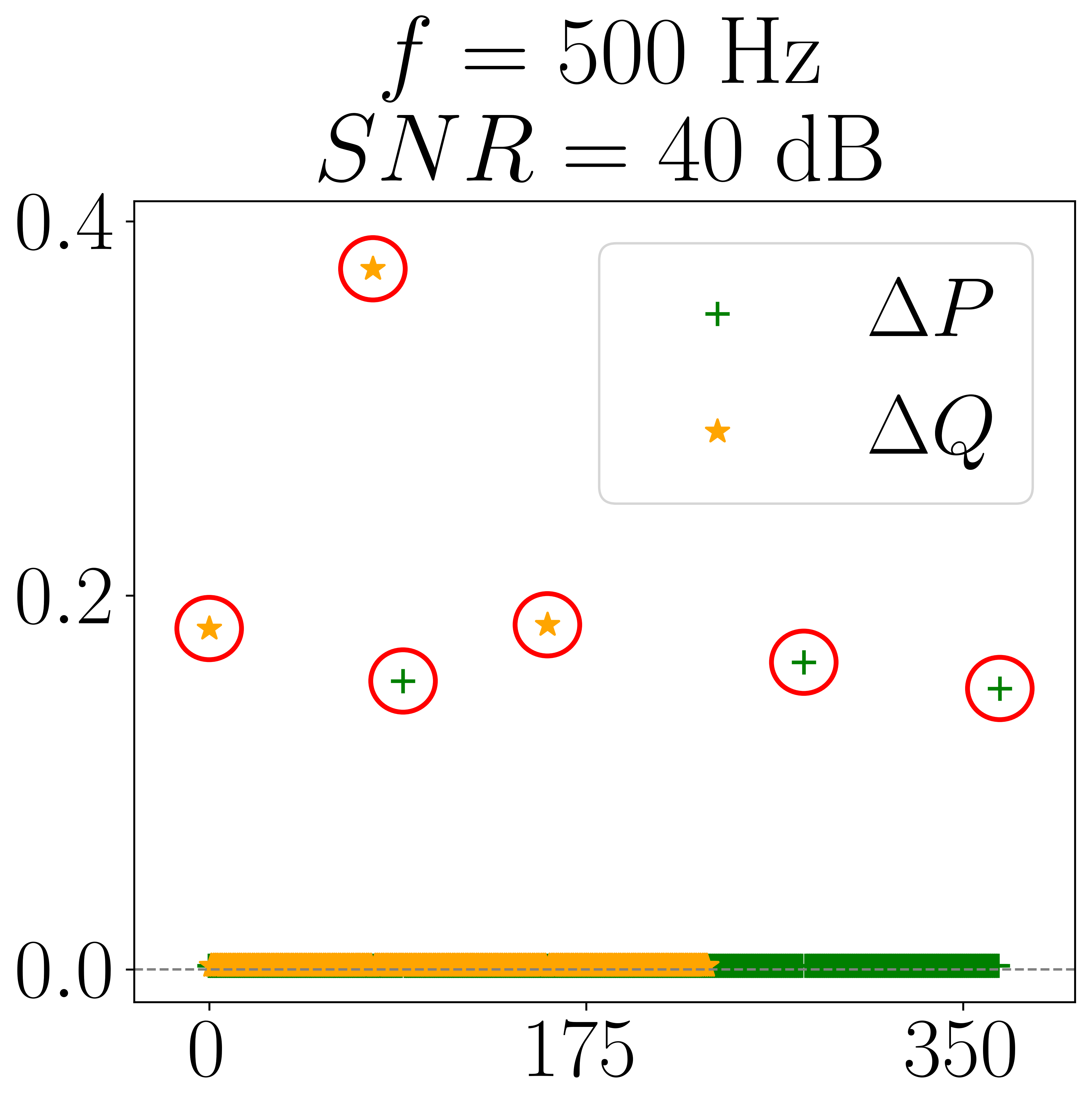}
\includegraphics[width=.235\textwidth]{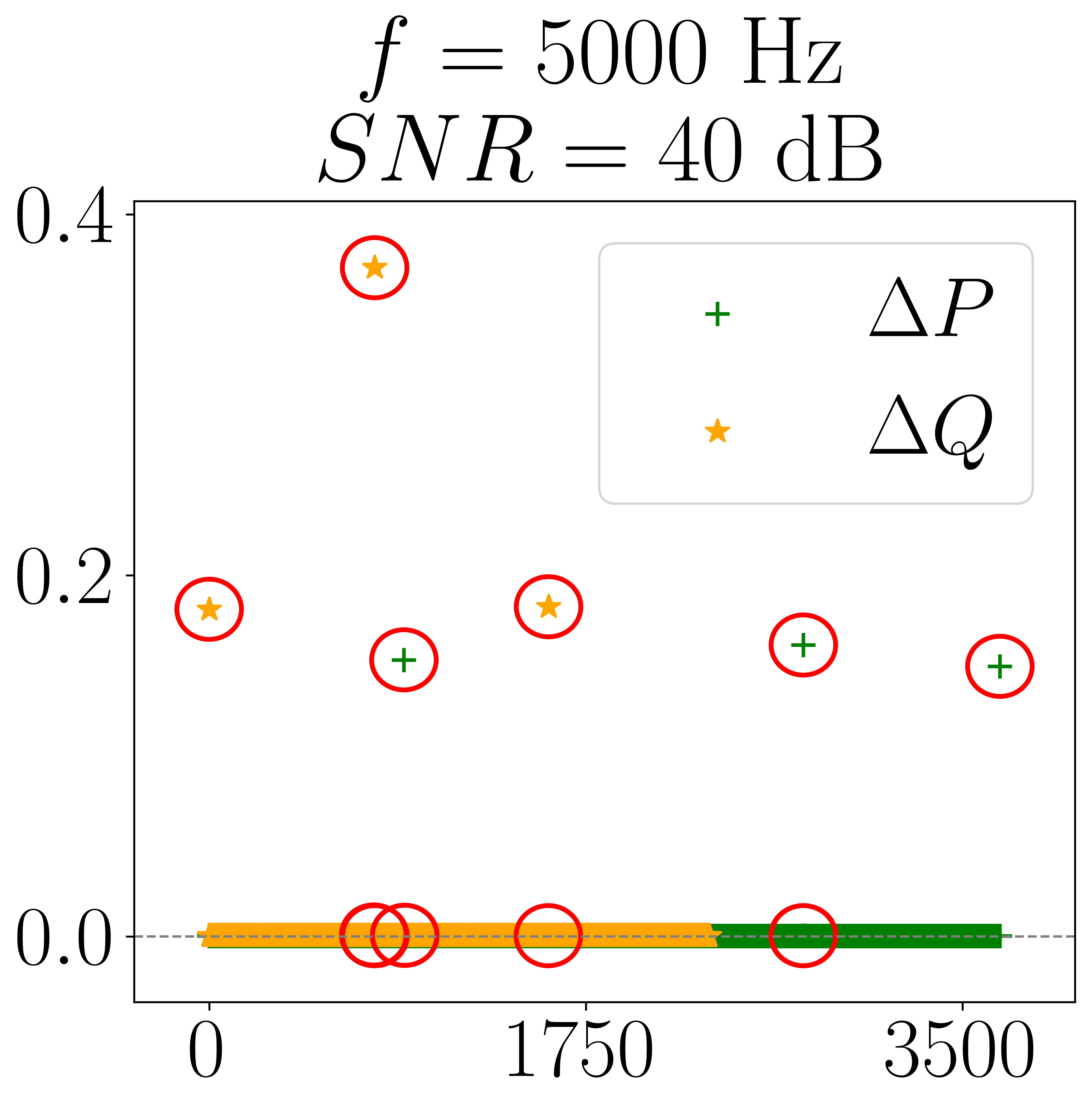}
\includegraphics[width=.235\textwidth]{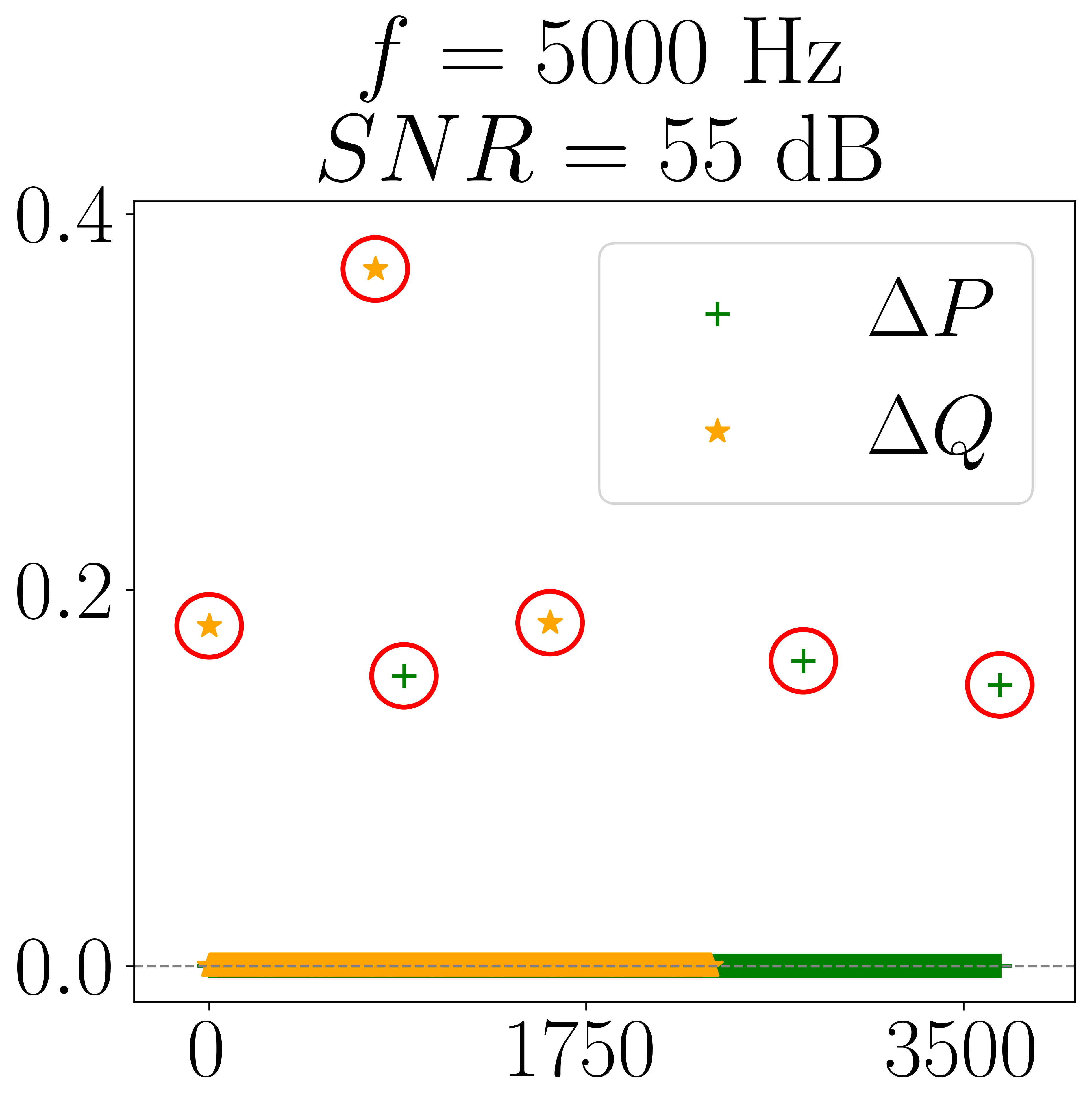}
\caption{Application of Isolation Forest on sorted persistence points of $x_{4}$.}
\label{fig:ml_iforest}
\end{figure}
In Fig.~\ref{fig:statistics}, we notice that with decreasing sampling frequency, the 3$\sigma$ threshold approaches one of the required outliers and would surpass it if the frequency were lowered further. This indicates that the method works better if sampling frequencies are high than low. On the other hand, we discussed the need of an increasingly clean signal for Isolation Forest as we raise our sampling frequency. Therefore, as a rule of thumb when working with noisy signals, the former method should be preferred for frequently-sampled signals, while the latter for signals with low sampling rates.

%% file: sections/sec-results.tex
\section{Results and Discussion}
\label{sec:R&D}
\begin{table}[!t]
\centering
\begin{tabular}{lllll}
\toprule
  & $f(t)$ & Zero crossings& Interval \\
\toprule
1 & $x_1(t)=(1 / 6) t^{6}-(52 / 25) t^{5}+(39 / 80) t^{4}$ & 4.052 & {[-1.5, 5]} & \\
& $+(71 / 10) t^{3}-(79 / 20) t^{2} $& & &\\
&$-t+1 / 10+1000$ & & &\\
2 & $x_2(t)=\sin (t)+\sin (10 / 3 t)$ & 2.9, 4.039, 4.3499, 5.79986,  & {[2.7, 7.5]} & \\
& & 6.73198, 7.2598 \\
3 & $x_3(t)=\left(-16 t^{2}+24 t-5\right) \mathrm{e}^{-t}+3$ & 2.064 & {[1.9, 3.9]} &\\
4 & $x_4(t)=(-3 t+1.4) \sin (18 t)+0.1$ & 0.181, 0.3349, 0.7059, 0.868,  & {[0, 1.2]} &\\
& & 1.0504 \\
5 & $x_5(t)=\sin (t)+\sin (2 / 3 t)$ & 3.77, 7.54, 9.425 & {[3.1, 11]} &\\
6 & $x_6(t)=-t \cdot \sin (t)+0.5$ & 0.741, 2.973, 6.362 & {[0, 8]}&\\
7 & $x_7(t)=-(2 \cos (t)+\cos (2 t))$ & -1.196, 1.196, 5.087 & {[-1.57, 6.28]} &\\
8 & $x_8(t)=\sin ^{3}(t)+\cos ^{3}(t)$ & 2.356, 5.5 & {[0, 6.28]} &\\
9 & $x_{9}(t)=-t^{3}+\left(t^{2}-1\right)^{6}$ & 0.525 & {[0.001, 0.99]} &\\
10 & $x_{10}(t)=-\mathrm{e}^{-t} \sin (2 \pi t)+0.5$ & 0.092, 0.371 & {[0, 4]} \\
11 & $x_{11}(t)=(t^{2}-5t+6)/(t^{2}+1)$ & 2, 3 & {[-5, 3]} \\
12 & $
x_{12}(t) =
    \begin{cases}
        (t-2)^{2} & \text{if } t \leq 3\\
        2\ln{(t-2)} + 1 & \text{if } \text{otherwise}
    \end{cases}
$
& 2 & {[0, 6]} \\
13 & $x_{13}(t)=-t+ \sin (3 t)+1$ & $1.035$ & {[0, 6.5]} &\\
14 & $x_{14}(t)=-(t - \sin{t})*e^{-t^{2}} + 0.01$ & 0.4159 & {[-2, 2]} &\\
\bottomrule
\end{tabular}
\caption{Examples used to evaluate the robustness of the proposed approach \cite{Molinaro2001}.}
\label{tab:func}
\end{table}

A series of experiments were conducted using the algorithm for finding the roots of the functions in Table~\ref{tab:func} over a range of sampling frequencies and signal-to-noise ratios (SNR), where SNR in decibels is defined in Eq.~\eqref{eq:snr}.
\begin{equation}
\centering
\label{eq:snr}
\text{SNR} = 10\cdot \log_{10}\frac{P_{s}}{P_{n}}
\end{equation}
where $P_{s}\text{ and } P_{n}$ represent the power of the clean signal and the noise respectively.
\begin{figure*}[!t]
\centering
\subfloat[]{\includegraphics[width=.4\linewidth]{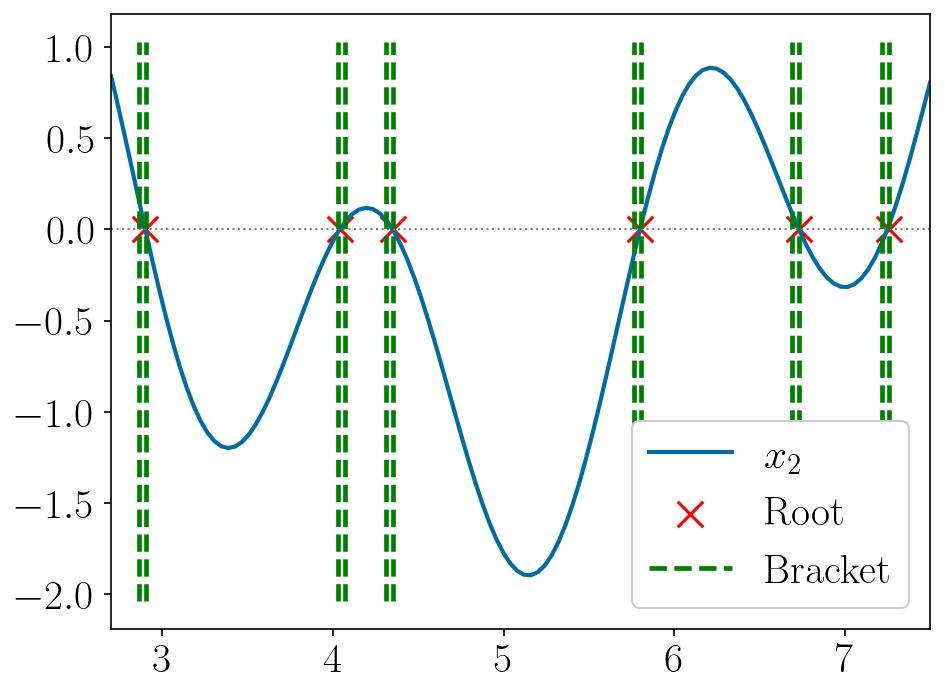}\label{fig:x8-roots}}
\hfil
\subfloat[]{\includegraphics[width=.4\linewidth]{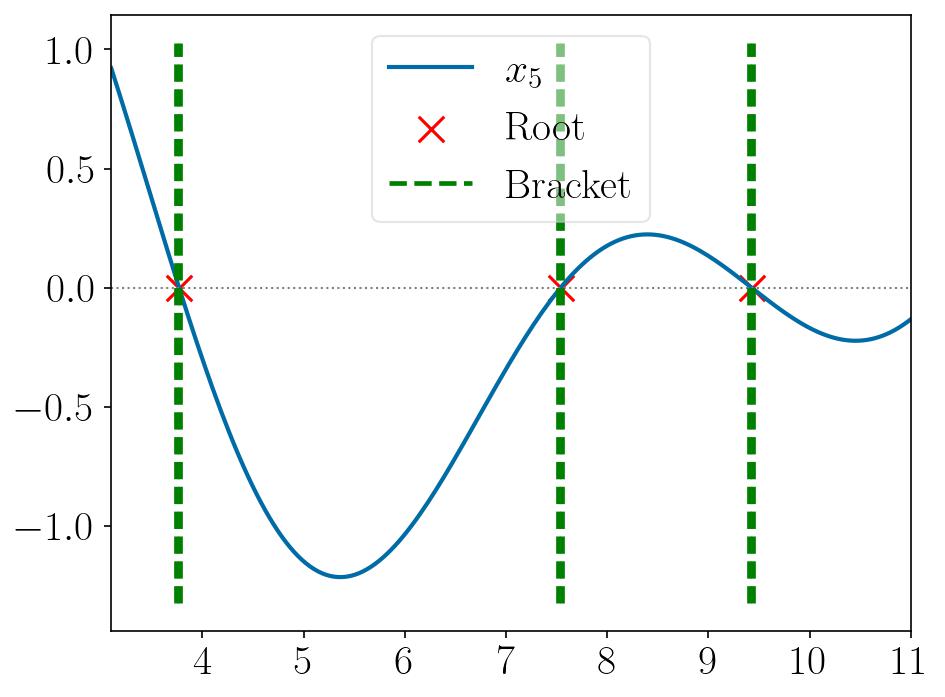}\label{fig:x14-roots}}
\hfil
\subfloat[]{\includegraphics[width=.4\linewidth]{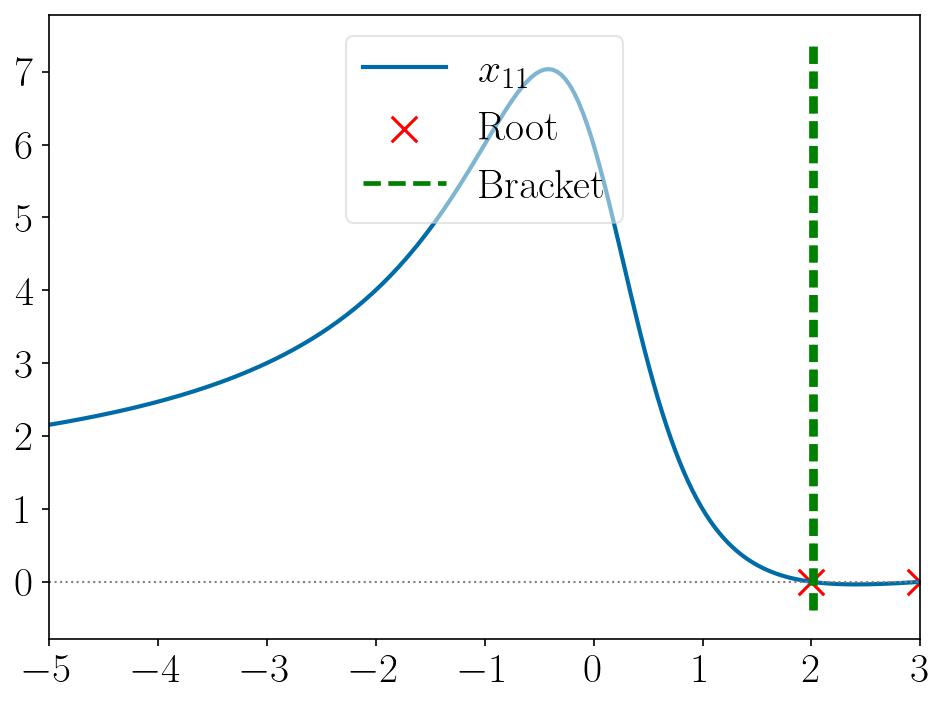}\label{fig:x11-roots}}
\hfil
\subfloat[]{\includegraphics[width=.4\linewidth]{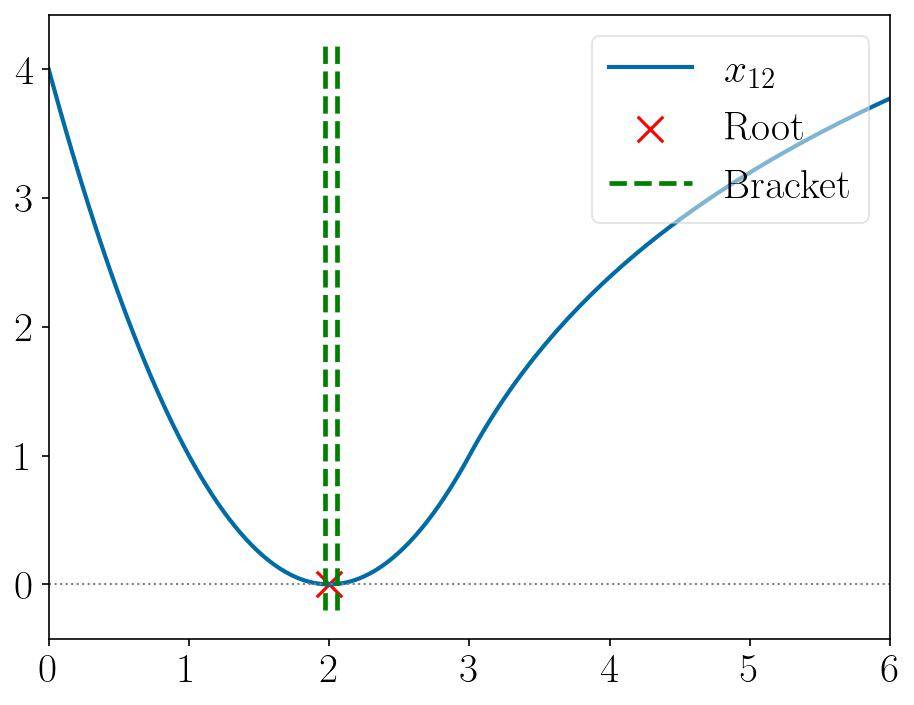}\label{fig:x12-roots}}
\caption{True roots versus brackets returned at a sampling frequency of 25 Hz for (a) $x_{2}$, (b) $x_{5}$, (c) $x_{11}$ and (d) $x_{12}$}
\label{fig:brackets}
\end{figure*}
Figure~\ref{fig:brackets} shows the brackets returned by the algorithm and the true roots of the noise-free functions $x_{2}$ and $x_{5}$ in Table \ref{tab:func} at a low sampling frequency of $25$ Hz. The figures demonstrate that the algorithm has the ability to capture multiple roots in the given interval. 
However, functions $x_{11}$ and $x_{12}$ illustrate particular shortcomings inherent in the method, even with no noise. 
First, if there is a root on the boundary, it is particularly difficult for the method to find it, such as in the example of $x_{11}$. Another issue is that in the case of noise free data with a zero of positive second derivative, it is possible for the zero to not be seen at all. In  $x_{12}$, no zero would be detected if the zero itself is not actually sampled. Interestingly in this case, we would do a better job detecting the crossing if there was noise in the system. The rest of the examples can be found in Appendix \ref{appendix:AppB}.\\
Figure~\ref{fig:SNR} shows the effect of adding gaussian noise of varying SNR in the signal on the brackets returned by the algorithm for the functions $x_{8}$ and $x_{14}$ in Table \ref{tab:func} at a sampling frequency of $1000$ Hz. The figures illustrate that the algorithm can efficiently bracket all roots regardless of SNR. However, in cases where the SNR is low enough to cause the emergence of artificial crossings in the signal, the algorithm may return more roots (for example, $x_{14}$ at SNR of $15$ dB) or less roots (for example, $x_{2}$ at SNR of $15$ dB, shown in Fig.~\ref{fig:SNR-shortcoming}). The rest of the cases can be found in Appendix \ref{appendix:AppC}.\\
\begin{figure*}[!t]
\centering
\subfloat[]{\includegraphics[width=2.75in]{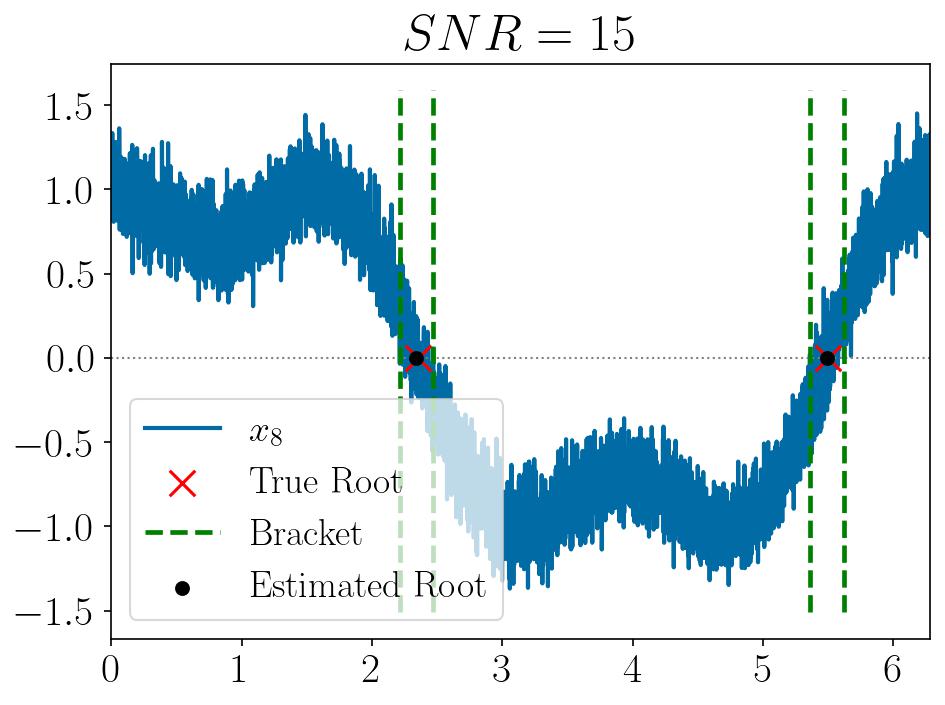}\label{fig:x8-15}}
\hfil
\subfloat[]{\includegraphics[width=2.75in]{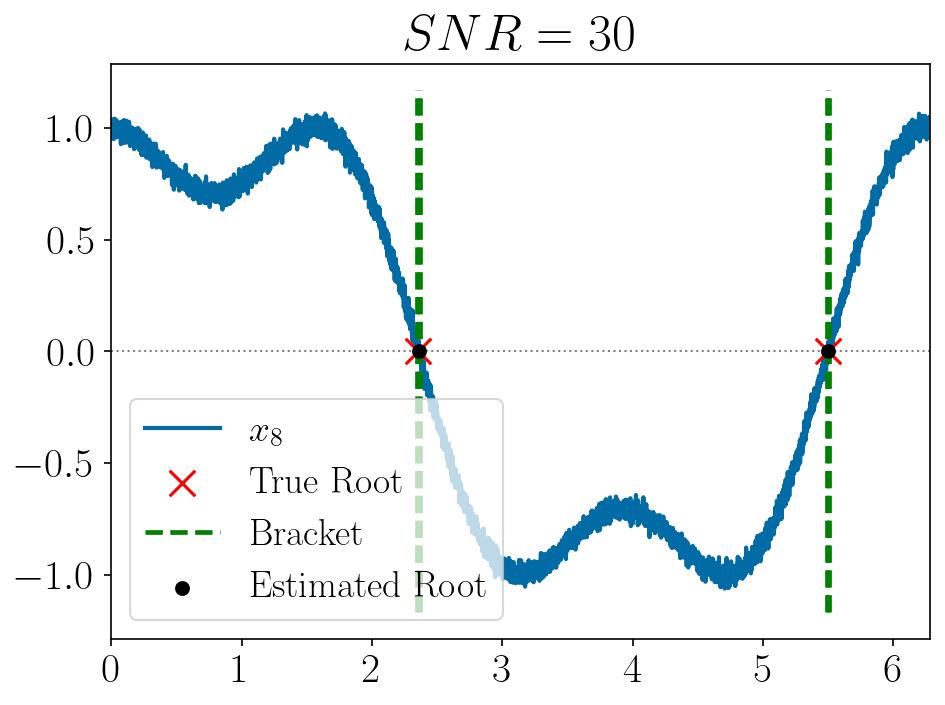}\label{fig:x8-30}}
\\
\subfloat[]{\includegraphics[width=2.75in]{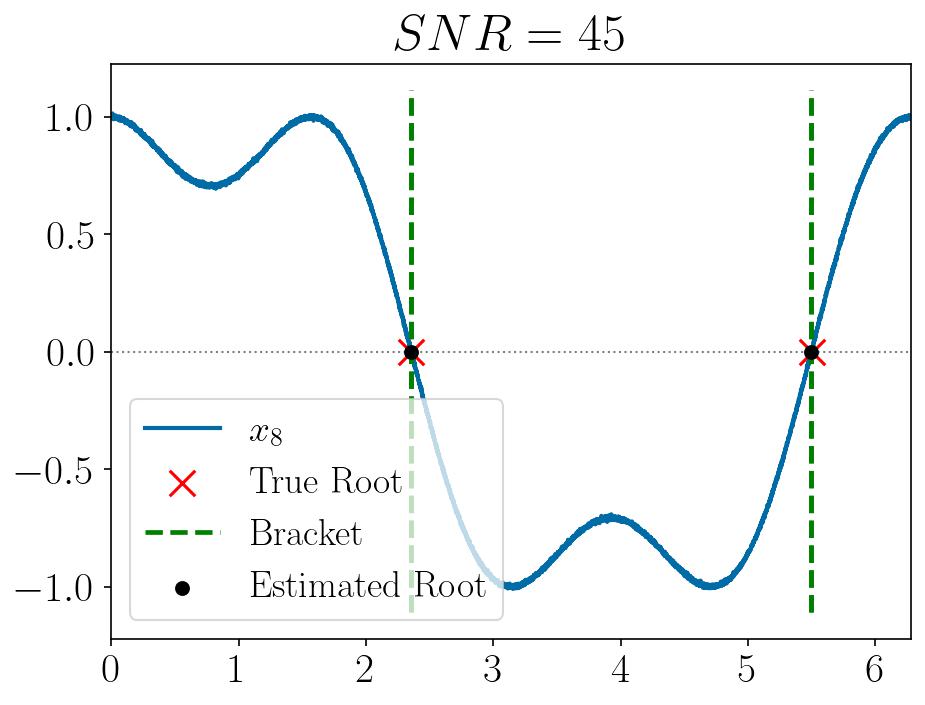}\label{fig:x8-45}}
\hfil
\subfloat[]{\includegraphics[width=2.75in]{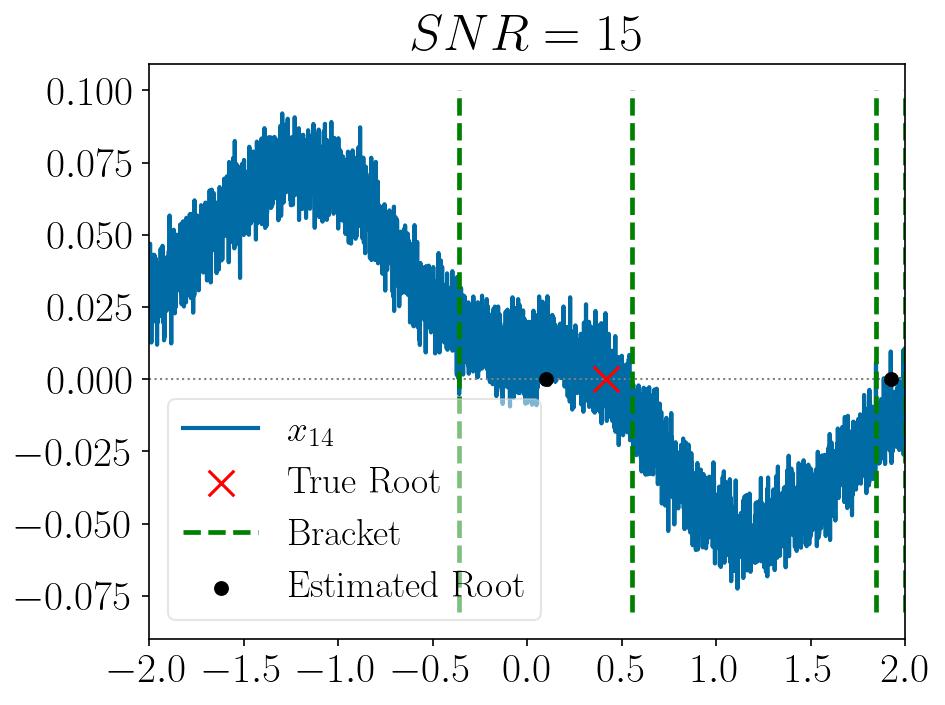}\label{fig:x14-15}}
\\
\subfloat[]{\includegraphics[width=2.75in]{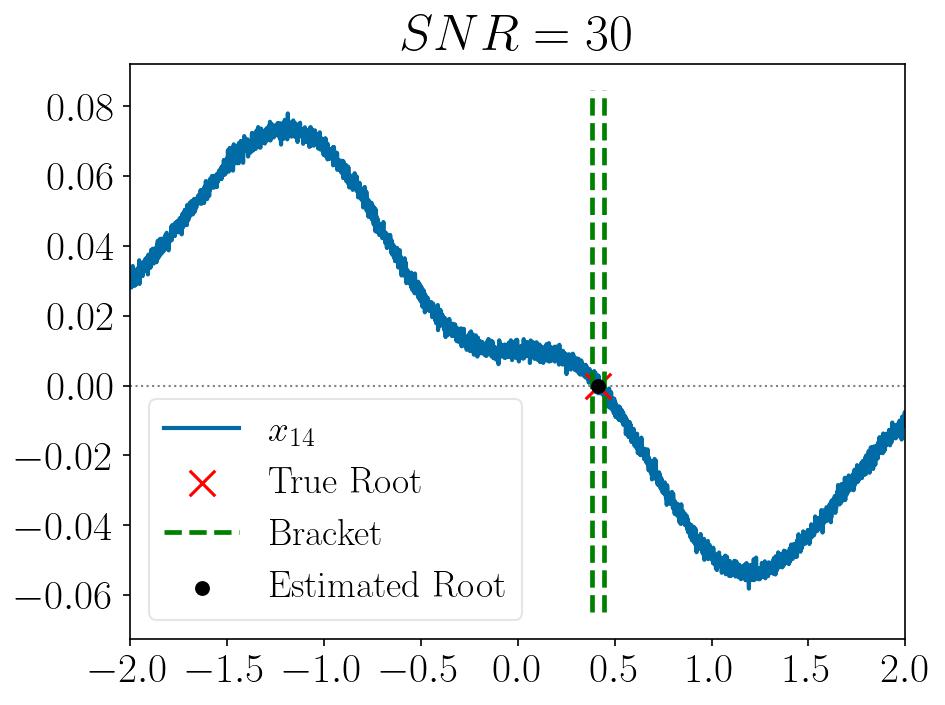}\label{fig:x14-30}}
\hfil
\subfloat[]{\includegraphics[width=2.75in]{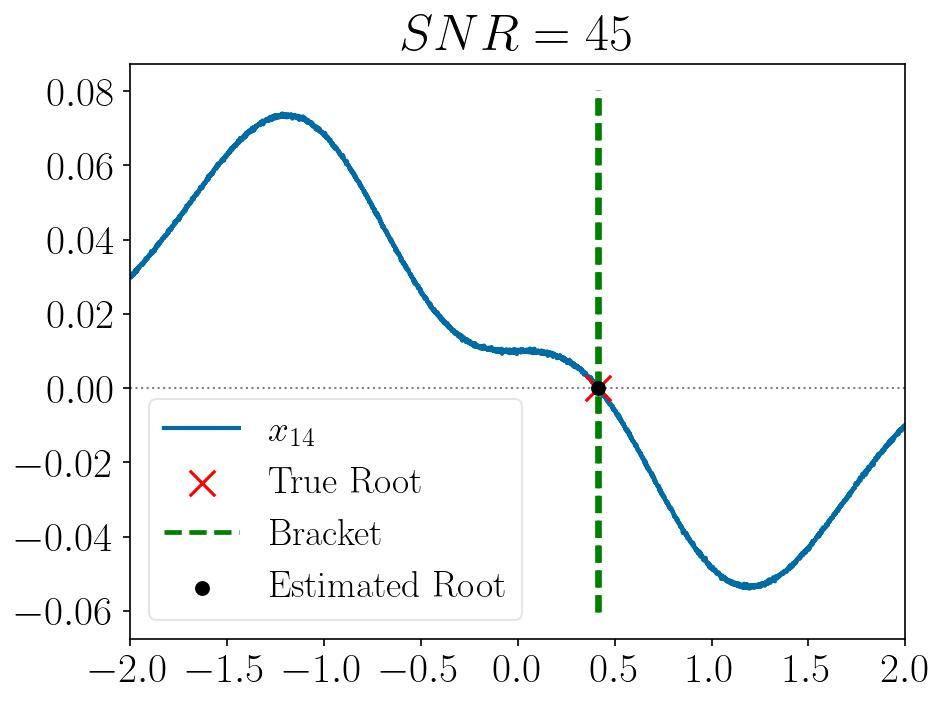}\label{fig:x14-45}}
\caption{True roots versus estimated roots returned at a sampling frequency of 1000 Hz for $x_{8}$ and $x_{14}$ at low, medium and high SNR}
\label{fig:SNR}
\end{figure*}
\begin{figure*}[!t]
\centering
\subfloat[]{\includegraphics[width=2.75in]{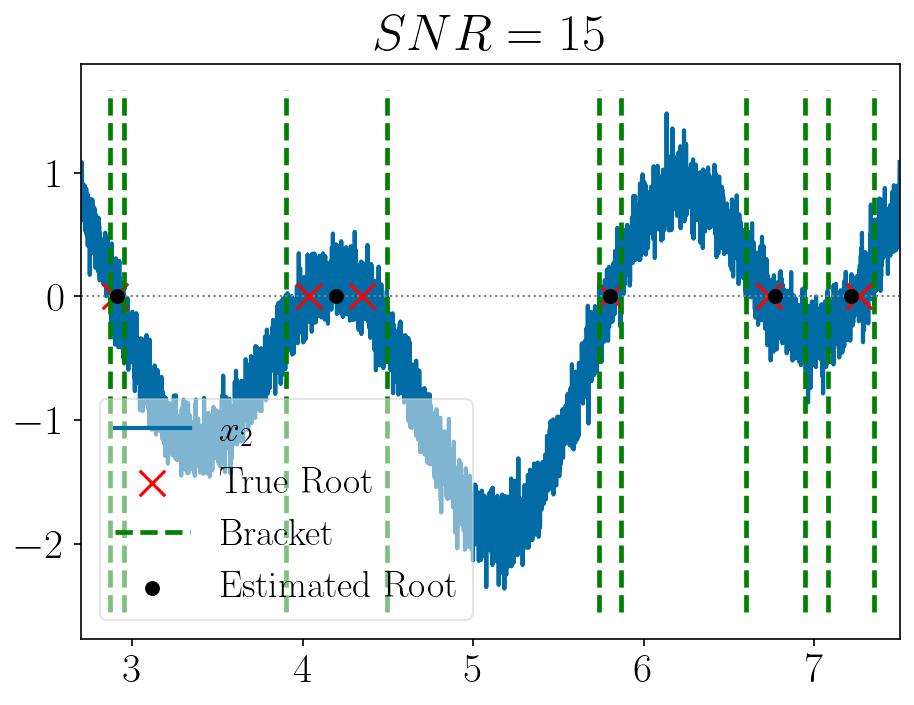}\label{fig:x2-15}}
\hfil
\subfloat[]{\includegraphics[width=2.75in]{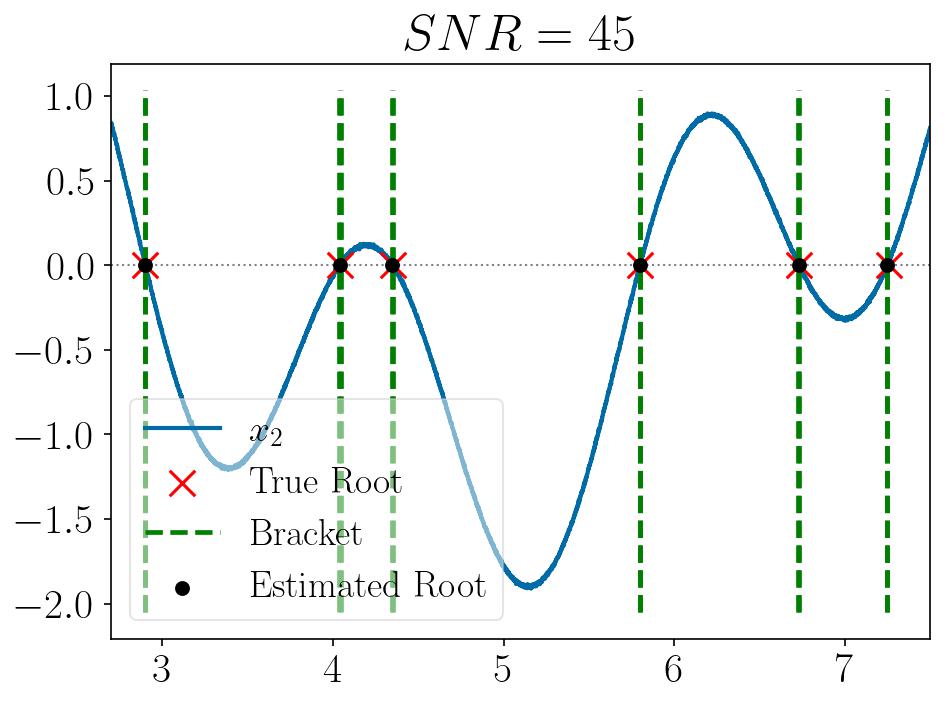}\label{fig:x2-45}}
\caption{True roots versus estimated roots for $x_{2}$ at a sampling frequency of 1000 Hz. Part (a) shows the algorithm missing a root due to high noise.}
\label{fig:SNR-shortcoming}
\end{figure*}
To go beyond one-off examples, we conducted a series of experiments  for establishing a thorough comparison between the algorithm presented herein and a state-of-the-art software-based zero detection algorithm developed by Molinaro and Sergeyev \cite{Molinaro2001}.
This method devises a support function for the time series and using estimates of local Lipschitz constant approximates the first zero-crossing in the interval. 
The authors have shown their algorithm to be substantially better than the simple grid technique for solving this problem for discrete signals. 
Appendix \ref{appendix:AppA} gives their algorithm in detail. 
Since the algorithm returns the first zero-crossing in the interval only, our comparisons consider the first root alone with a convergence criterion, $\epsilon$, set equal to the difference between two consecutive time values in a uniformly sampled series. 
The convergence criterion used is
\begin{equation}
\centering
\label{eq:convergence}
\epsilon = \frac{t_{b} - t_{a}}{n - 1}
\end{equation}
where $n = f(t_{b} - t_{a})$ is the number of discrete values sampled from the interval $[t_{a}$, $t_{b}]$ at a sampling frequency of $f$.

Fig.~\ref{fig:barchart} provides a comparison of the relative errors obtained by using both algorithms for all functions considered along with the time taken for convergence for signals having an SNR of 45 dB, 30 dB and 15 dB, each sampled at a frequency of $500$ Hz. 
The bar chart shows that in general, the relative error produced by Molinaro and Sergeyev's algorithm is significantly higher than the ones produced by the proposed algorithm. 
It is also evident that the relative error, in general, reduces with the increase in SNR.
\begin{figure}%
\centering
\includegraphics[width=.65\textwidth]{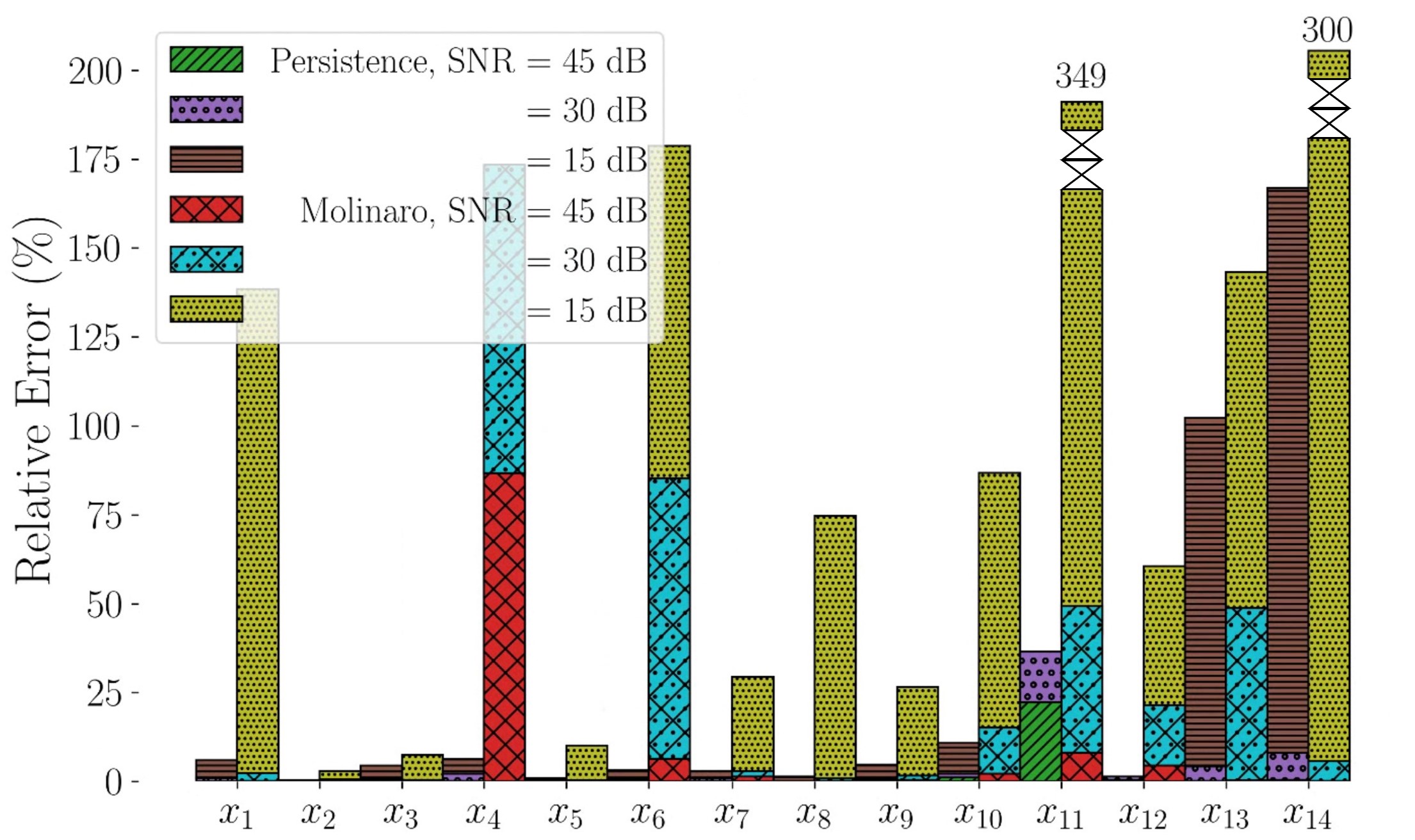}
\caption{Relative errors at low, medium and high SNR from both algorithms at $f = 500$ Hz}
\label{fig:barchart}
\end{figure}

Fig.~\ref{fig:heat} provides heat maps of relative error and time taken for convergence for the function $x_{2}$ using both algorithms over a range of SNRs and sampling frequencies. The figure demonstrates that the maximum relative error yielded by the proposed algorithm is $\sim3.5\%$ while that by Molinaro's algorithm is $\sim6\%$, while the maximum time taken by their algorithm is twice as high as the proposed algorithm for this case. Both these aspects further reinforce the conclusion that the proposed algorithm is more efficient. The rest of the cases can be found in Appendix \ref{appendix:AppD}.\\
\begin{figure}[!h]
\centering
\subfloat[]{\includegraphics[width=.4\linewidth]{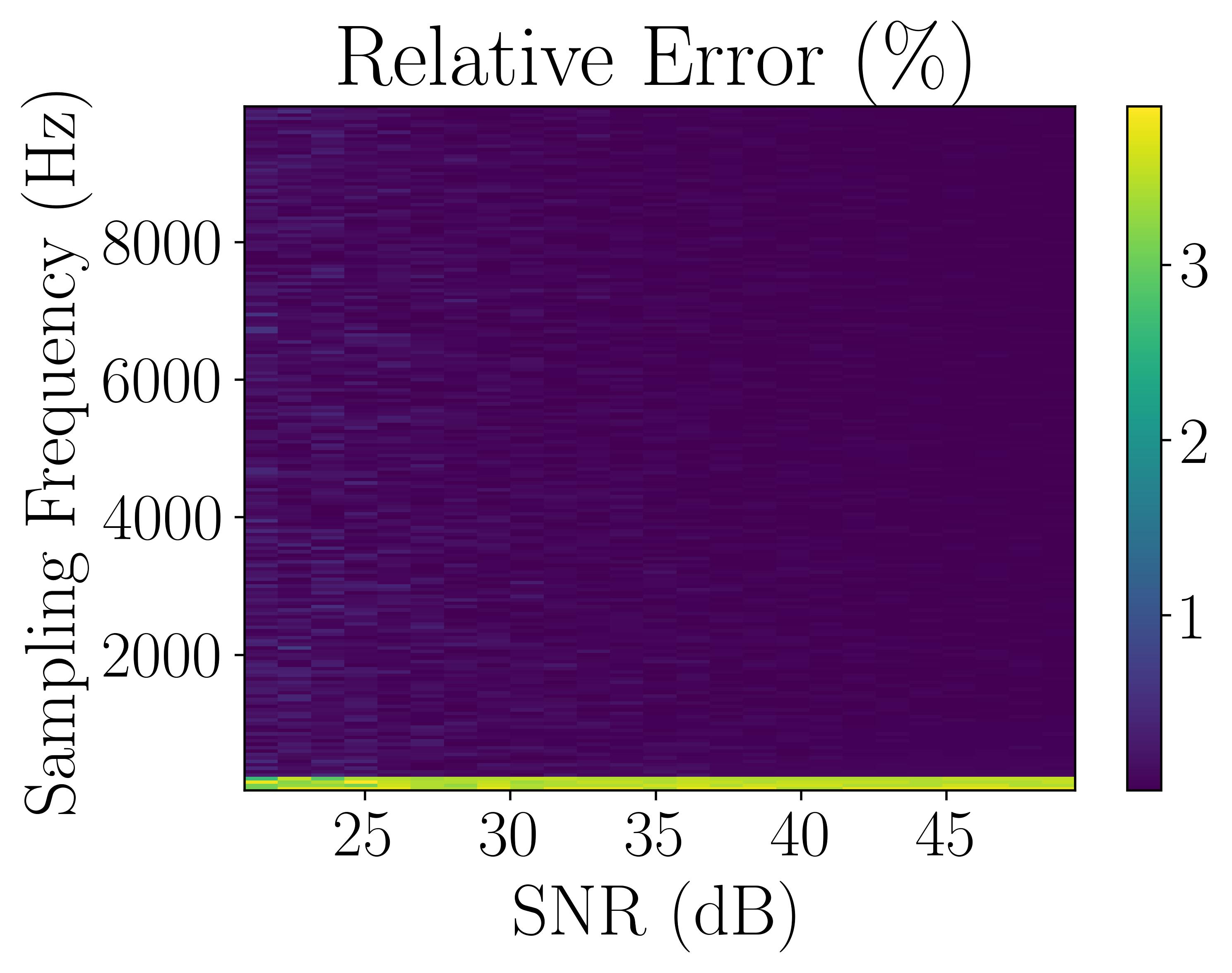}\label{fig:x2-0D-Err}}
\hfil
\subfloat[]{\includegraphics[width=.4\linewidth]{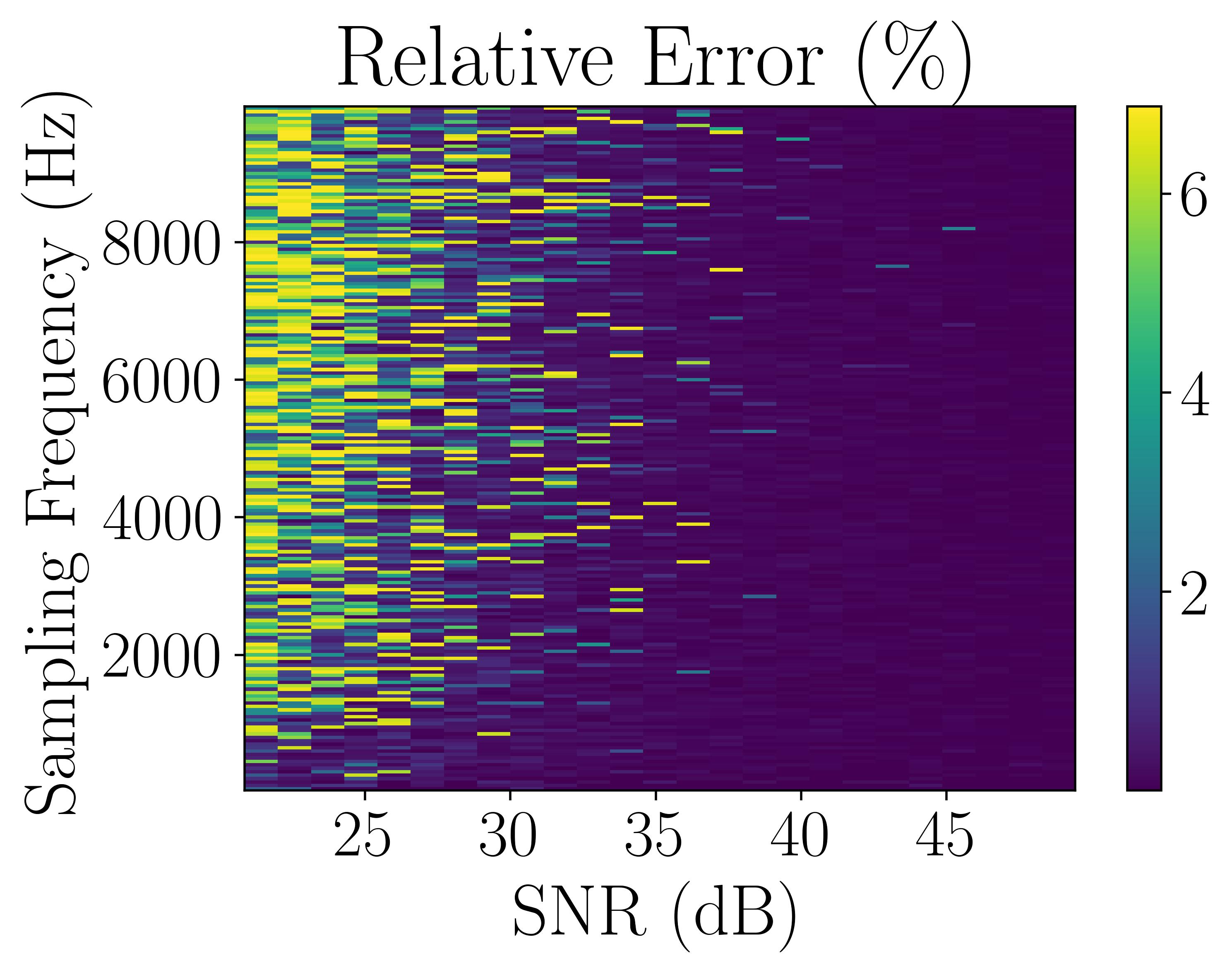}\label{fig:x2-mol-Err}}
\\
\subfloat[]{\includegraphics[width=.4\linewidth]{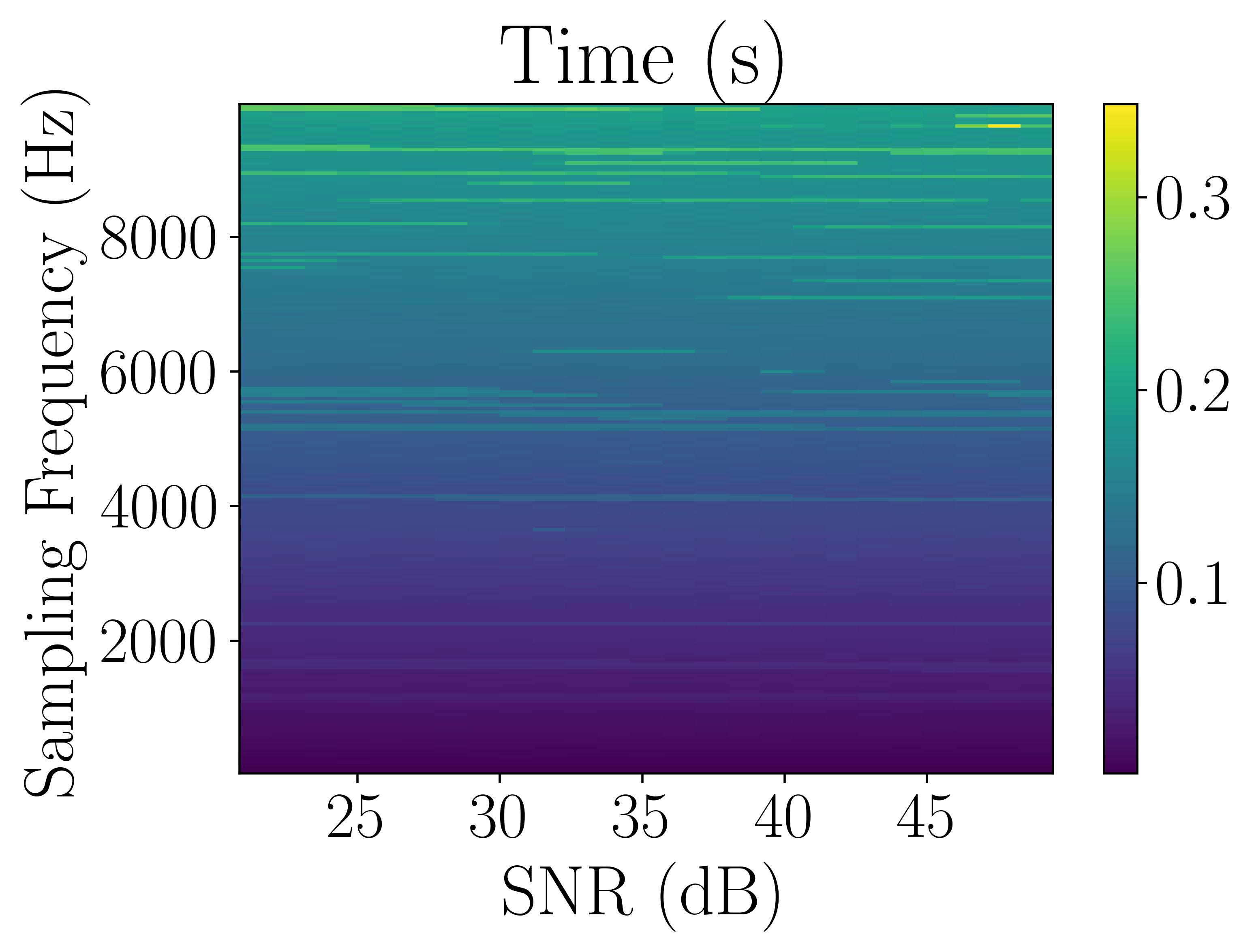}\label{fig:x2-0D-dur}}
\hfil
\subfloat[]{\includegraphics[width=.4\linewidth]{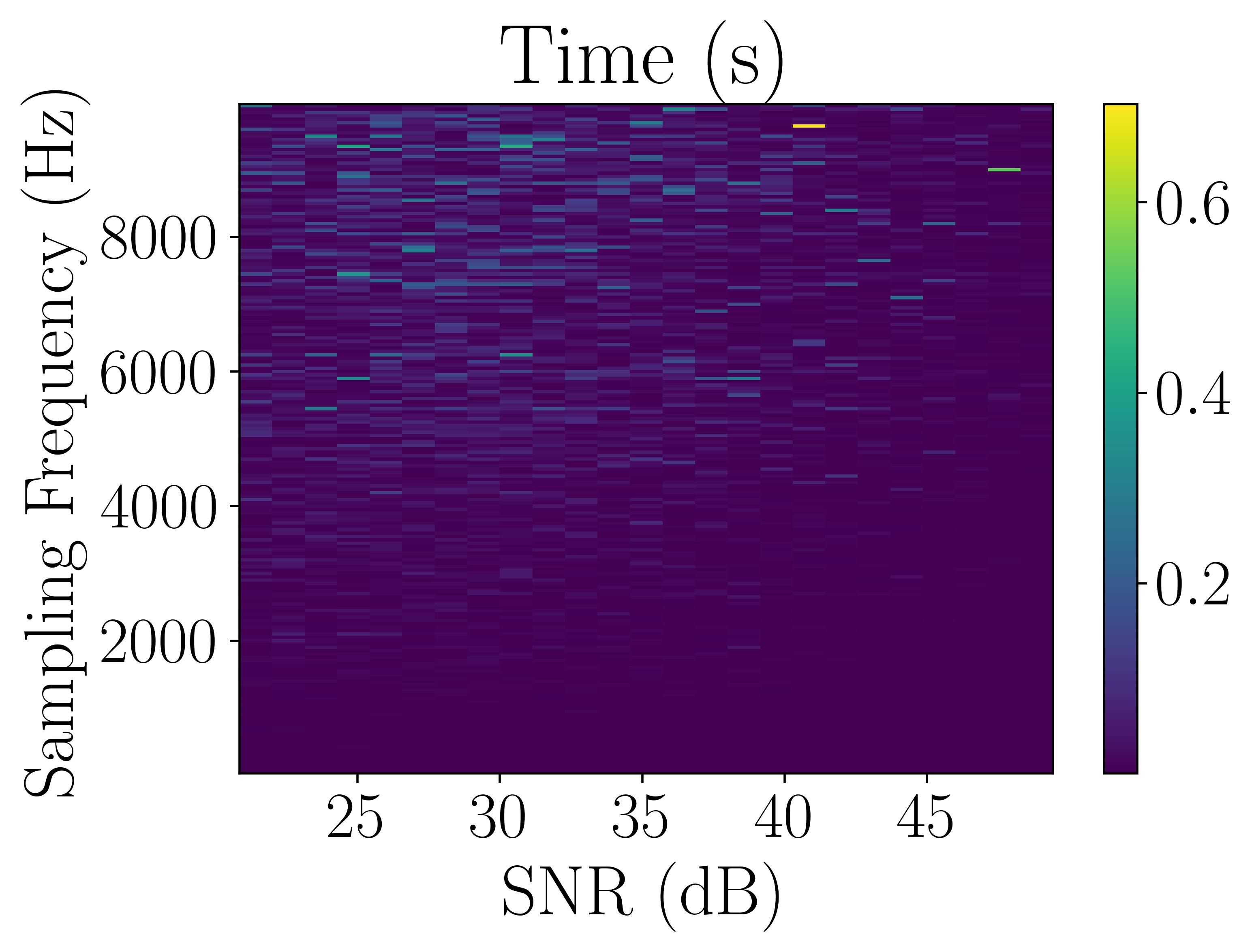}\label{fig:x2-mol-dur}}
\caption{Relative error and time taken for convergence by 0D Persistence (a, c) and Molinaro's algorithm (b, d) for $x_{2}$}
\label{fig:heat}
\end{figure}

%% file: sections/sec-conclusion.tex
\section{Conclusion}
\label{sec:Conc}

The mathematical problem of finding the roots of a signal is significant and has many applications in science and engineering. However, there is a lack of techniques in the literature that can determine all of the roots of a function, particularly for discrete time series. Some algorithms can only find roots if the function's expression is known, while others only converge to a single root within a given interval, ignoring any other roots that may be present.

To overcome these limitations, a new algorithm has been proposed for identifying the zero crossings of discrete time signals. This algorithm is based on the 0-dimensional persistence of the binarized pulse obtained from the signal. The points in the pulse corresponding to the correct zero brackets have a distinguishing characteristic: a higher persistence. Because of this high persistence and low frequency, these desired points can be treated as outliers - with the popular methods of z-score and Isolation Forest utilized to efficiently detect them. As a general rule, the z-score technique is more effective for medium to high sampling frequencies, while Isolation Forest is better suited to low sampling frequencies or small data sets.

Furthermore, the algorithm has been demonstrated to be robust to noise for a wide range of cases, except when noise creates artificial crossings, or if the root lies on the x-axis or the interval boundary in noise-free data. However, it is clear that these limitations will be inherent to most zero-bracketing algorithms.

Finally, the algorithm has been compared with an available software-based technique for zero-crossing detection in discrete signals presented by Molinaro and Sergeyev \cite{Molinaro2001}. The comparison deduces that our proposed algorithm is generally faster, more accurate and returns all the crossings in the given interval. While the algorithm may take longer to converge in some cases in comparison with Molinaro's algorithm, it is not possible to make a fair time-based comparison between the two algorithms since the competing algorithm returns only the first crossing of the interval, while our algorithm returns all of them. To conclude, the proposed zero-crossing algorithm is a reliable and powerful tool with numerous possible applications in the fields of science and engineering.

%% file: sections/sec-appendixA.tex
\section{Global Optimization Based Zero-Crossing Algorithm}
\label{appendix:AppA}

\begin{algorithm}[!htbp]
 \caption{Algorithm for first zero-crossing by Molinaro and Sergeyev \cite{Molinaro2001}}
\label{alg:molinaro}
\SetAlgoLined
\KwData{
A time series $f\left(t_{1}\right), \ \cdots,  f\left(t_{N}\right)$ with times $\left\{t_{1}<\cdots<t_{N}\right\}$ such that \\
\hspace{3.5em}$t^{(1)} = f(t_{1}), \text{ } t^{(2)} = f(t_{N}) \text{ and } t^{(n+1)} \text{ is found as below where n is the}$\\
\hspace{3.5em}{iteration number and } $t \in T_\delta$ for a $\delta$ grid.}
\KwResult{First zero-crossing or global minimum of the signal $f(t)$.}
\phantom{x}\\

Order the trial points $t^{(j)}, 1 \leq j \leq n$, such that $t_{1} < t_{2} < \cdots < t_{k} = t^{(n)} \leq t_{N}$. \\
Estimate the local Lipschitz constants as $m_{i}$ for each interval $[t_{i-1}, t_{i}]$, $[t_{i}, t_{i+1}]$, $2 \leq i \leq N-1$.
\[m_{i} = r\cdot max\{\lambda'_{i}, \lambda''_{i}, \epsilon\}\]
where $r > 1$ is a reliability parameter, $\epsilon$ is a small number such that $0 < \epsilon\leq 10^{-3}$, and $\lambda', \lambda''$ are found as follows:

\begin{equation*}
\lambda' =
    \begin{cases}
        max\{S_{k-1}, S_{k}\} & \text{if } k = 3 \text{ or } k = n\\
        max\{S_{k-1}, S_{k}, S_{k+1}\} & \text{if } \text{otherwise}
    \end{cases}
\end{equation*}

\[S_{i} = |f(t_{i}) - f(t_{i-1})|/(t_{i} - t_{i-1}), \quad 2 \leq i \leq k\]
\[\lambda''_{i} = \lambda^{(n+1)}_{max} (t_{i}-t_{i-1})/\Delta t^{(n+1)}_{max}, \quad 2 \leq i \leq k\]
\[\lambda^{(n+1)}_{max} = max\{S_{i}: 2 \leq i \leq k\}\]
\[\Delta t^{(n+1)}_{max} = max\{t_{i} - t_{i-1}: 2 \leq i \leq k\}\]

Compute $R_{i}$ and $\tau^{\text{th}}$ interval to calculate $t^{(n+1)}$, as follows:
\[R_{i} = 0.5\cdot\{f(t_{i}) + f(t_{i-1}) - m_{i}\cdot (t_{i} - t_{i-1})\}\]
\[\hat{t}_{i} = 0.5\cdot\{t_{i} + t_{i-1} - (f(t_{i}) - f(t_{i-1}))/m_{i}\}\]
\quad \quad If any $R_{i} \leq 0$, there is a zero-crossing. Using first $R_{i} \leq 0$, put $\tau = i$ and
\[t^{*}_{\tau} = t_{\tau - 1} + f(t_{\tau - 1})/m_{\tau} \quad \quad t^{(n+1)} = argmin\{|t_{\tau} - t|: t \in T_{\delta}\]
\quad \quad If all $R_{i} > 0$, find global minimum.
\[\tau = min\{j: j = argmin\{R_{i}: 2\leq i \leq k\}\}\]
\[t^{(n+1)} = argmin\{|\hat{t_{\tau}} - t|: t \in T_{\delta}\}\]

Verify if $|t^{(n+1)} - t_{\tau - 1}| \leq \sigma$ where $\sigma \text{ is the required tolerance}$. If not, repeat the procedure with $t_{N} = t^{(n+1)} \text{ if } f(t^{(n+1)} \leq 0$, and with $t_{1} < \cdots < t^{(n+1)} < \cdots < t_{N}$, otherwise.

\phantom{x}\\
\Return{First zero-crossing (if exists) or the global minimum in the interval $[t_{1}, t_{N}]$ with a tolerance of $\sigma$.}
\end{algorithm}

%% file: sections/sec-appendixB.tex
\section{Zero-Crossing Brackets}
\label{appendix:AppB}
Fig.~\ref{fig:appB_fig1} and \ref{fig:appB_fig2} demonstrate the capability of the algorithm to accurately bracket all zero-crossings of any time series.
\begin{figure}[!htbp]
\centering
\begin{subfigure}{0.43\textwidth}
  \centering
  \includegraphics[width=1.\linewidth]{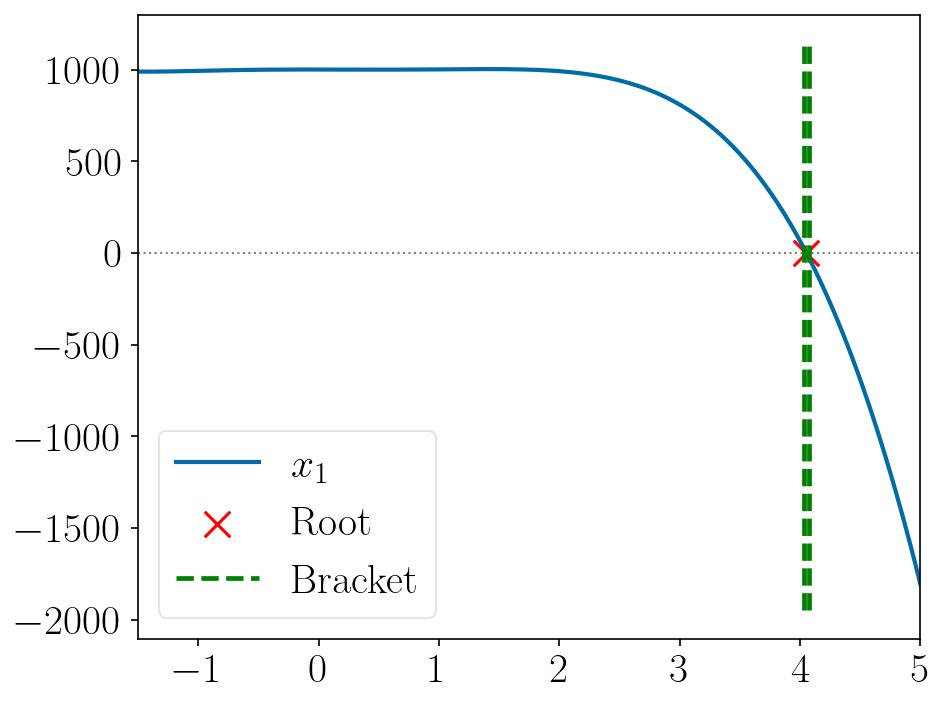}
  \caption{}
\end{subfigure}%
\begin{subfigure}{0.43\textwidth}
  \centering
  \includegraphics[width=1.\linewidth]{./figures/TrueRoots_Brackets/2.jpg}
  \caption{}
\end{subfigure}
\centering
\begin{subfigure}{0.43\textwidth}
  \centering
  \includegraphics[width=1.\linewidth]{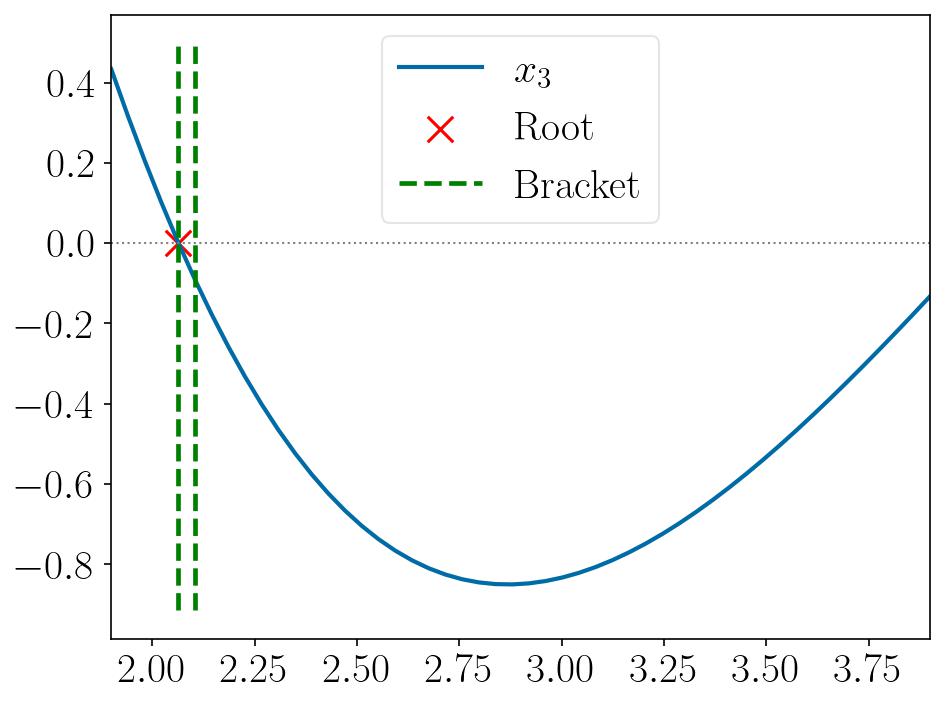}
  \caption{}
\end{subfigure}%
\begin{subfigure}{0.43\textwidth}
  \centering
  \includegraphics[width=1.\linewidth]{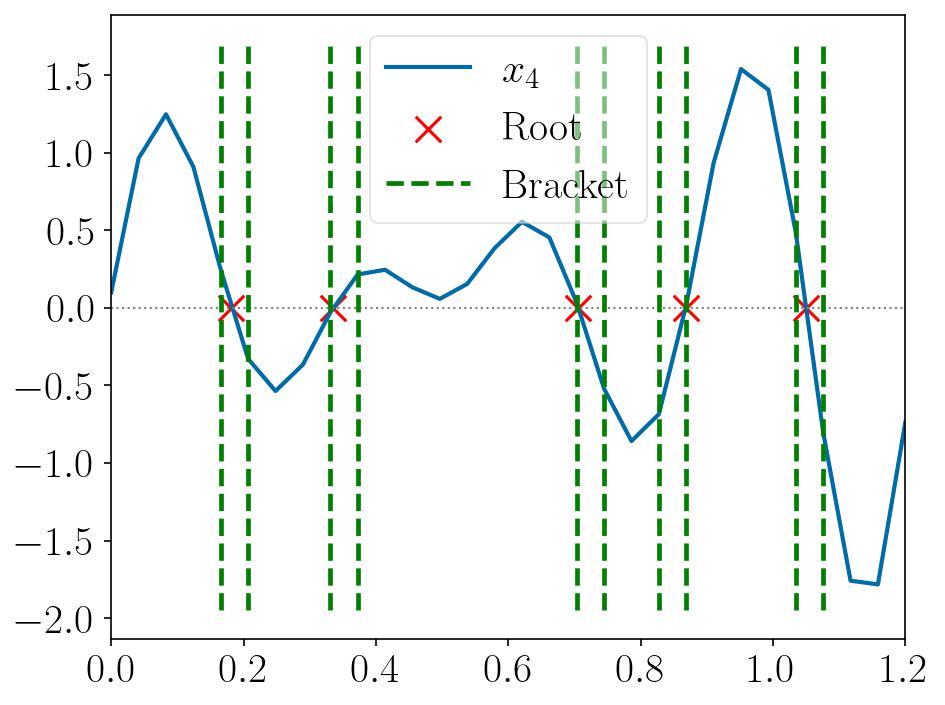}
  \caption{}
\end{subfigure}
\centering
\begin{subfigure}{0.43\textwidth}
  \centering
  \includegraphics[width=1.\linewidth]{./figures/TrueRoots_Brackets/5.jpg}
  \caption{}
\end{subfigure}%
\begin{subfigure}{0.43\textwidth}
  \centering
  \includegraphics[width=1.\linewidth]{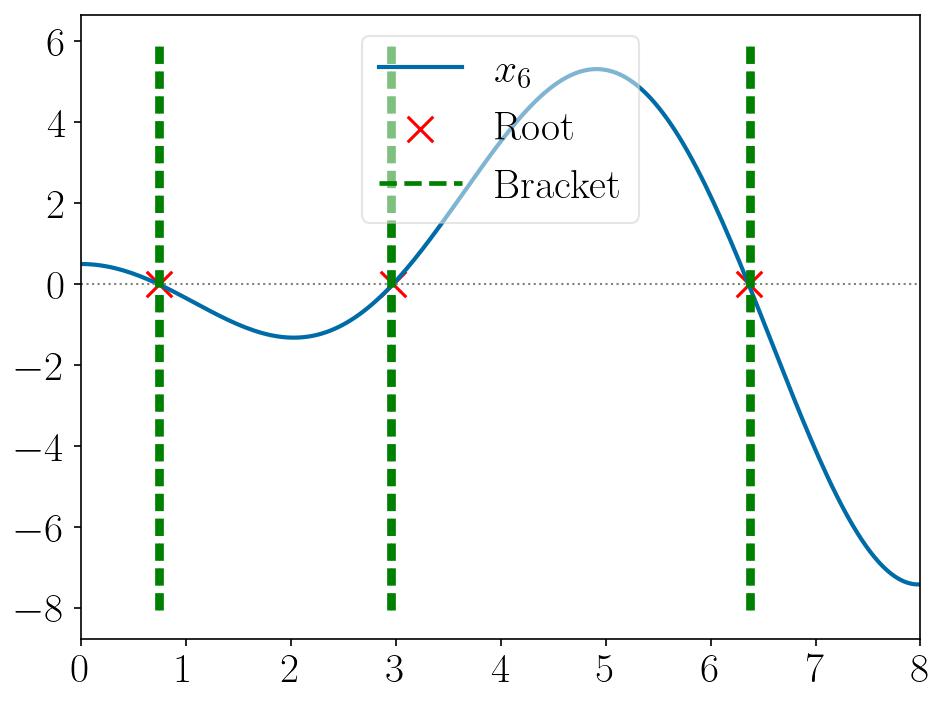}
  \caption{}
\end{subfigure}
\caption{True zero-crossing(s) and brackets returned by the algorithm for functions $x_{1}$ to $x_{6}$ at a low sampling frequency of 25 Hz.}
\label{fig:appB_fig1}
\end{figure}
\begin{figure}
\centering
\begin{subfigure}{0.43\textwidth}
  \centering
  \includegraphics[width=1.\linewidth]{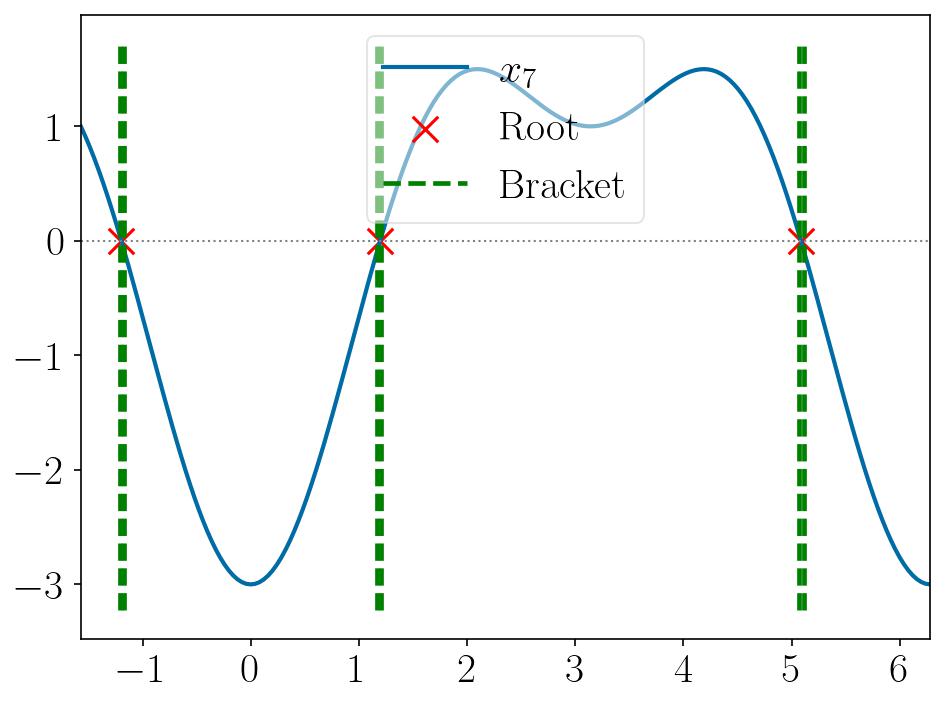}
  \caption{}
\end{subfigure}%
\begin{subfigure}{0.43\textwidth}
  \centering
  \includegraphics[width=1.\linewidth]{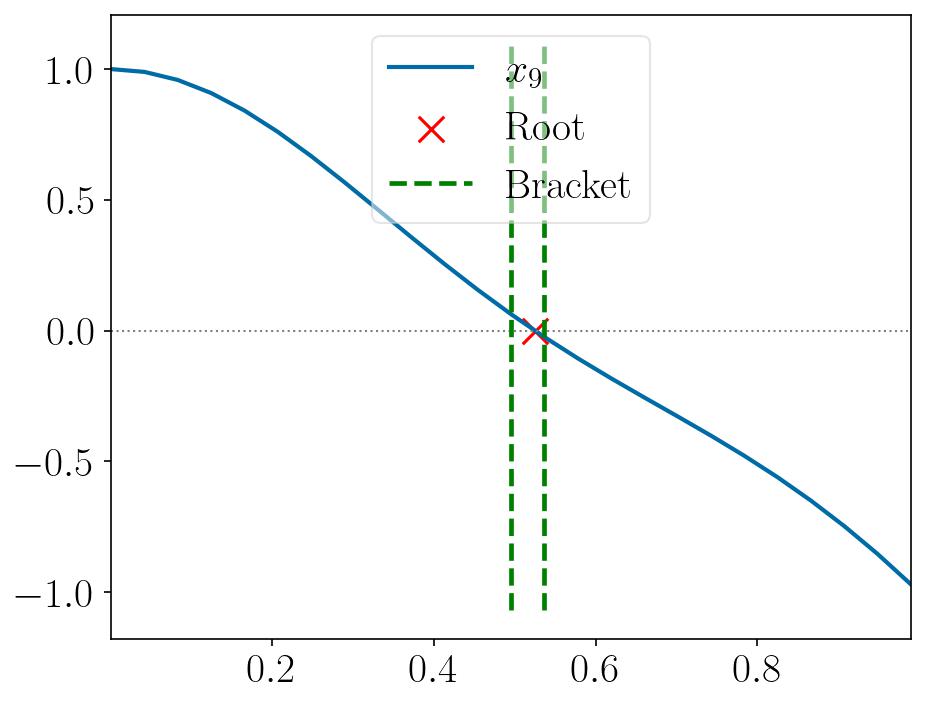}
  \caption{}
\end{subfigure}
\centering
\begin{subfigure}{0.43\textwidth}
  \centering
  \includegraphics[width=1.\linewidth]{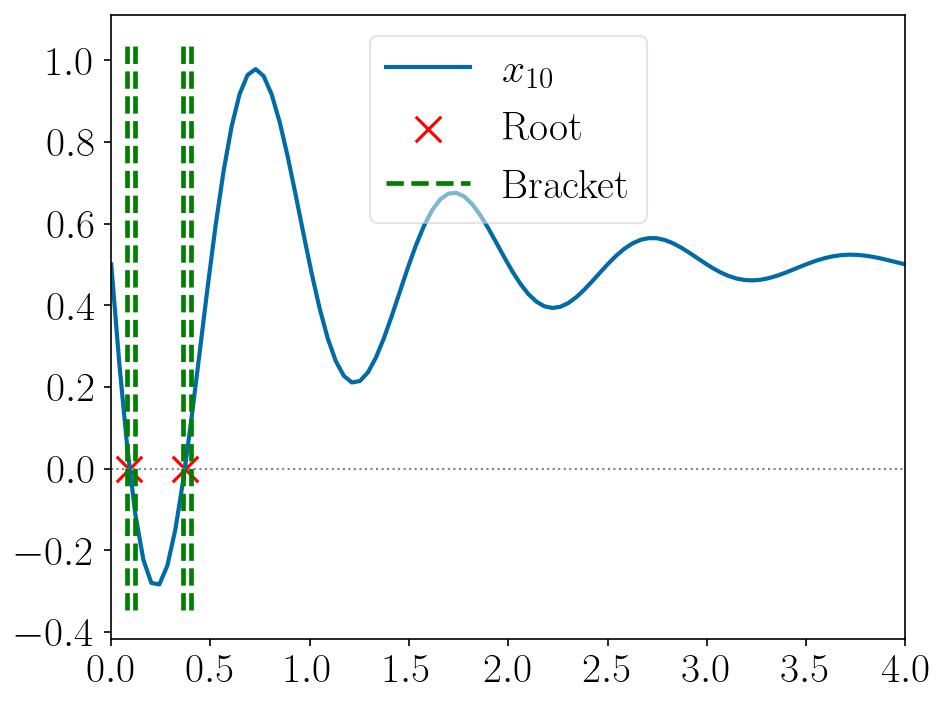}
  \caption{}
\end{subfigure}%
\begin{subfigure}{0.43\textwidth}
  \centering
  \includegraphics[width=1.\linewidth]{./figures/TrueRoots_Brackets/11.jpg}
  \caption{}
\end{subfigure}
\centering
\begin{subfigure}{0.43\textwidth}
  \centering
  \includegraphics[width=1.\linewidth]{./figures/TrueRoots_Brackets/12.jpg}
  \caption{}
\end{subfigure}%
\begin{subfigure}{0.43\textwidth}
  \centering
  \includegraphics[width=1.\linewidth]{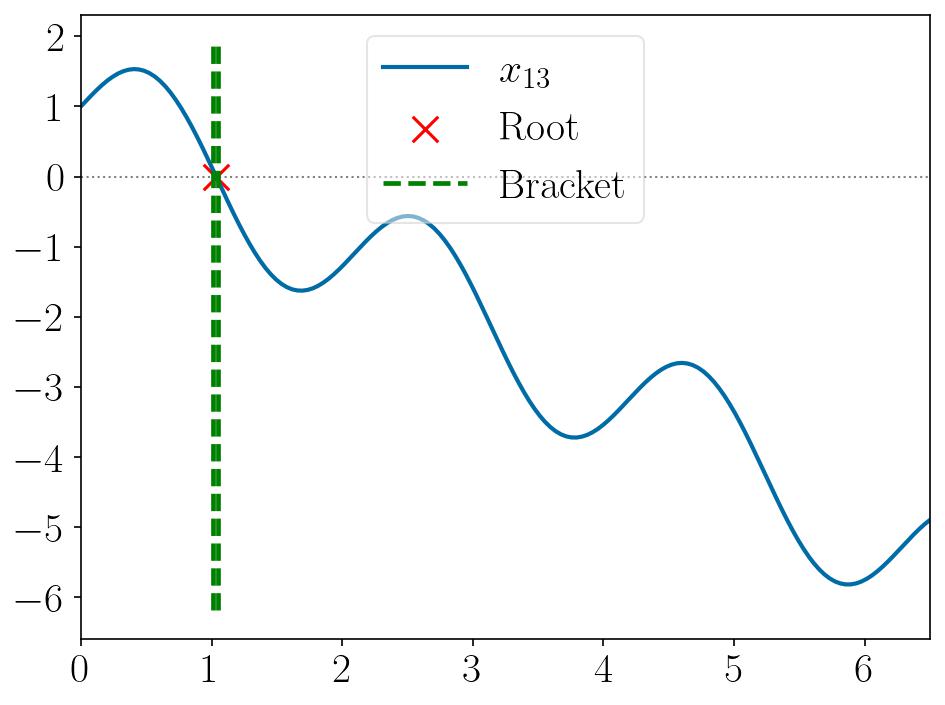}
  \caption{}
\end{subfigure}
\caption{True zero-crossing(s) and brackets returned by the algorithm for functions $x_{7}$, and $x_{9}$ to $x_{13}$ at a low sampling frequency of 25 Hz.}
\label{fig:appB_fig2}
\end{figure}

%% file: sections/sec-appendixC.tex
\section{Root Brackets with Varying SNR}
\label{appendix:AppC}
Fig.~\ref{fig:appC_fig1} to \ref{fig:appC_fig2} show the robustness of the algorithm to noise by estimating the root(s) with a sufficient accuracy even at low SNR.
\begin{figure}[!htbp]
\centering
\begin{subfigure}{0.43\textwidth}
  \centering
  \includegraphics[width=1.\linewidth]{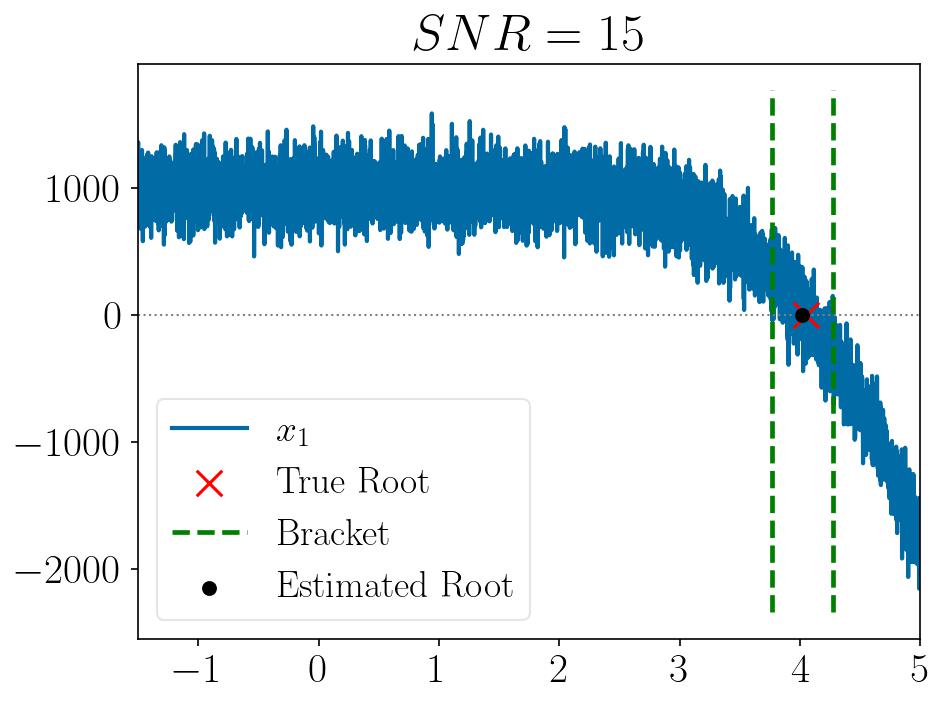}
  \caption{}
\end{subfigure}%
\begin{subfigure}{0.43\textwidth}
  \centering
  \includegraphics[width=1.\linewidth]{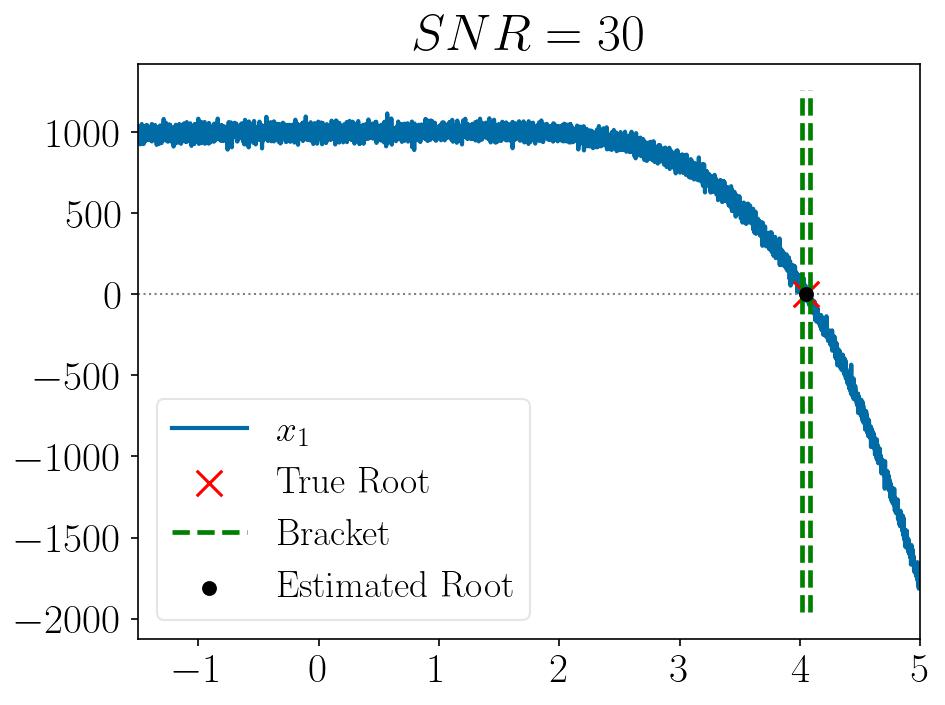}
  \caption{}
\end{subfigure}
\begin{subfigure}{0.43\textwidth}
  \centering
  \includegraphics[width=1.\linewidth]{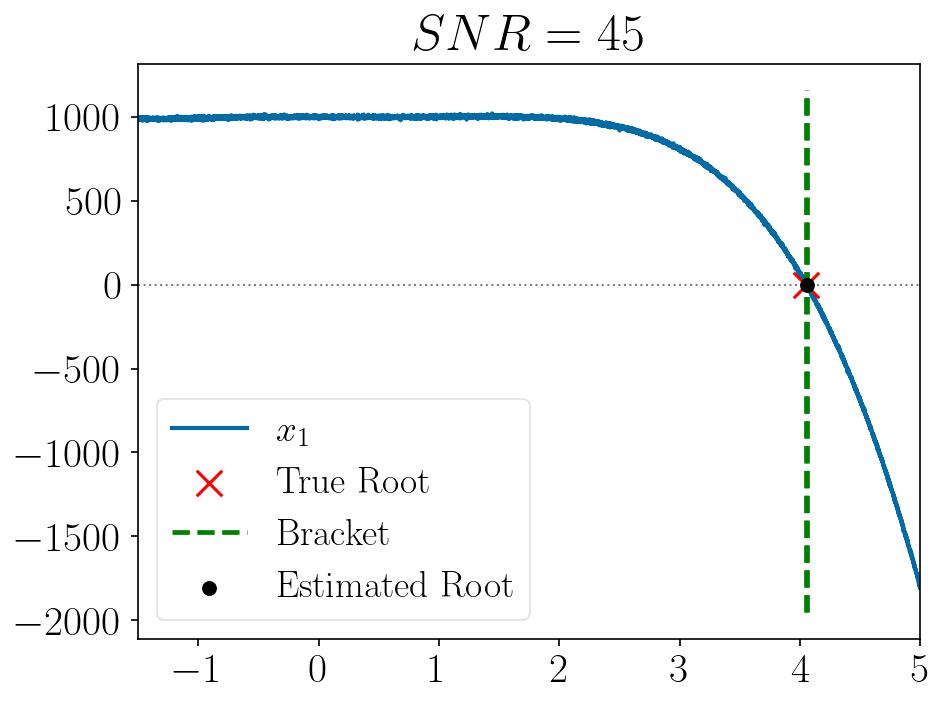}
  \caption{}
\end{subfigure}%
\begin{subfigure}{0.43\textwidth}
  \centering
  \includegraphics[width=1.\linewidth]{./figures/SNR_varied/2_15.jpg}
  \caption{}
\end{subfigure}
\begin{subfigure}{0.43\textwidth}
  \centering
  \includegraphics[width=1.\linewidth]{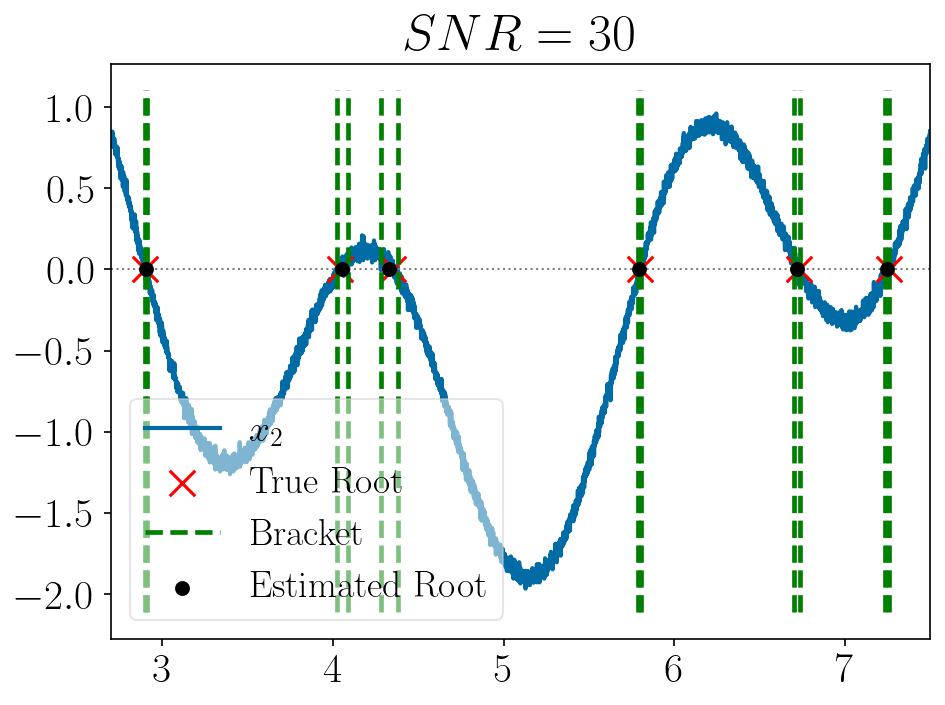}
  \caption{}
\end{subfigure}%
\begin{subfigure}{0.43\textwidth}
  \centering
  \includegraphics[width=1.\linewidth]{./figures/SNR_varied/2_45.jpg}
  \caption{}
\end{subfigure}
\caption{True roots, brackets and estimated roots returned at a sampling frequency of 1000 Hz with low, medium and high SNR for functions $x_{1}$ and $x_{2}$}
\label{fig:appC_fig1}
\end{figure}
\begin{figure}[!htbp]
\centering
\begin{subfigure}{0.43\textwidth}
  \centering
  \includegraphics[width=1.\linewidth]{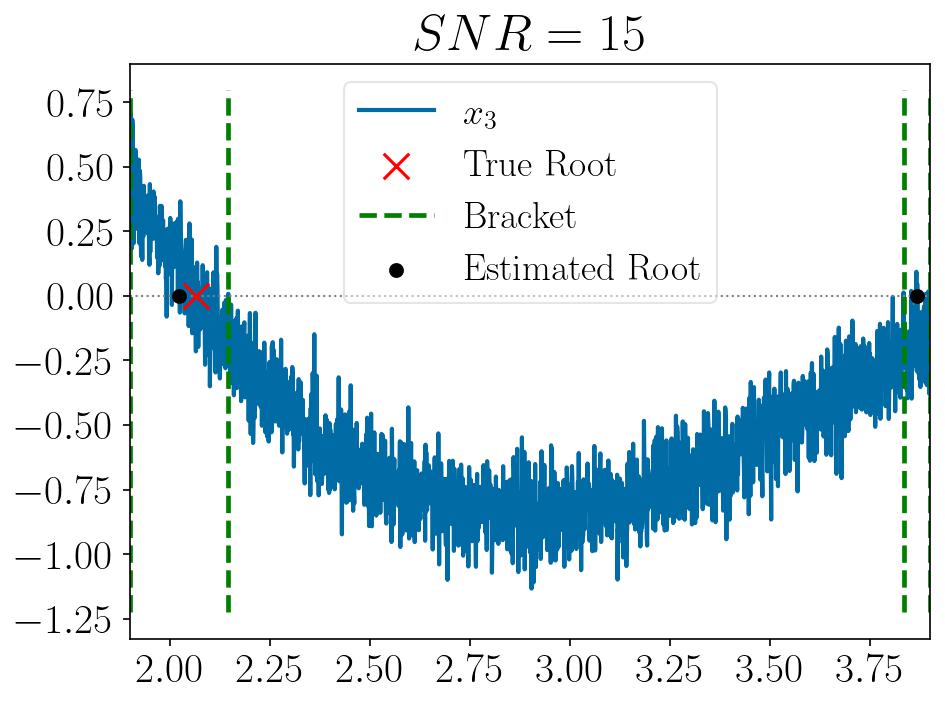}
  \caption{}
\end{subfigure}%
\begin{subfigure}{0.43\textwidth}
  \centering
  \includegraphics[width=1.\linewidth]{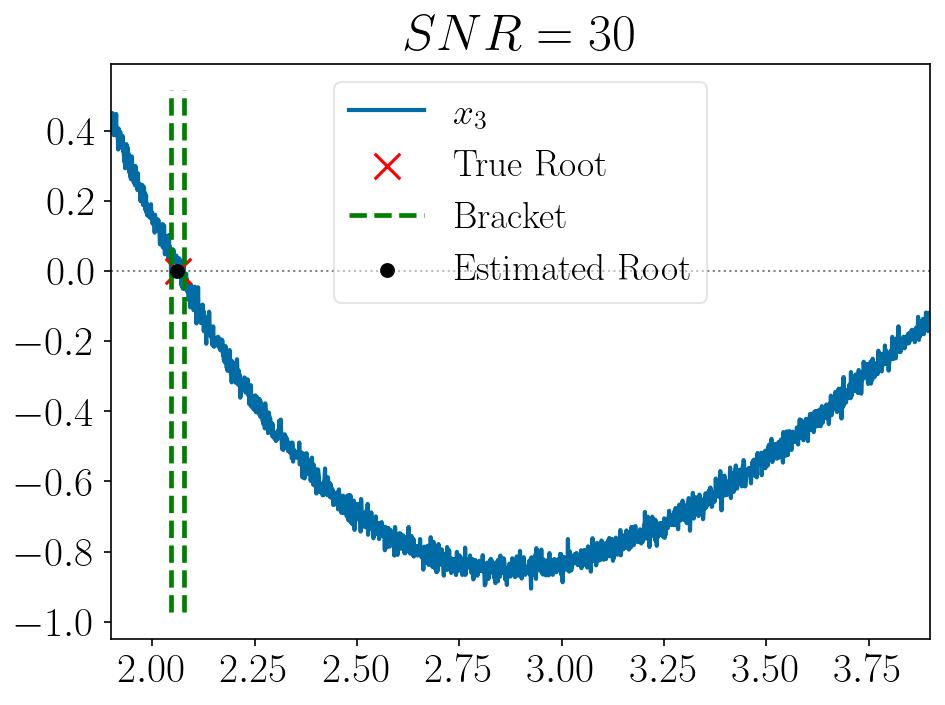}
  \caption{}
\end{subfigure}
\begin{subfigure}{0.43\textwidth}
  \centering
  \includegraphics[width=1.\linewidth]{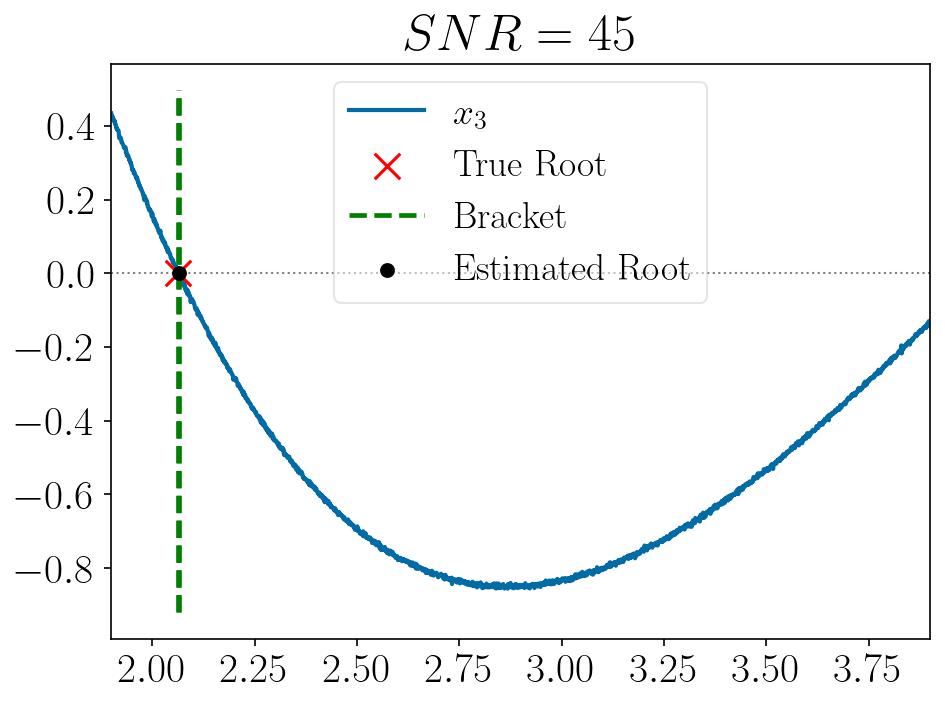}
  \caption{}
\end{subfigure}%
\begin{subfigure}{0.43\textwidth}
  \centering
  \includegraphics[width=1.\linewidth]{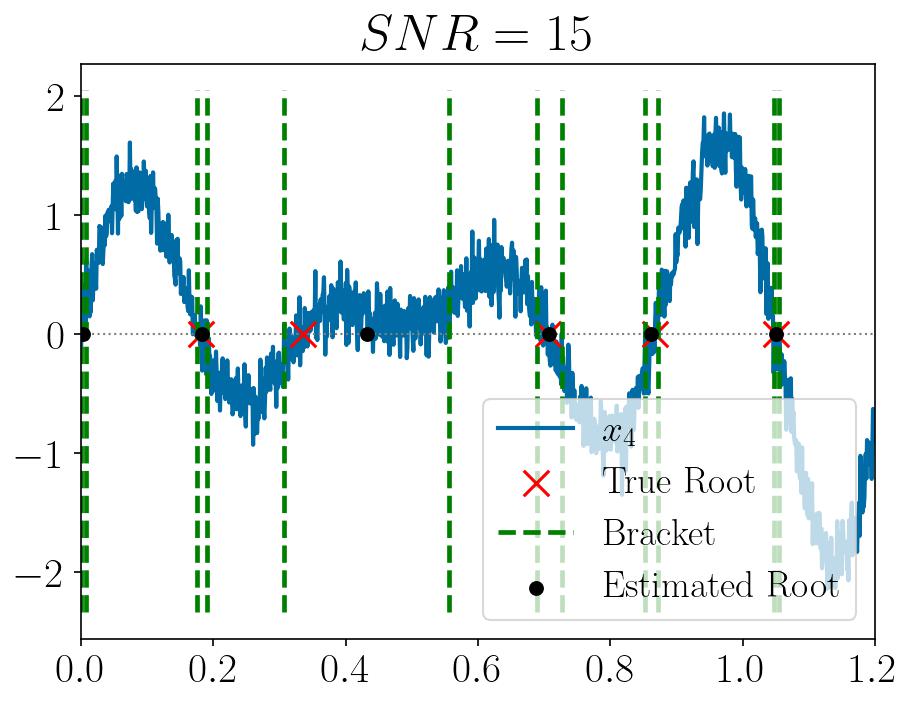}
  \caption{}
\end{subfigure}
\begin{subfigure}{0.43\textwidth}
  \centering
  \includegraphics[width=1.\linewidth]{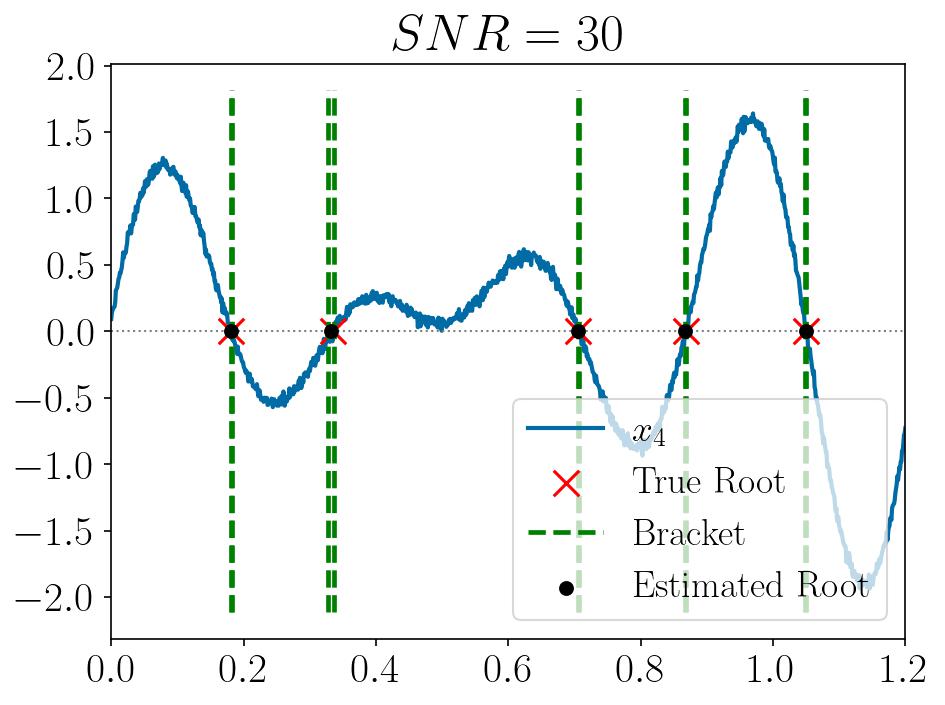}
  \caption{}
\end{subfigure}%
\begin{subfigure}{0.43\textwidth}
  \centering
  \includegraphics[width=1.\linewidth]{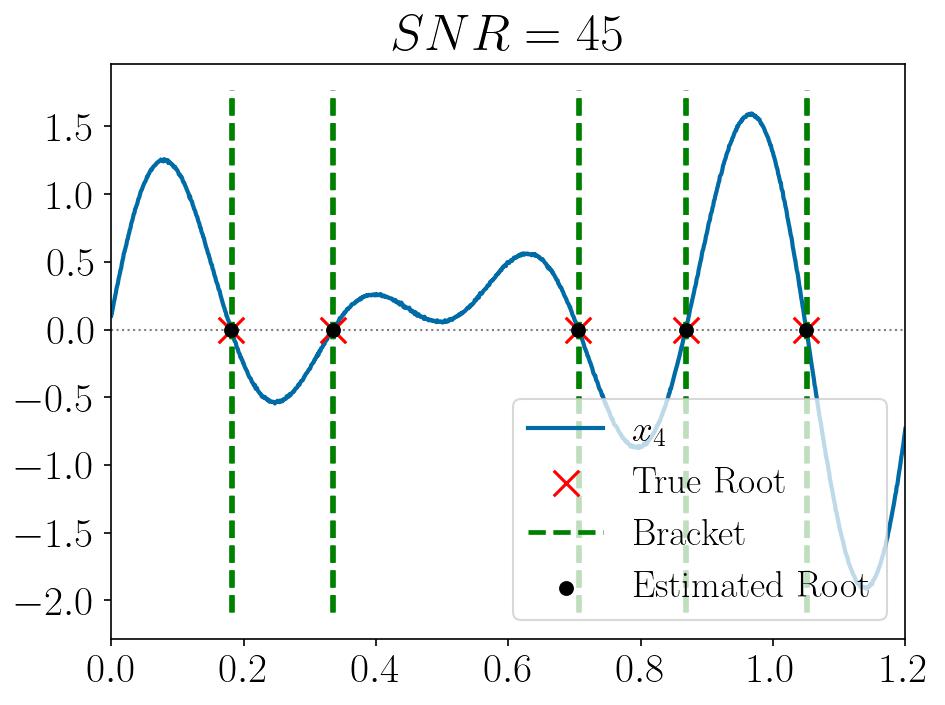}
  \caption{}
\end{subfigure}
\caption{True roots, brackets and estimated roots returned at a sampling frequency of 1000 Hz with low, medium and high SNR for functions $x_{3}$ and $x_{4}$}
\end{figure}
\begin{figure}[!htbp]
\centering
\begin{subfigure}{0.43\textwidth}
  \centering
  \includegraphics[width=1.\linewidth]{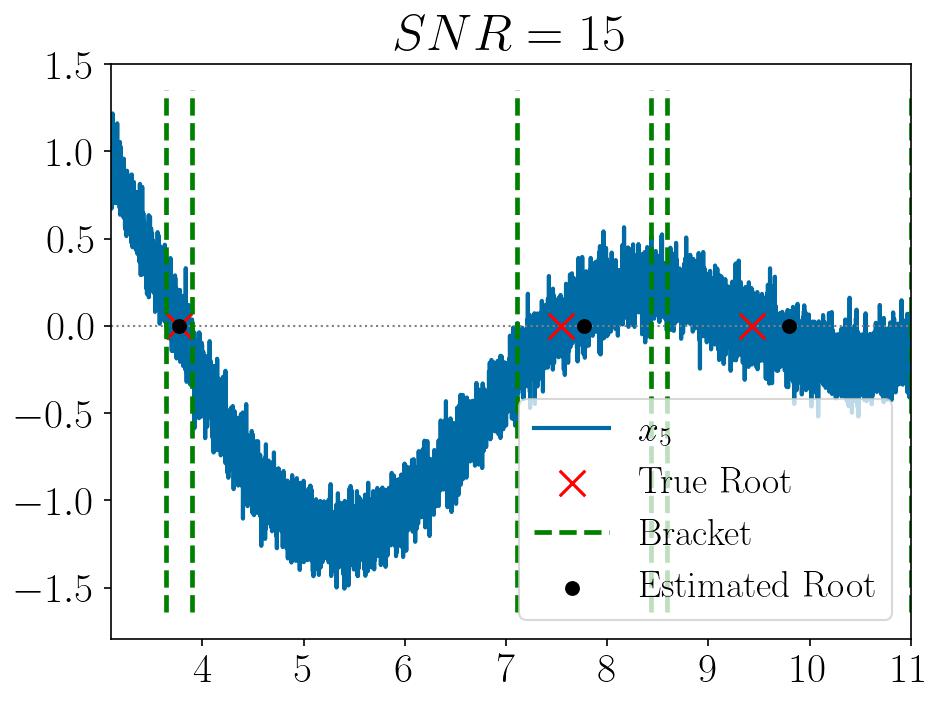}
  \caption{}
\end{subfigure}%
\begin{subfigure}{0.43\textwidth}
  \centering
  \includegraphics[width=1.\linewidth]{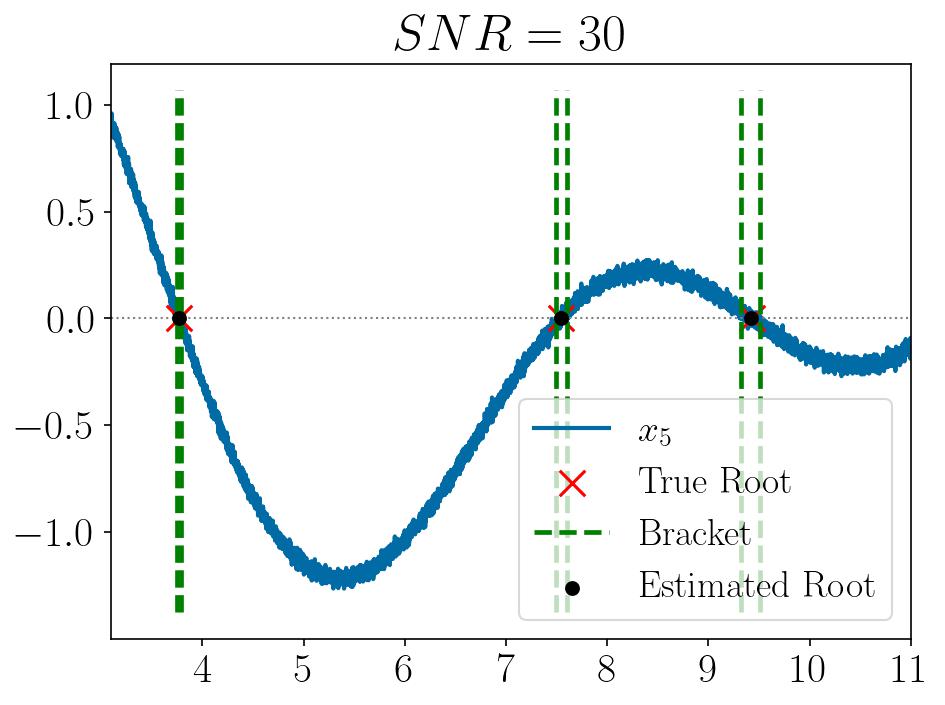}
  \caption{}
\end{subfigure}
\begin{subfigure}{0.43\textwidth}
  \centering
  \includegraphics[width=1.\linewidth]{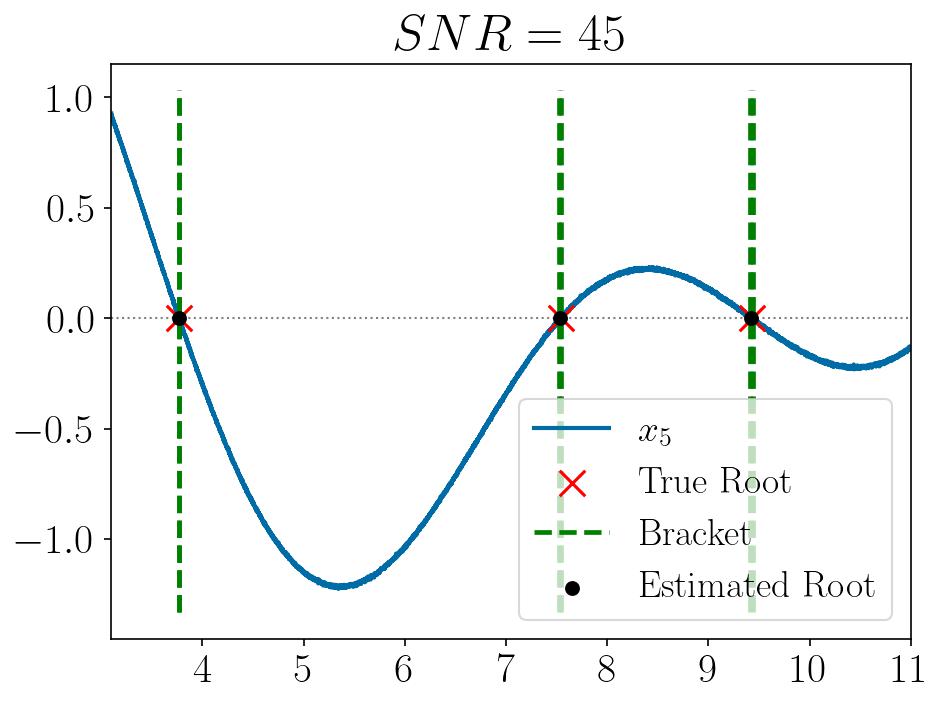}
  \caption{}
\end{subfigure}%
\begin{subfigure}{0.43\textwidth}
  \centering
  \includegraphics[width=1.\linewidth]{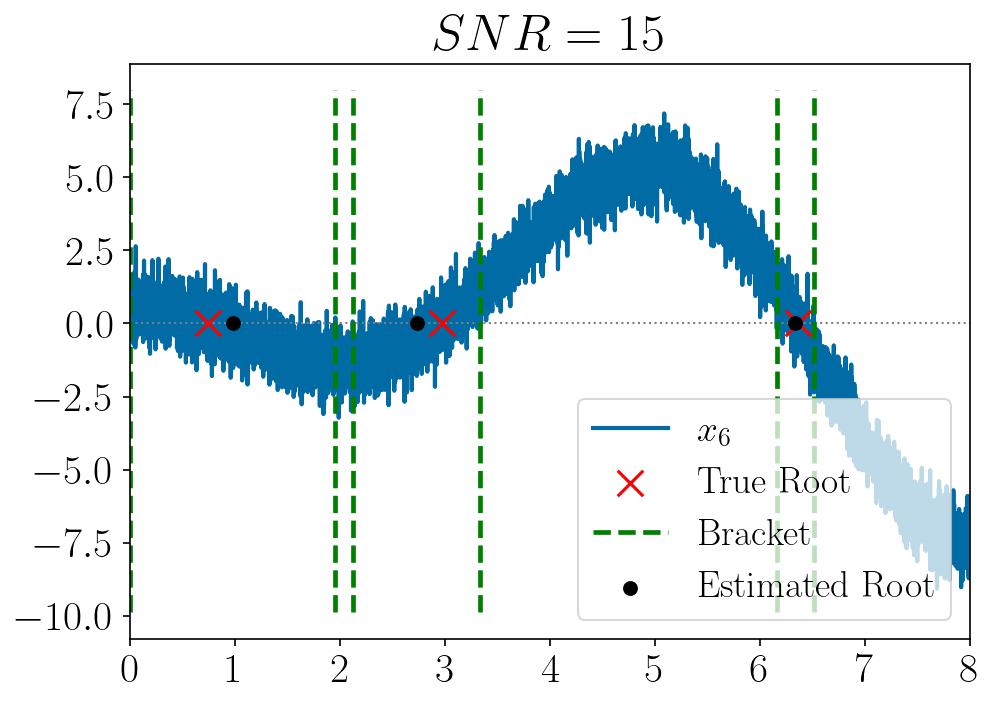}
  \caption{}
\end{subfigure}
\begin{subfigure}{0.43\textwidth}
  \centering
  \includegraphics[width=1.\linewidth]{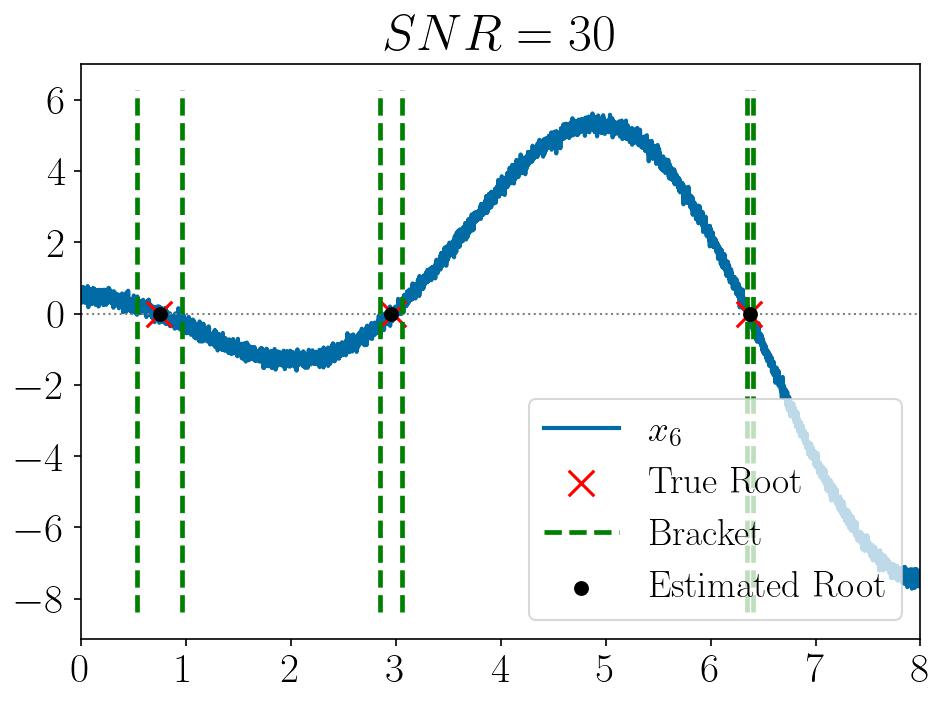}
  \caption{}
\end{subfigure}%
\begin{subfigure}{0.43\textwidth}
  \centering
  \includegraphics[width=1.\linewidth]{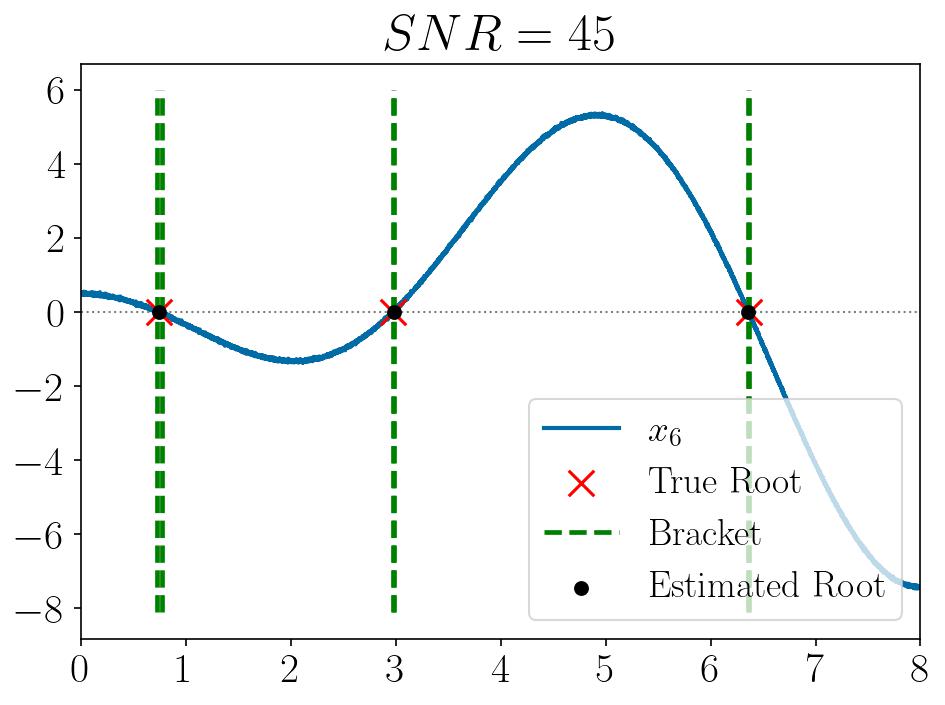}
  \caption{}
\end{subfigure}
\caption{True roots, brackets and estimated roots returned at a sampling frequency of 1000 Hz with low, medium and high SNR for functions $x_{5}$ and $x_{6}$}
\end{figure}
\begin{figure}[!htbp]
\centering
\begin{subfigure}{0.43\textwidth}
  \centering
  \includegraphics[width=1.\linewidth]{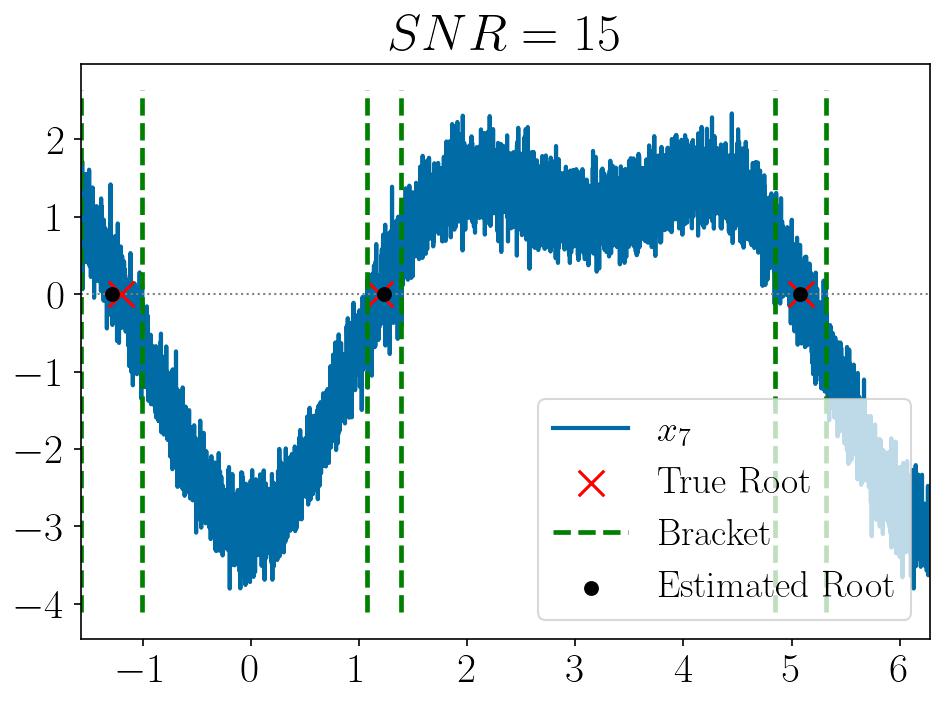}
  \caption{}
\end{subfigure}%
\begin{subfigure}{0.43\textwidth}
  \centering
  \includegraphics[width=1.\linewidth]{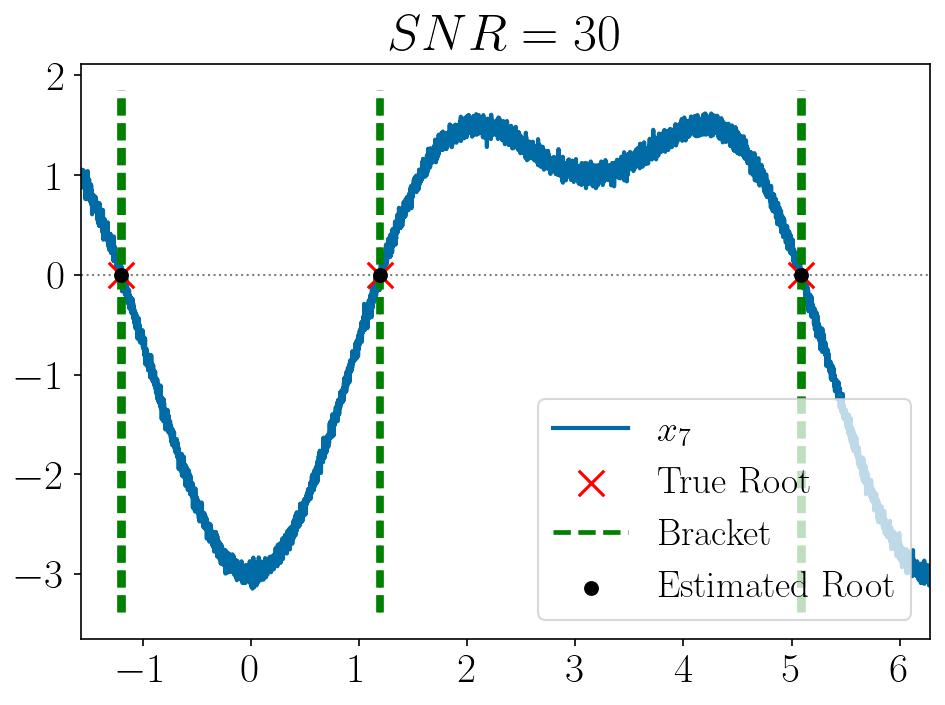}
  \caption{}
\end{subfigure}
\begin{subfigure}{0.43\textwidth}
  \centering
  \includegraphics[width=1.\linewidth]{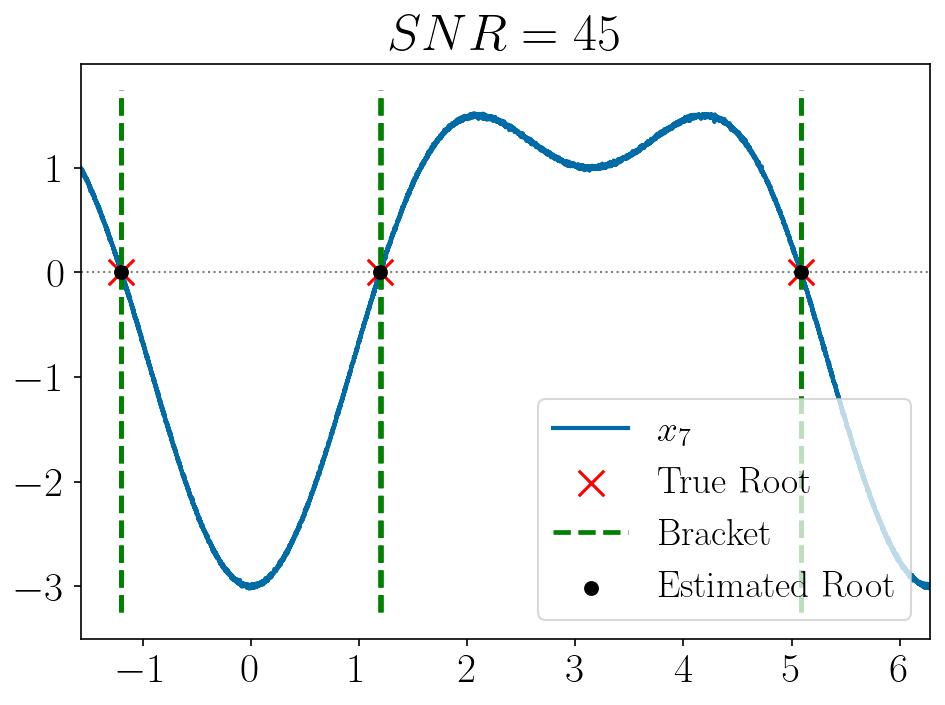}
  \caption{}
\end{subfigure}%
\begin{subfigure}{0.43\textwidth}
  \centering
  \includegraphics[width=1.\linewidth]{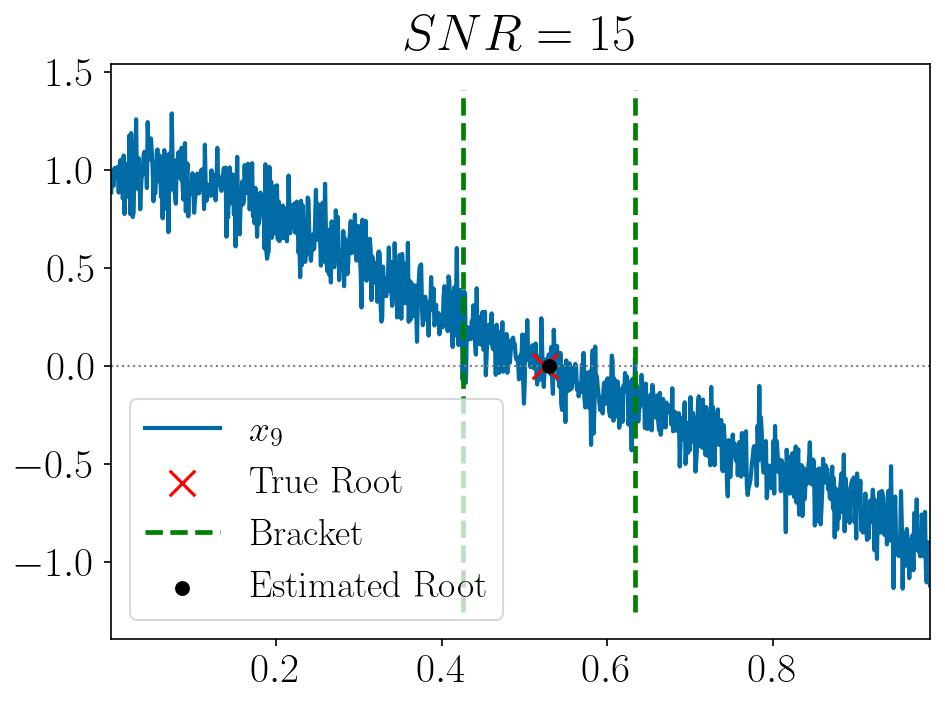}
  \caption{}
\end{subfigure}
\begin{subfigure}{0.43\textwidth}
  \centering
  \includegraphics[width=1.\linewidth]{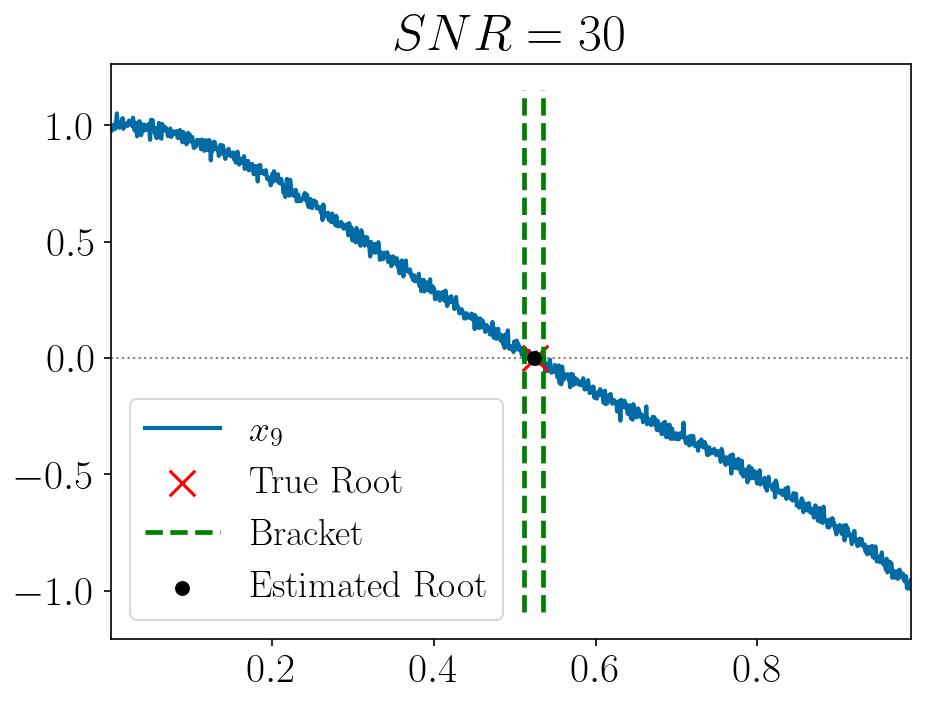}
  \caption{}
\end{subfigure}%
\begin{subfigure}{0.43\textwidth}
  \centering
  \includegraphics[width=1.\linewidth]{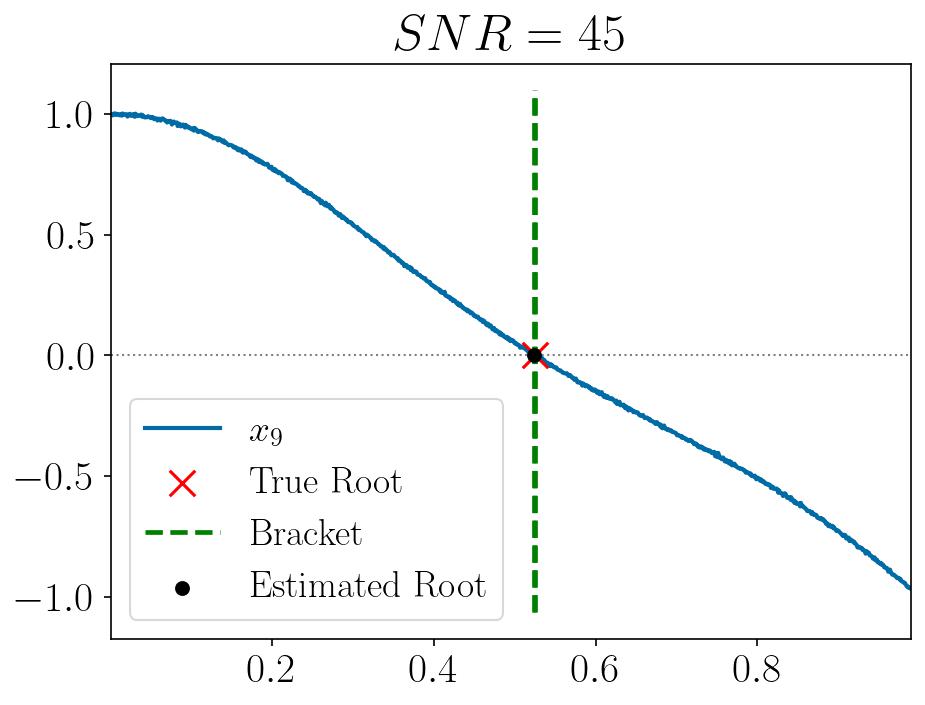}
  \caption{}
\end{subfigure}
\caption{True roots, brackets and estimated roots returned at a sampling frequency of 1000 Hz with low, medium and high SNR for functions $x_{7}$ and $x_{9}$}
\end{figure}
\begin{figure}[!htbp]
\centering
\begin{subfigure}{0.43\textwidth}
  \centering
  \includegraphics[width=1.\linewidth]{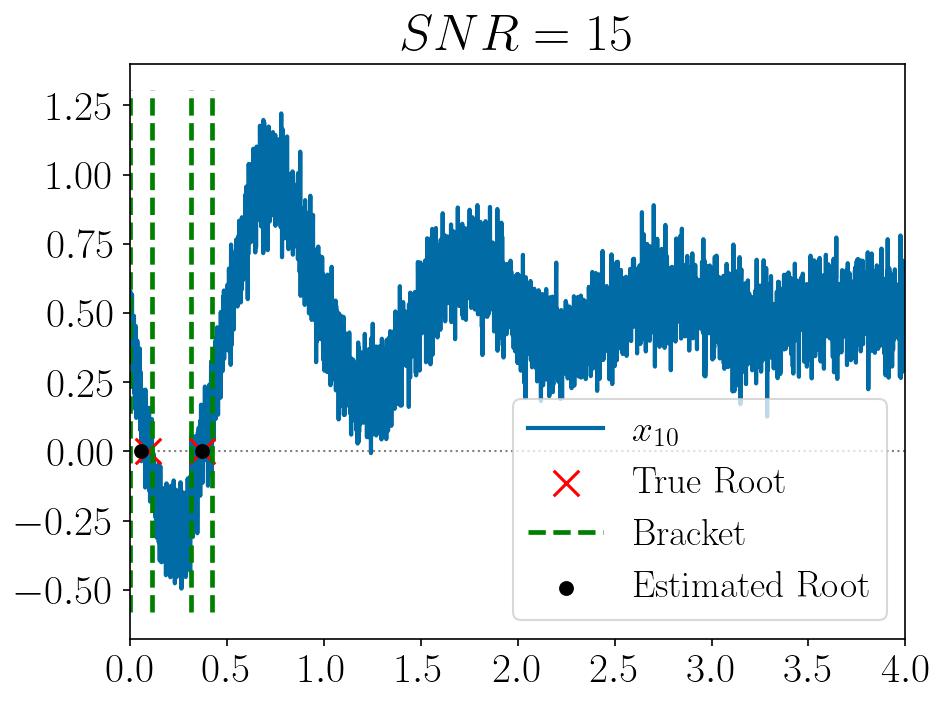}
  \caption{}
\end{subfigure}%
\begin{subfigure}{0.43\textwidth}
  \centering
  \includegraphics[width=1.\linewidth]{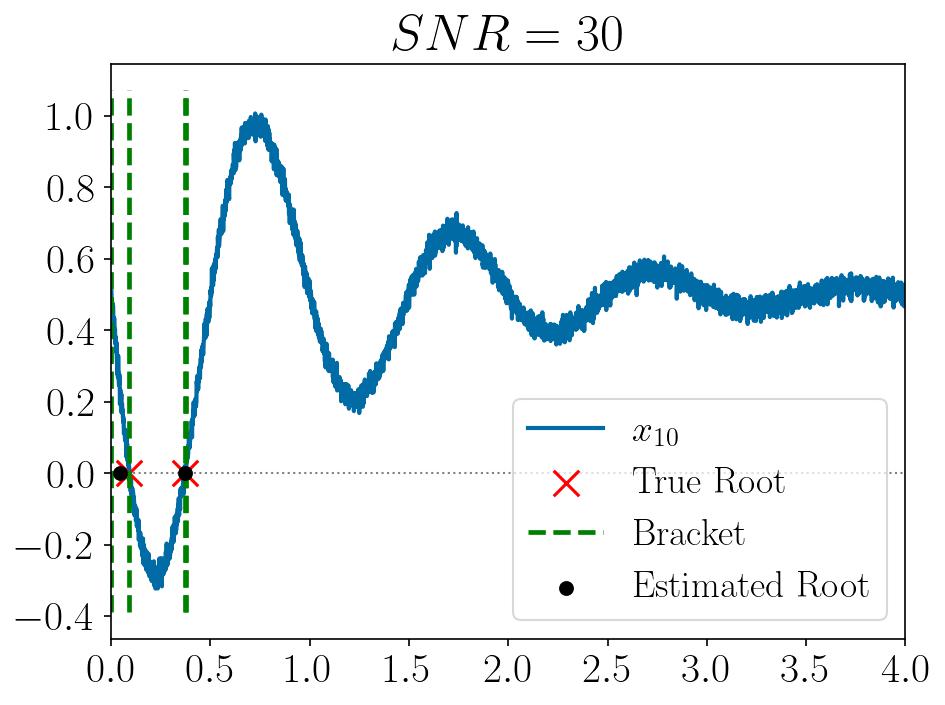}
  \caption{}
\end{subfigure}
\begin{subfigure}{0.43\textwidth}
  \centering
  \includegraphics[width=1.\linewidth]{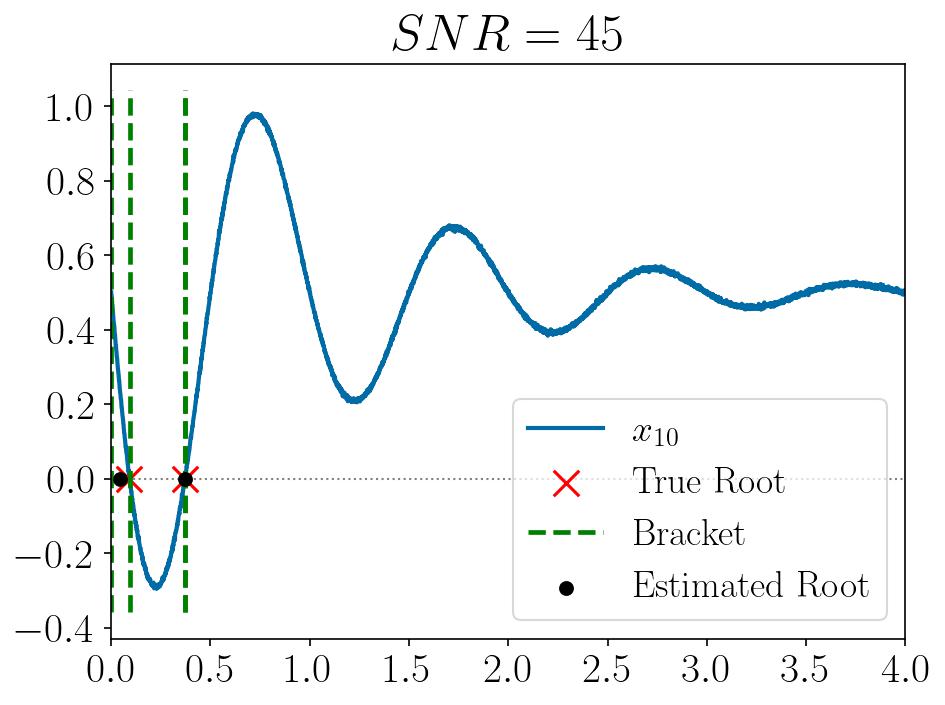}
  \caption{}
\end{subfigure}%
\begin{subfigure}{0.43\textwidth}
  \centering
  \includegraphics[width=1.\linewidth]{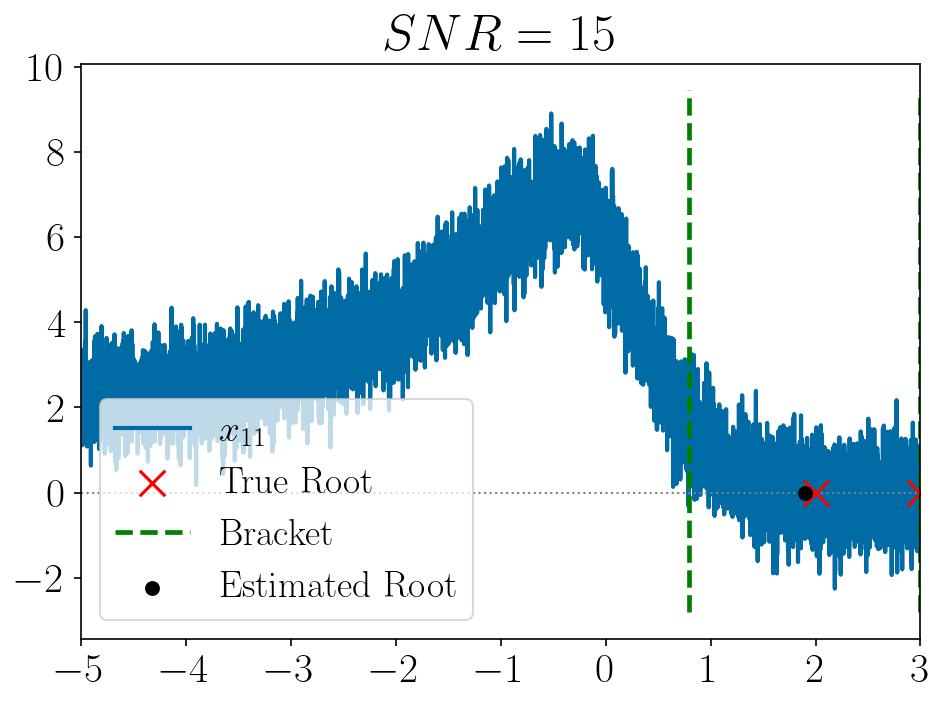}
  \caption{}
\end{subfigure}
\begin{subfigure}{0.43\textwidth}
  \centering
  \includegraphics[width=1.\linewidth]{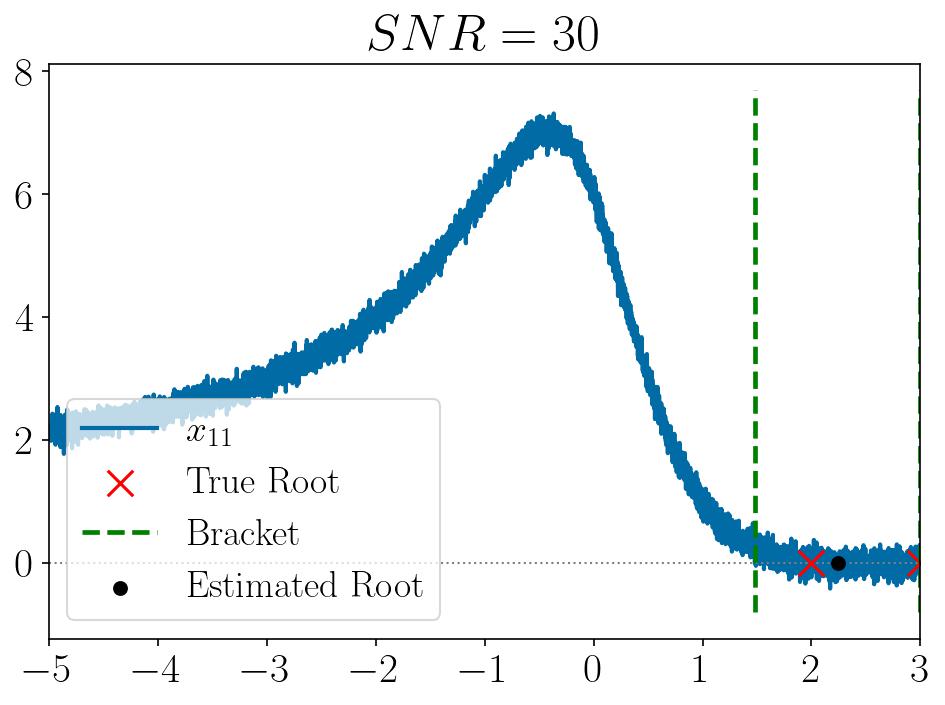}
  \caption{}
\end{subfigure}%
\begin{subfigure}{0.43\textwidth}
  \centering
  \includegraphics[width=1.\linewidth]{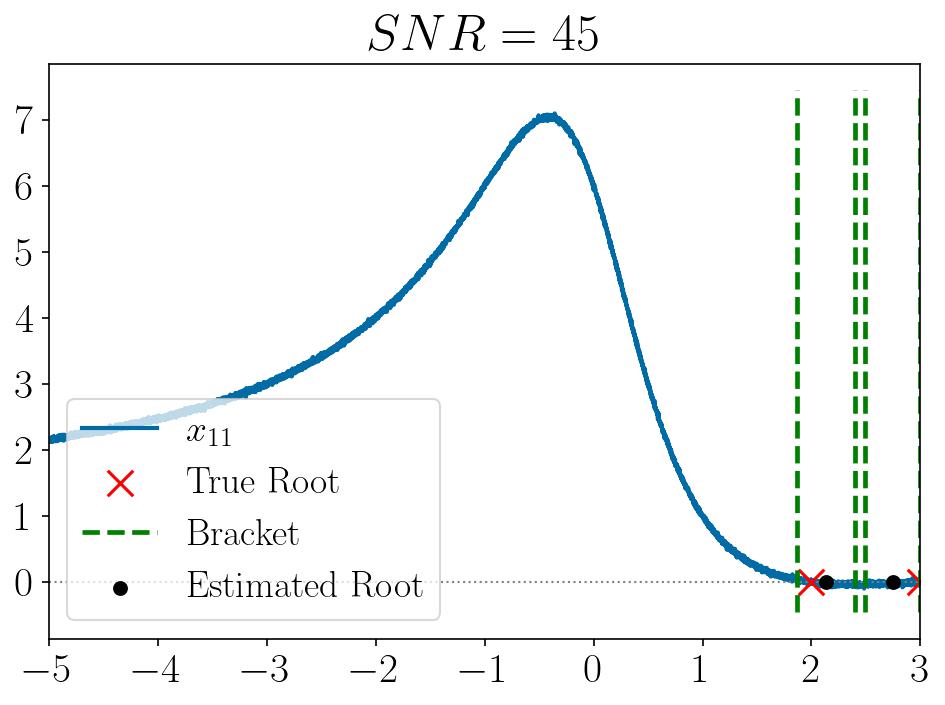}
  \caption{}
\end{subfigure}
\caption{True roots, brackets and estimated roots returned at a sampling frequency of 1000 Hz with low, medium and high SNR for functions $x_{10}$ and $x_{11}$}
\end{figure}
\begin{figure}[!htbp]
\centering
\begin{subfigure}{0.43\textwidth}
  \centering
  \includegraphics[width=1.\linewidth]{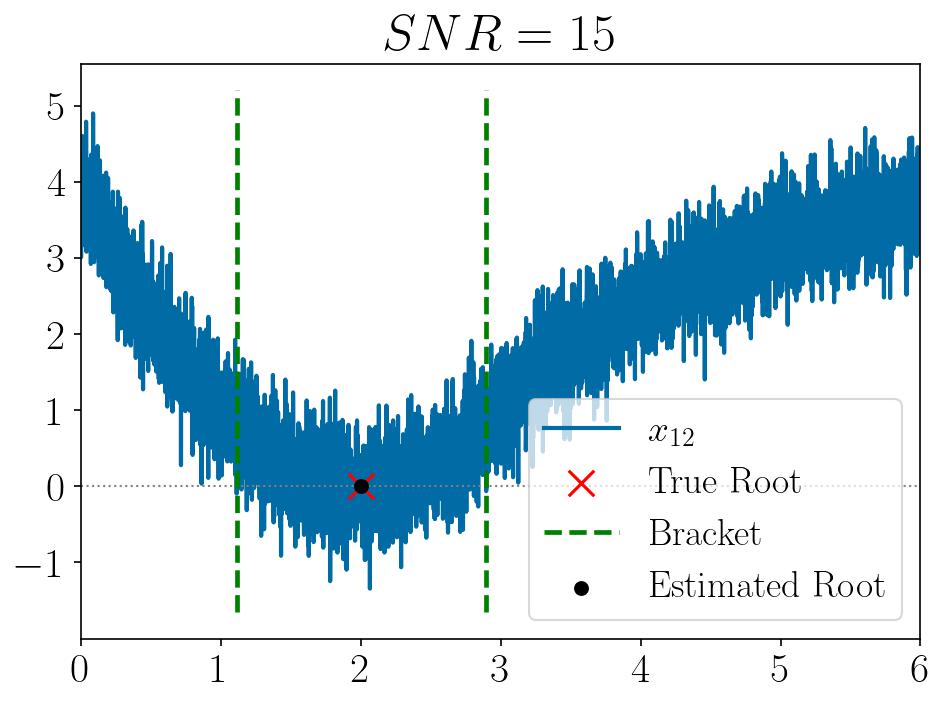}
  \caption{}
\end{subfigure}%
\begin{subfigure}{0.43\textwidth}
  \centering
  \includegraphics[width=1.\linewidth]{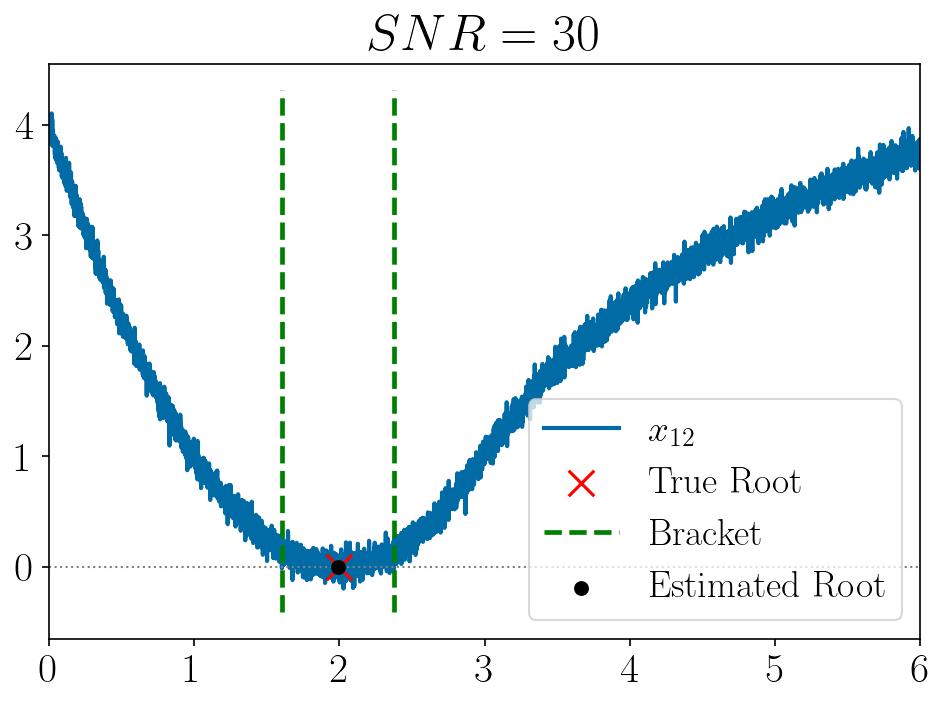}
  \caption{}
\end{subfigure}
\begin{subfigure}{0.43\textwidth}
  \centering
  \includegraphics[width=1.\linewidth]{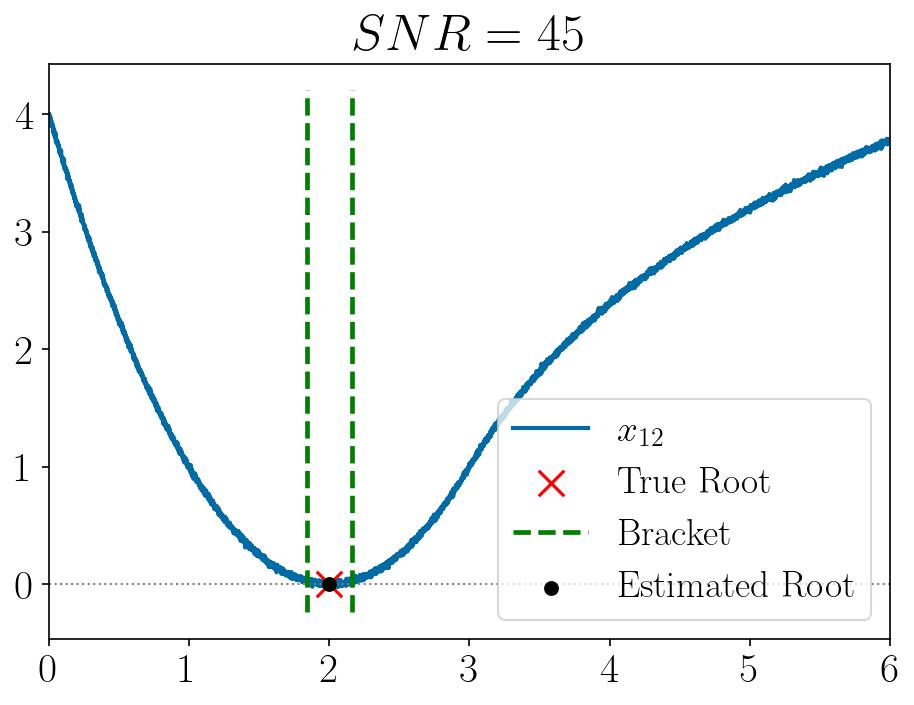}
  \caption{}
\end{subfigure}%
\begin{subfigure}{0.43\textwidth}
  \centering
  \includegraphics[width=1.\linewidth]{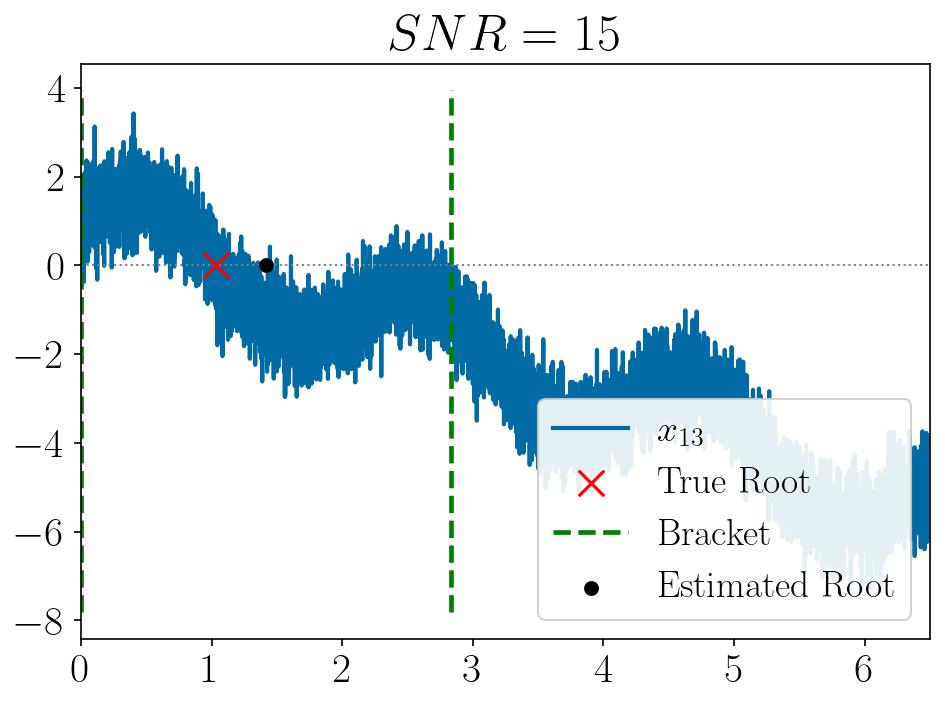}
  \caption{}
\end{subfigure}
\begin{subfigure}{0.43\textwidth}
  \centering
  \includegraphics[width=1.\linewidth]{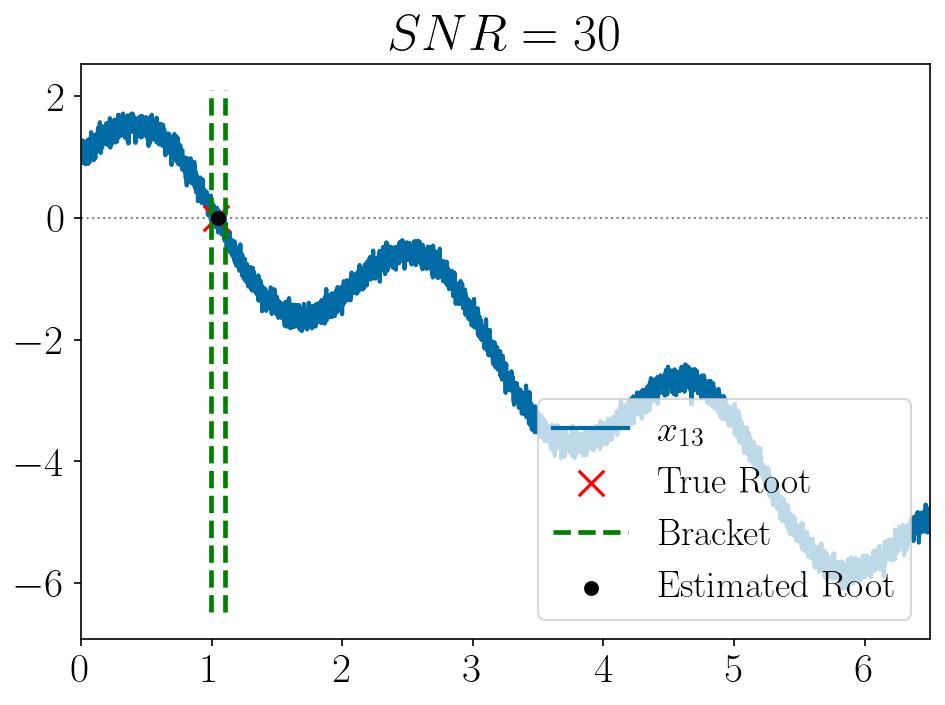}
  \caption{}
\end{subfigure}%
\begin{subfigure}{0.43\textwidth}
  \centering
  \includegraphics[width=1.\linewidth]{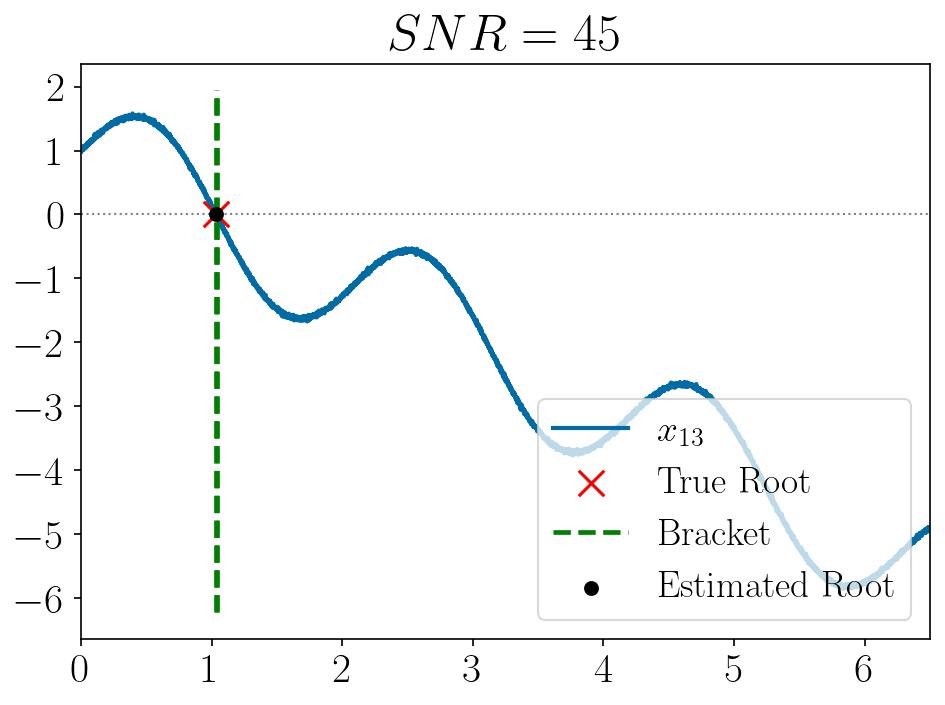}
  \caption{}
\end{subfigure}
\caption{True roots, brackets and estimated roots returned at a sampling frequency of 1000 Hz with low, medium and high SNR for functions $x_{12}$ and $x_{13}$}
\label{fig:appC_fig2}
\end{figure}

%% file: sections/sec-appendixD.tex
\section{Relative Error and Convergence Time}
\label{appendix:AppD}
Fig.~\ref{fig:appD_fig1} shows a maximum of $6\%$ relative error from 0D persistence compared to a maximum error of over $100\%$ yielded by the global optimization algorithm. Similarly, the maximum time of convergence for persistence is $0.15$ s while global optimization required over $200$ s. The same disparity can be seen in heat maps of all functions in Figs.~\ref{fig:appD_fig2} to \ref{fig:appD_fig13} proving the higher efficiency of 0D persistence algorithm.
\begin{figure}[!htbp]
\centering
\begin{subfigure}{0.45\textwidth}
  \centering
  \includegraphics[width=1.\linewidth]{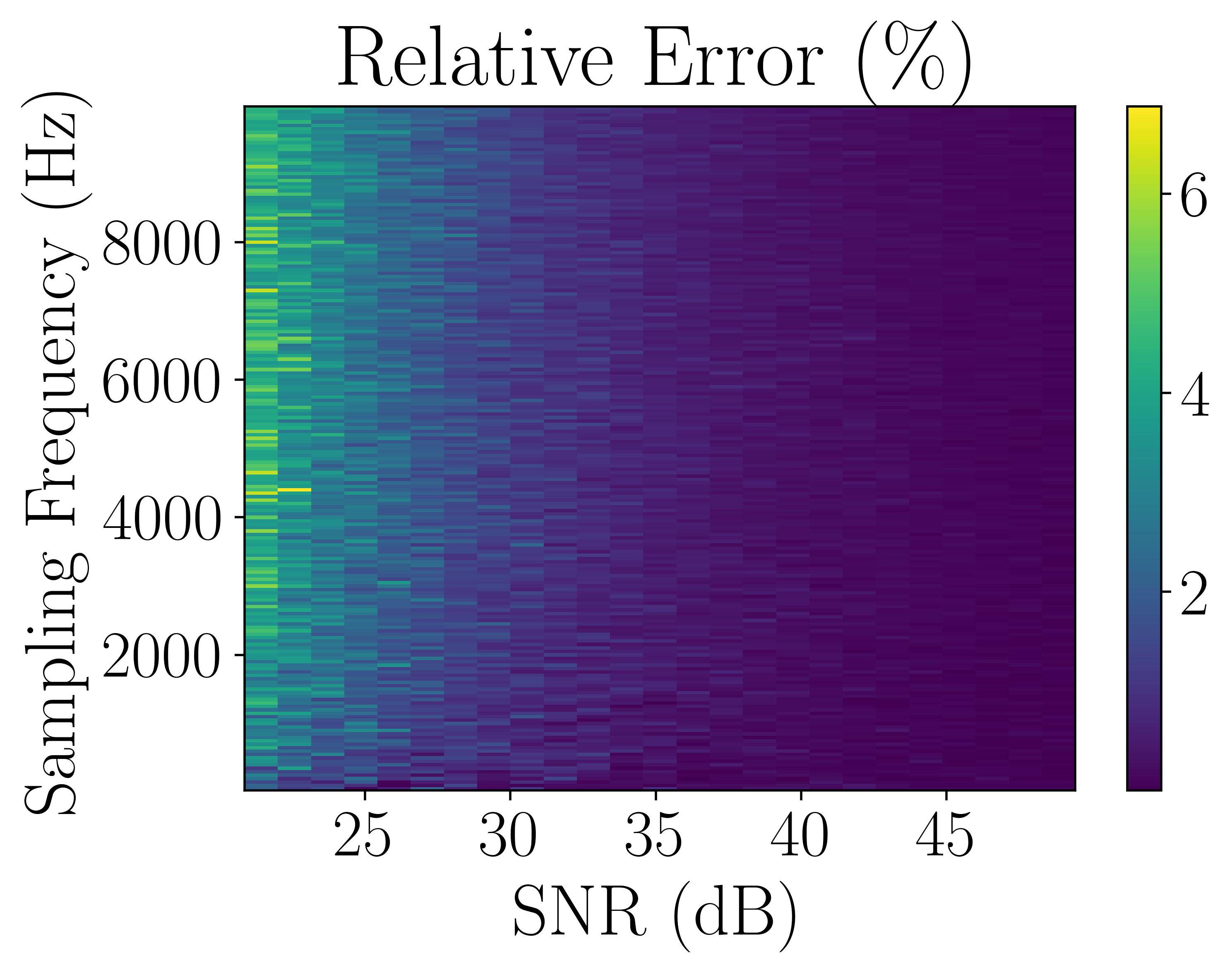}
    \caption{}
\end{subfigure}%
\begin{subfigure}{0.45\textwidth}
  \centering
  \includegraphics[width=1.\linewidth]{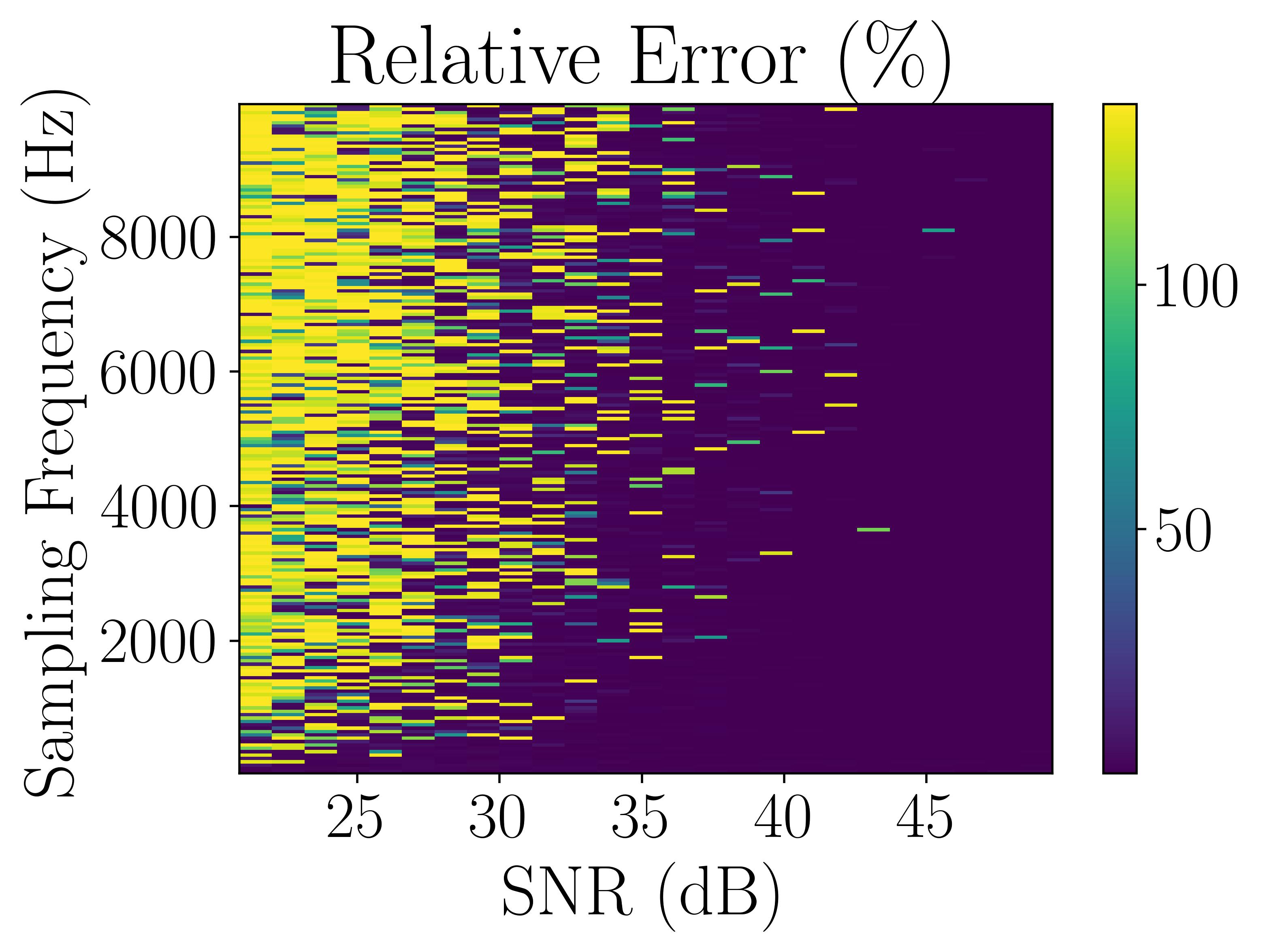}
      \caption{}
\end{subfigure}
\begin{subfigure}{0.45\textwidth}
  \centering
  \includegraphics[width=1.\linewidth]{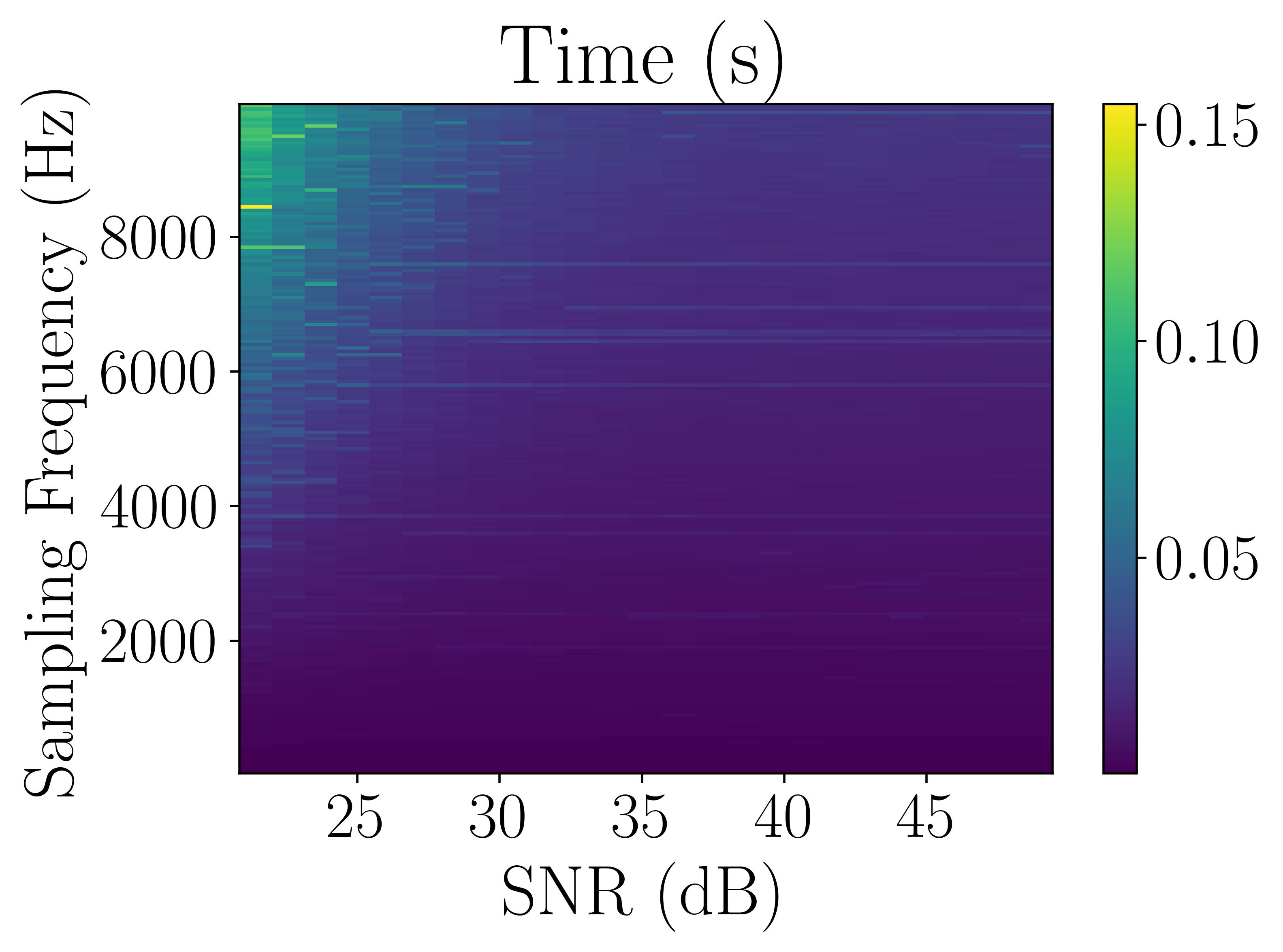}
      \caption{}
\end{subfigure}%
\begin{subfigure}{0.45\textwidth}
  \centering
  \includegraphics[width=1.\linewidth]{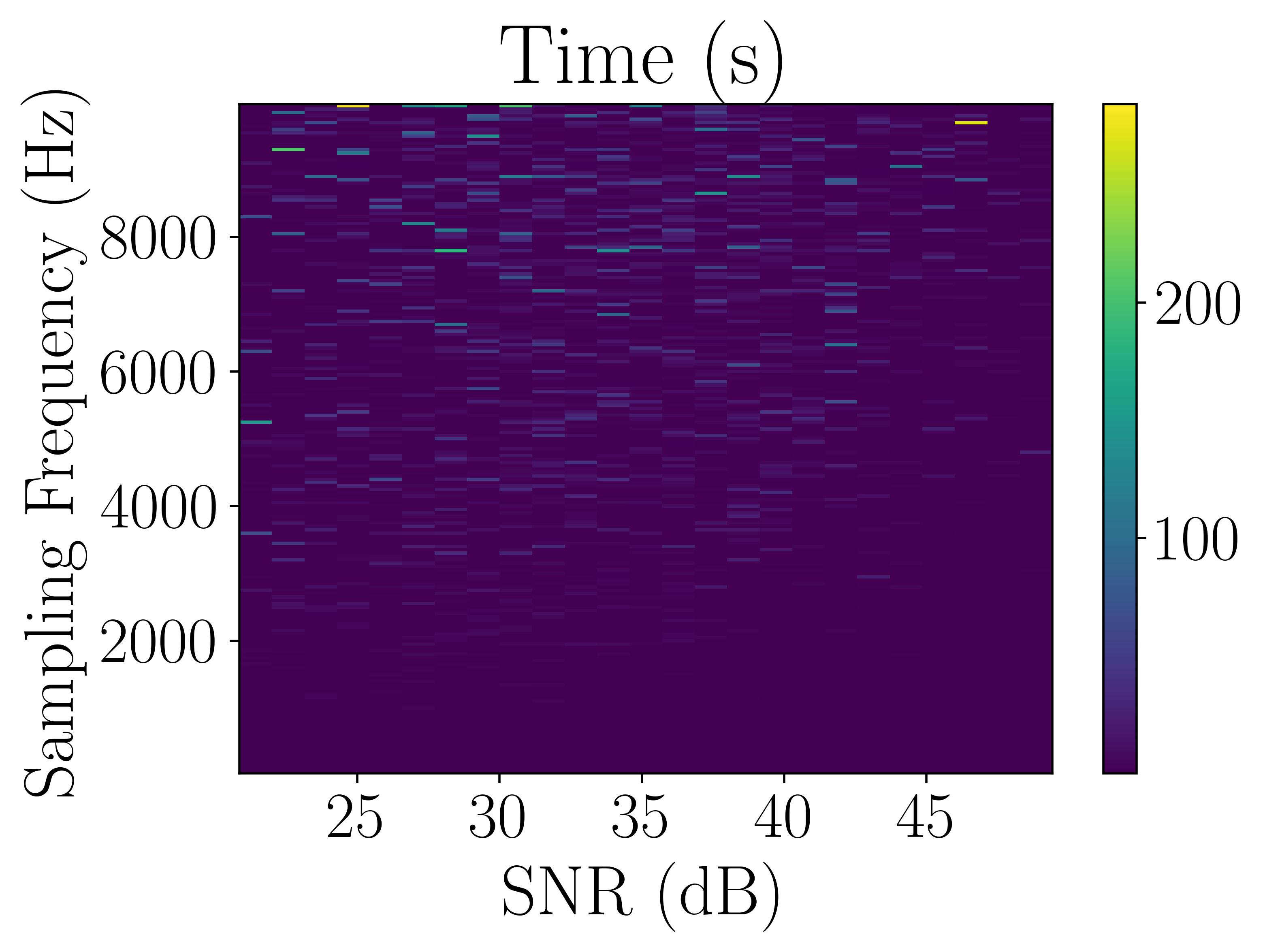}
      \caption{}
\end{subfigure}
\caption{Relative error and time taken for convergence by 0D Persistence (a, c) and Molinaro's algorithm (b, d) for $x_{1}$}
\label{fig:appD_fig1}
\end{figure}
\begin{figure}[!htbp]
\centering
\begin{subfigure}{0.45\textwidth}
  \centering
  \includegraphics[width=1.\linewidth]{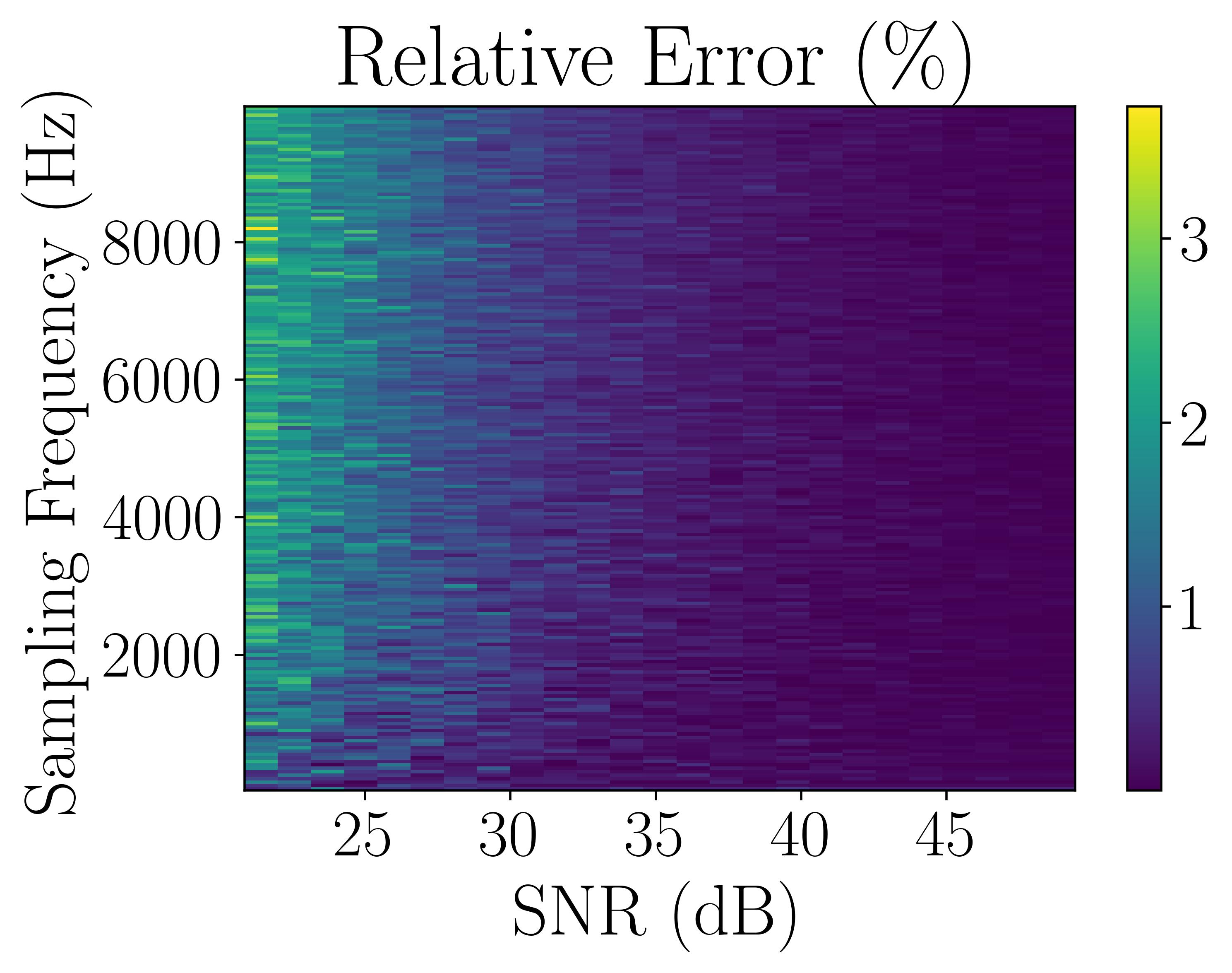}
    \caption{}
\end{subfigure}%
\begin{subfigure}{0.45\textwidth}
  \centering
  \includegraphics[width=1.\linewidth]{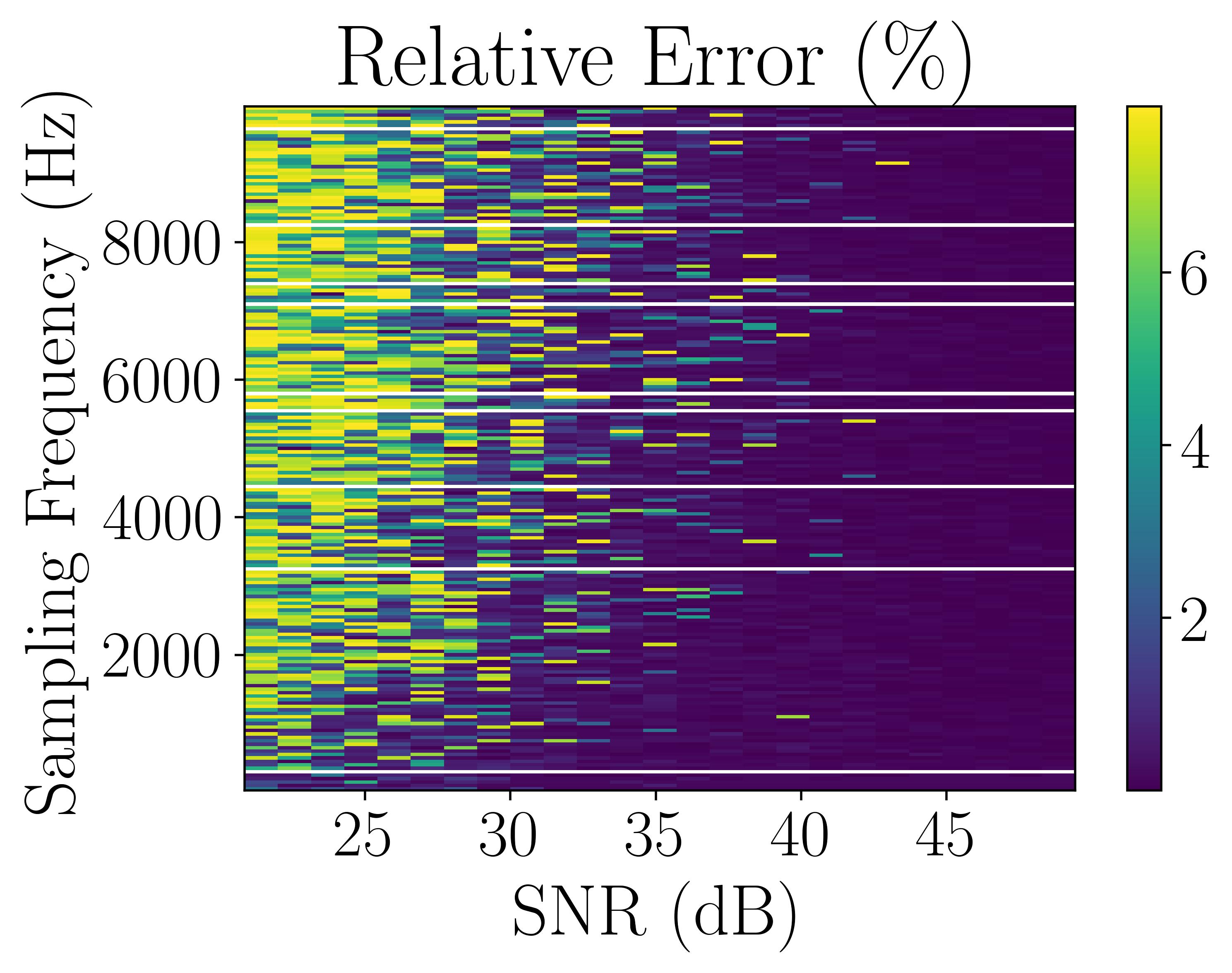}
      \caption{}
\end{subfigure}
\begin{subfigure}{0.45\textwidth}
  \centering
  \includegraphics[width=1.\linewidth]{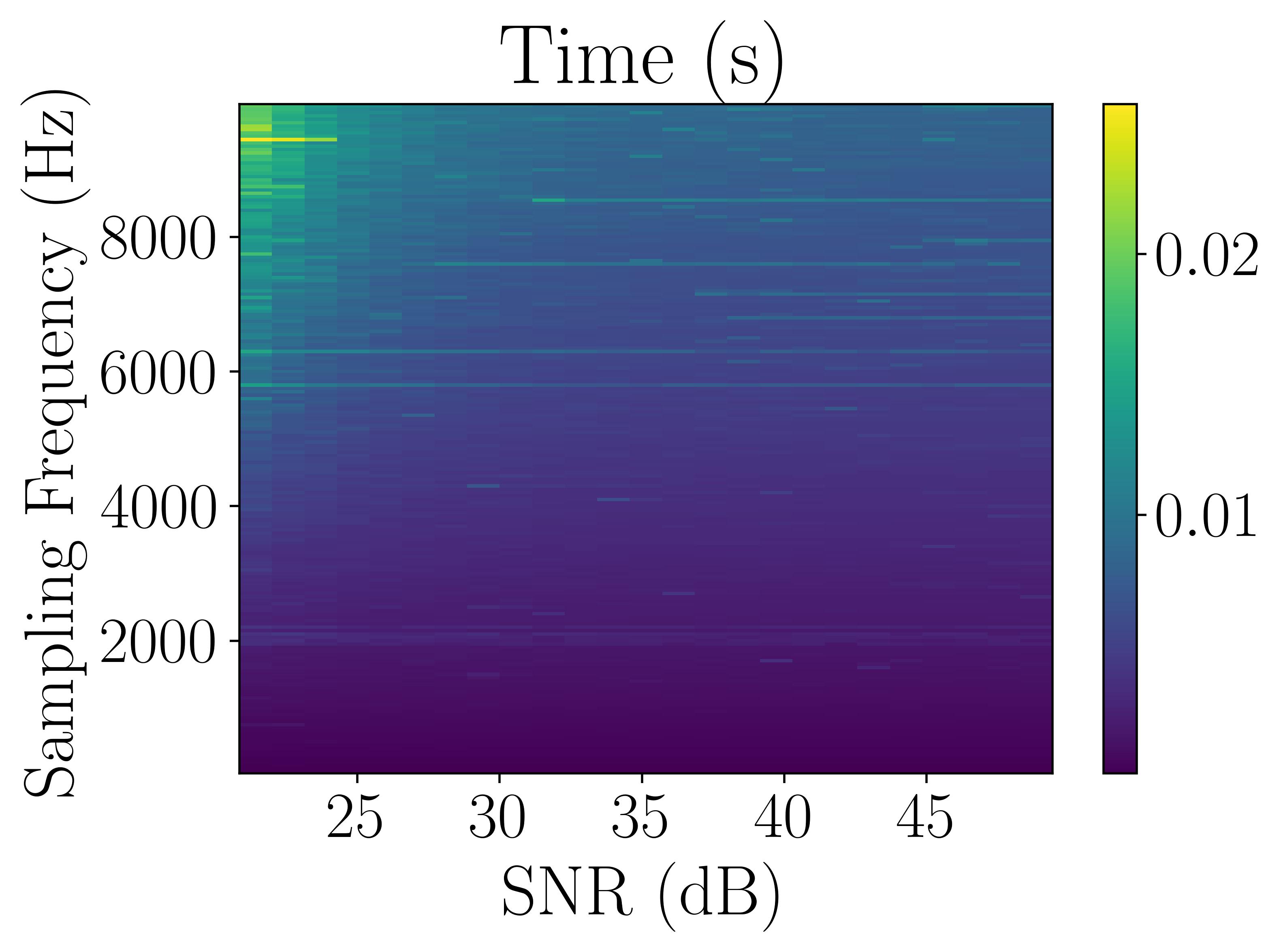}
      \caption{}
\end{subfigure}%
\begin{subfigure}{0.45\textwidth}
  \centering
  \includegraphics[width=1.\linewidth]{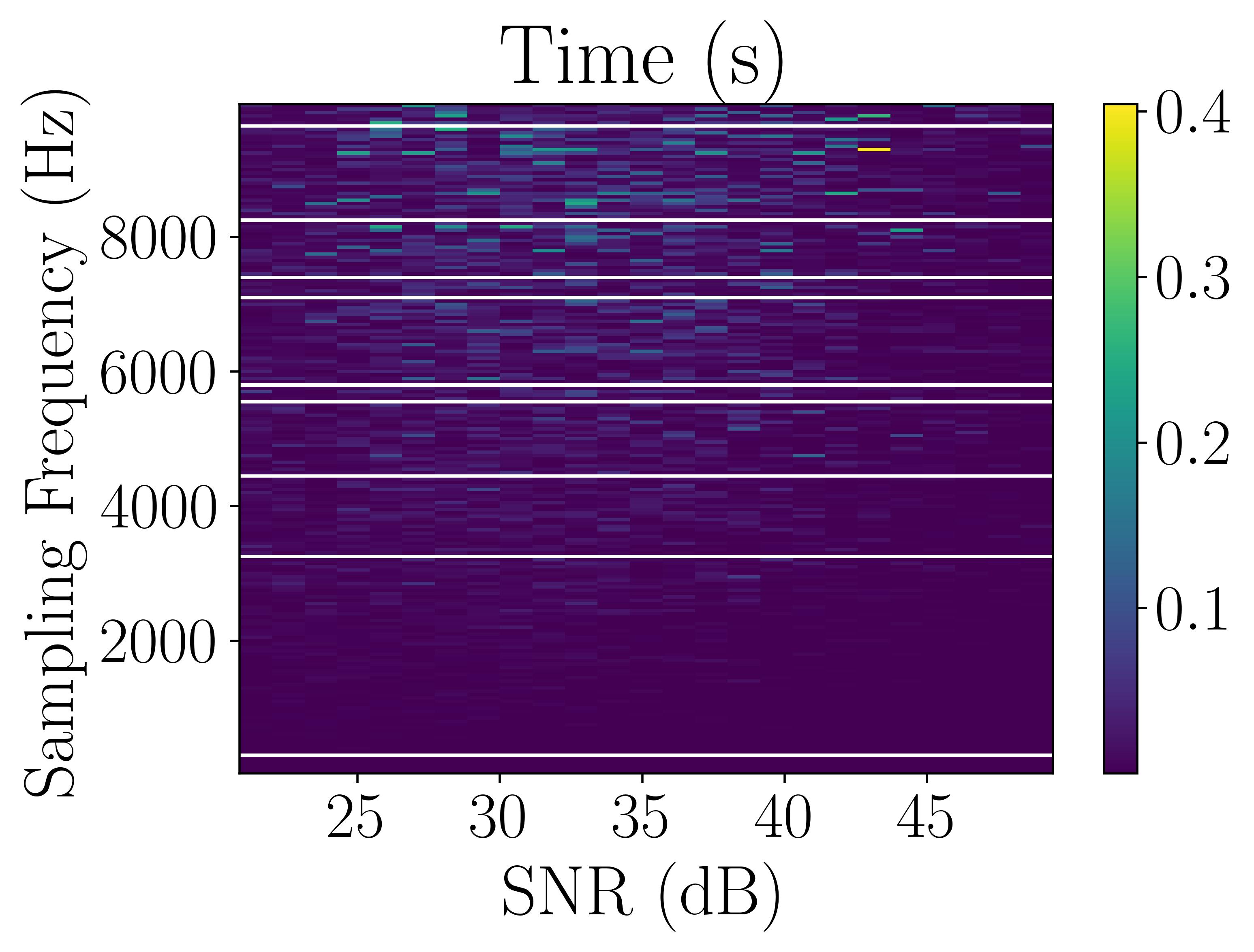}
      \caption{}
\end{subfigure}
\caption{Relative error and time taken for convergence by 0D Persistence (a, c) and Molinaro's algorithm (b, d) for $x_{3}$}
\label{fig:appD_fig2}
\end{figure}
\begin{figure}[!htbp]
\centering
\begin{subfigure}{0.45\textwidth}
  \centering
  \includegraphics[width=1.\linewidth]{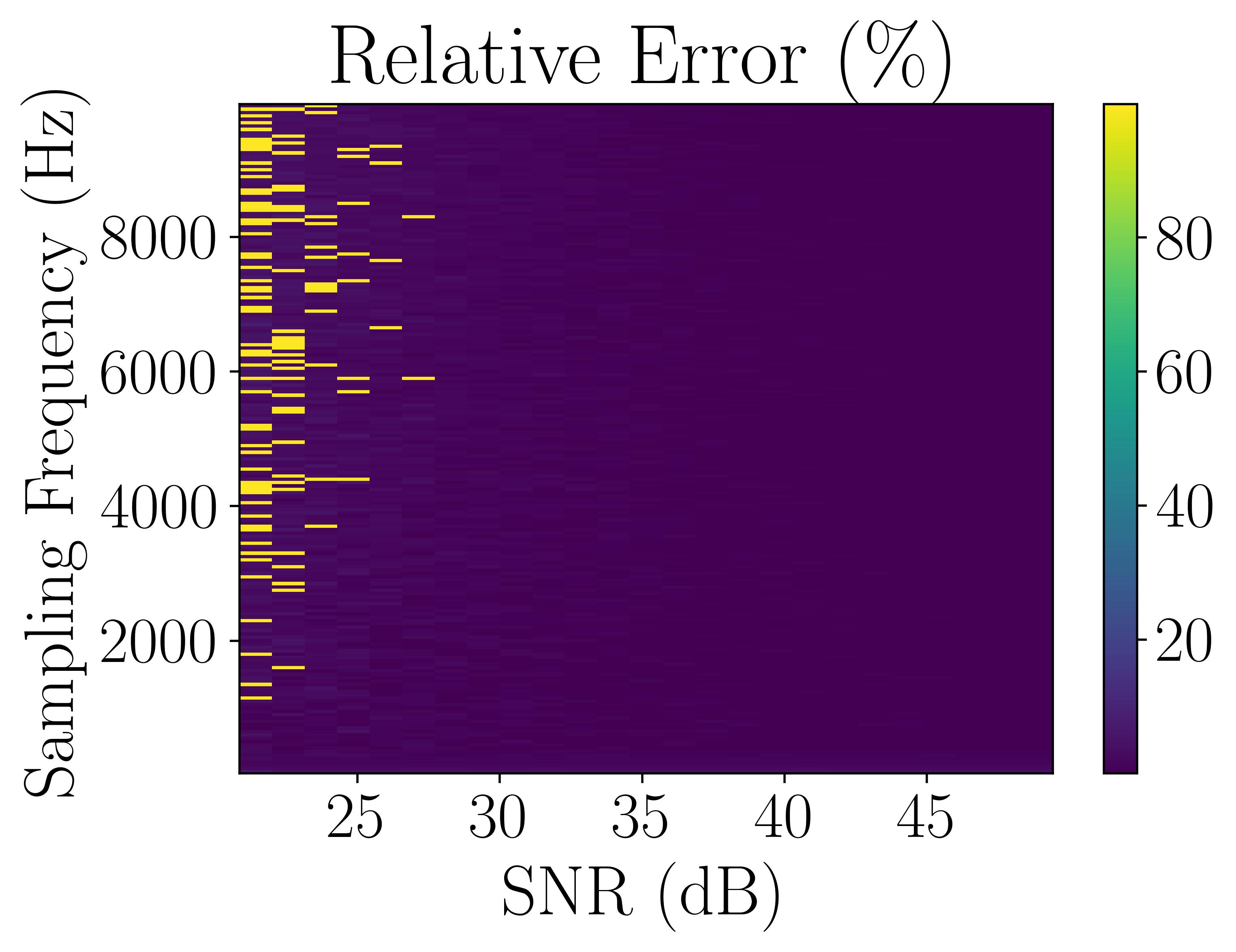}
    \caption{}
\end{subfigure}%
\begin{subfigure}{0.45\textwidth}
  \centering
  \includegraphics[width=1.\linewidth]{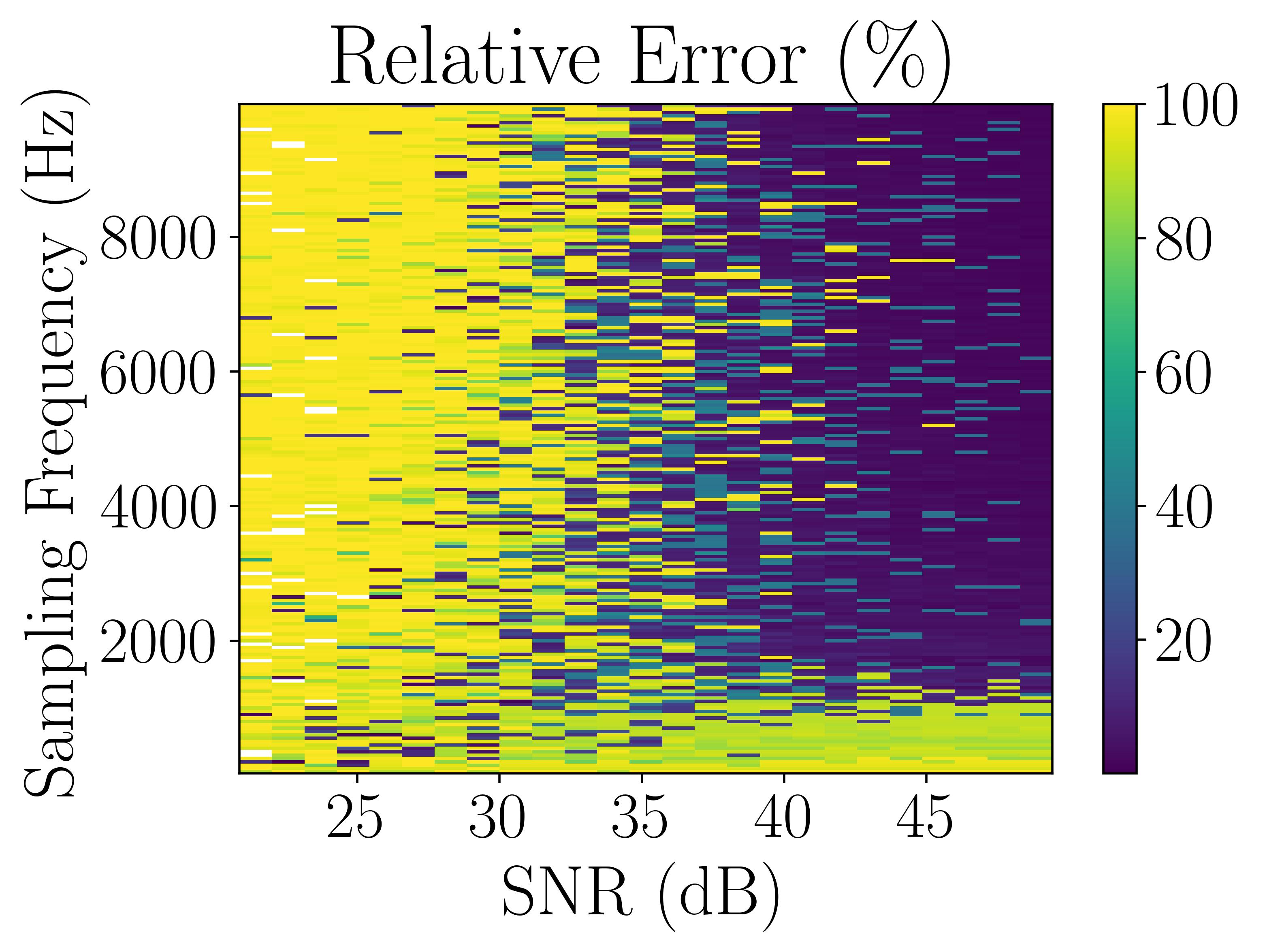}
      \caption{}
\end{subfigure}
\begin{subfigure}{0.45\textwidth}
  \centering
  \includegraphics[width=1.\linewidth]{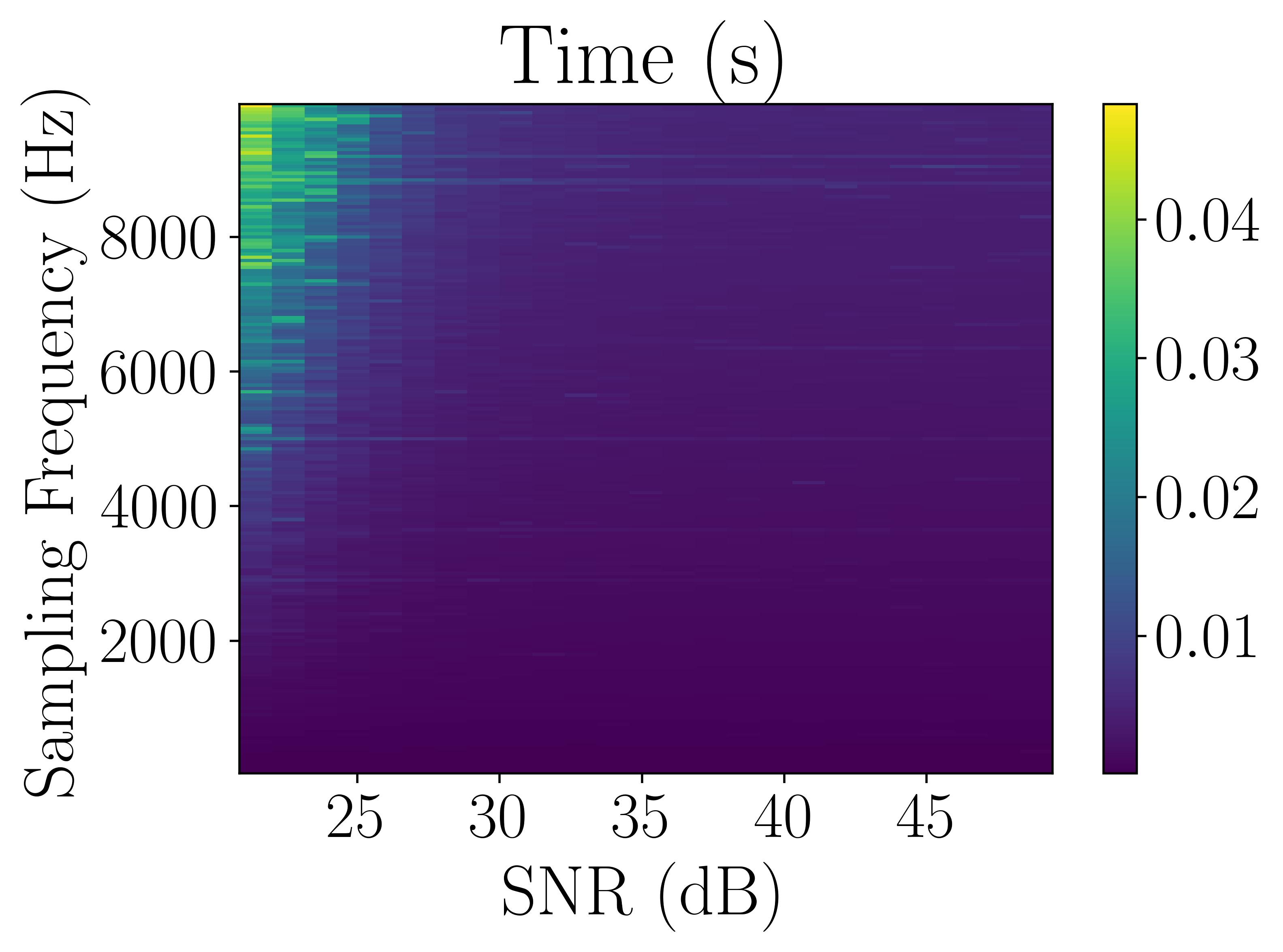}
      \caption{}
\end{subfigure}%
\begin{subfigure}{0.45\textwidth}
  \centering
  \includegraphics[width=1.\linewidth]{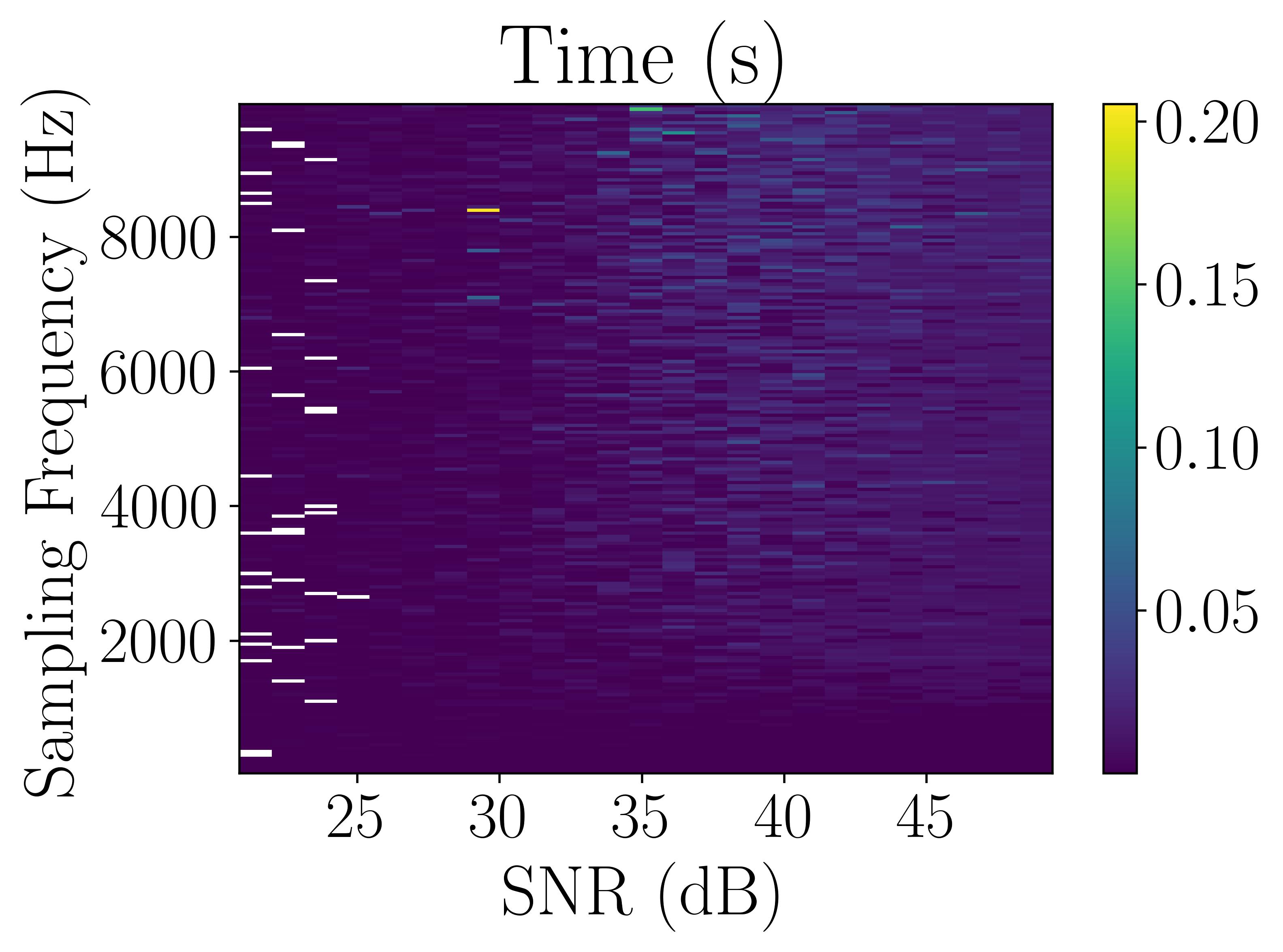}
      \caption{}
\end{subfigure}
\caption{Relative error and time taken for convergence by 0D Persistence (a, c) and Molinaro's algorithm (b, d) for $x_{4}$}
\label{fig:appD_fig3}
\end{figure}
\begin{figure}[!htbp]
\centering
\begin{subfigure}{0.45\textwidth}
  \centering
  \includegraphics[width=1.\linewidth]{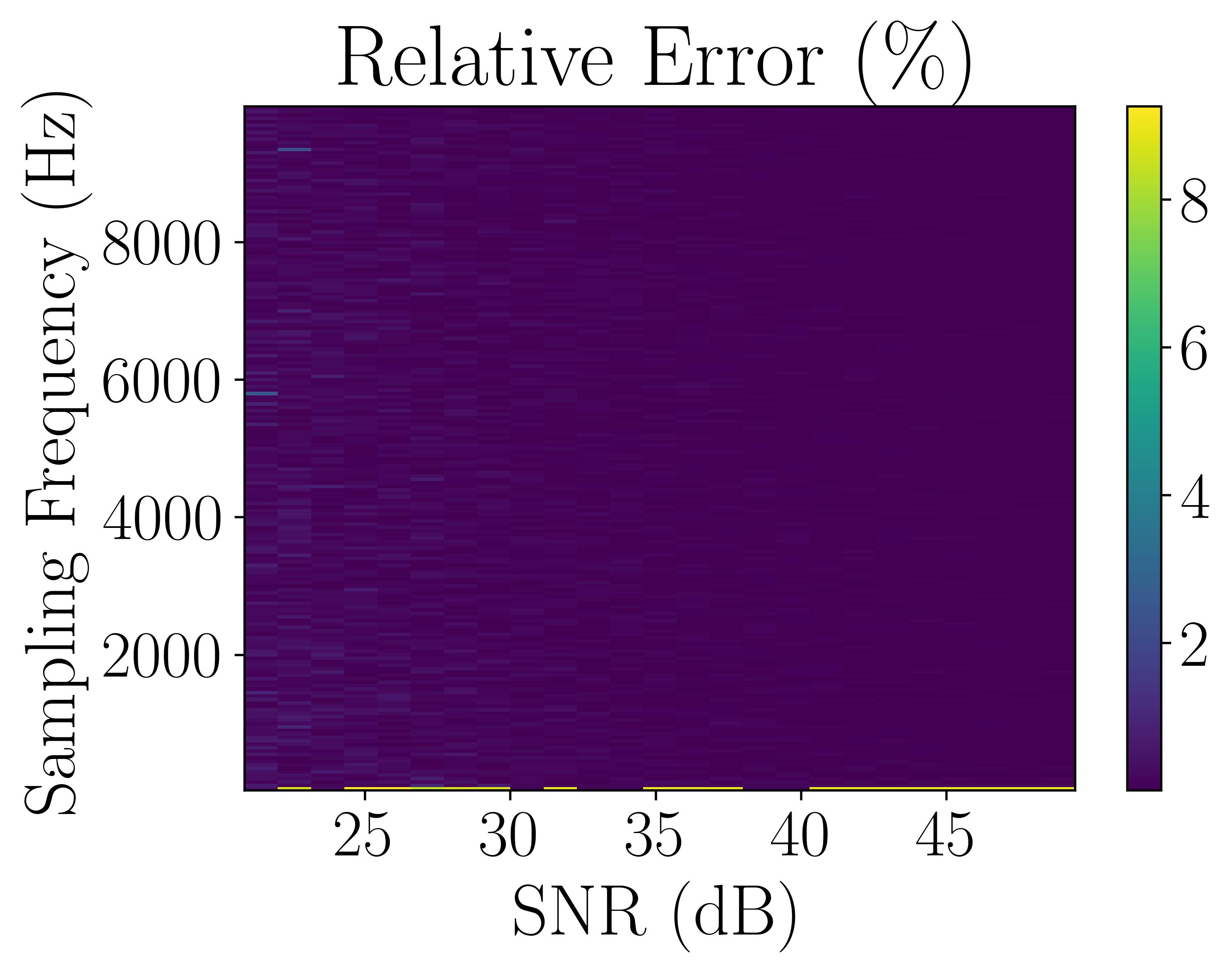}
    \caption{}
\end{subfigure}%
\begin{subfigure}{0.45\textwidth}
  \centering
  \includegraphics[width=1.\linewidth]{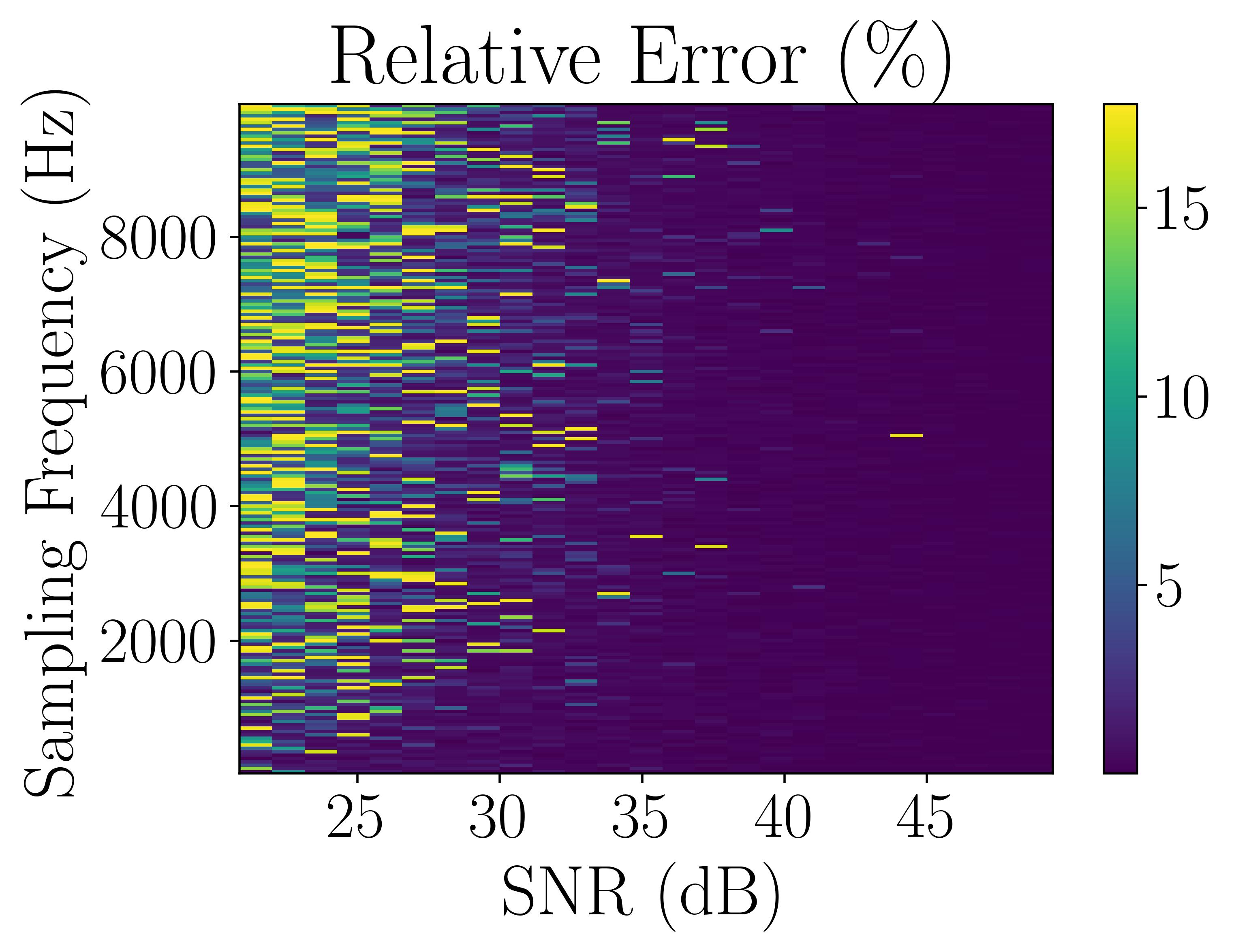}
      \caption{}
\end{subfigure}
\begin{subfigure}{0.45\textwidth}
  \centering
  \includegraphics[width=1.\linewidth]{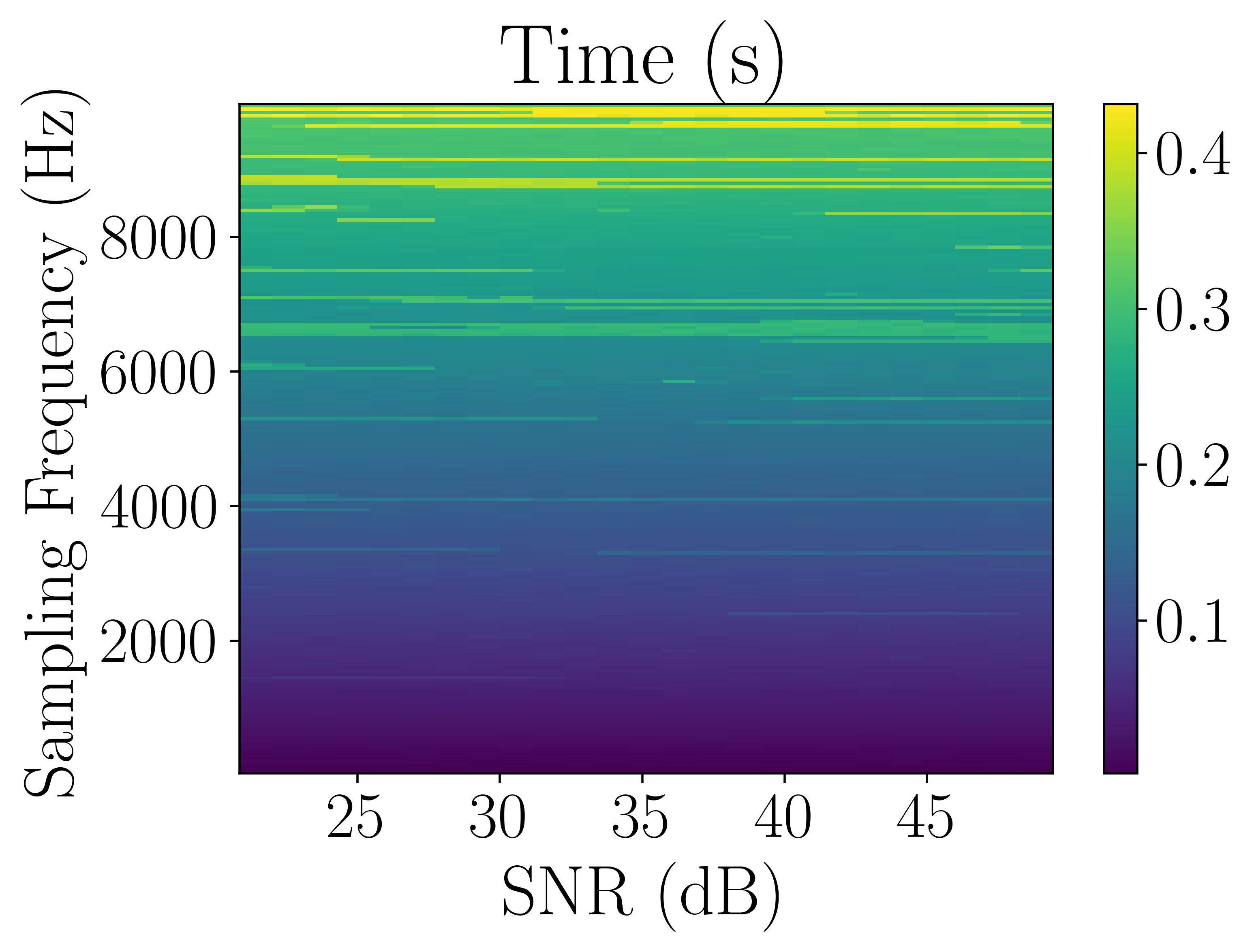}
      \caption{}
\end{subfigure}%
\begin{subfigure}{0.45\textwidth}
  \centering
  \includegraphics[width=1.\linewidth]{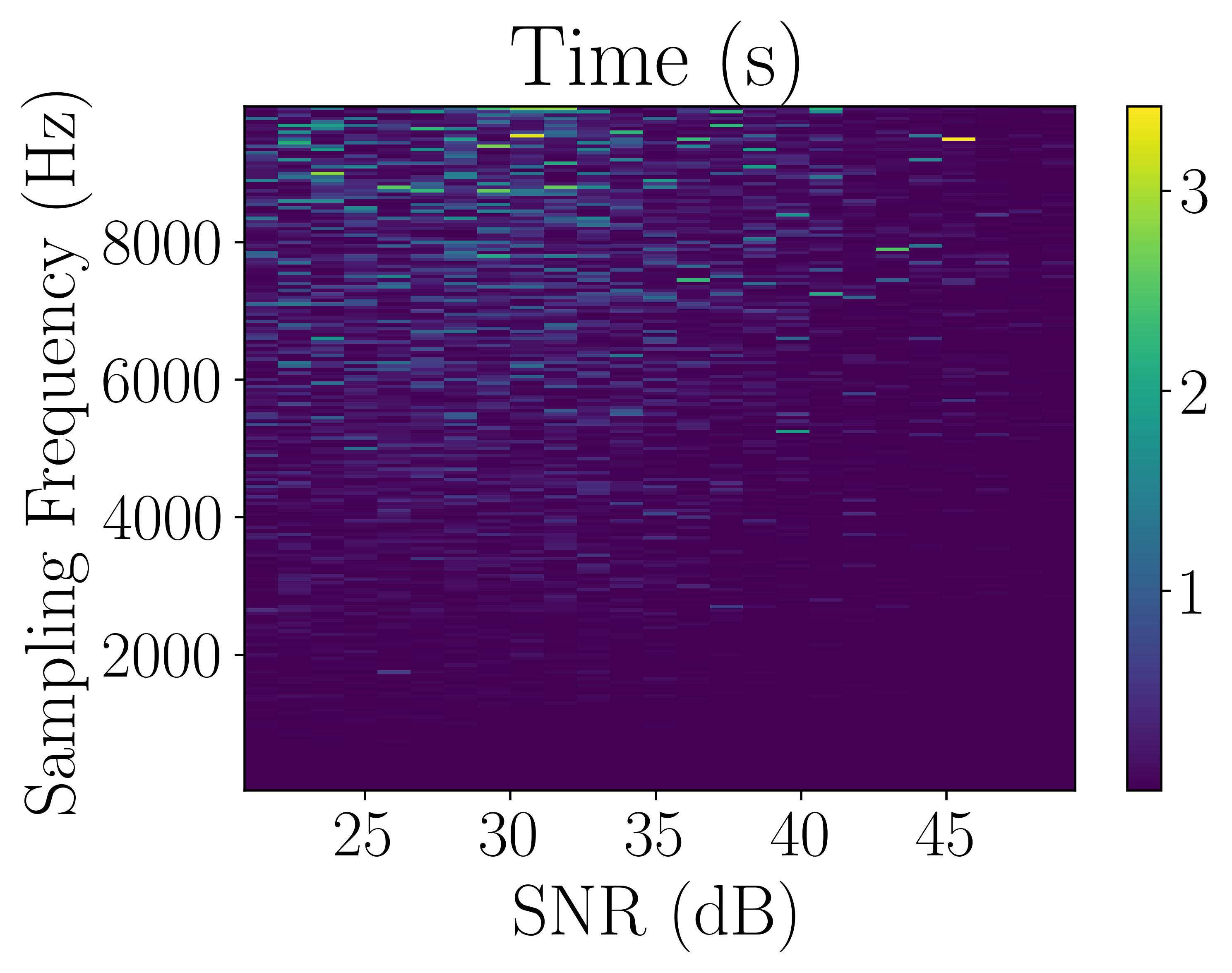}
      \caption{}
\end{subfigure}
\caption{Relative error and time taken for convergence by 0D Persistence (a, c) and Molinaro's algorithm (b, d) for $x_{5}$}
\label{fig:appD_fig4}
\end{figure}
\begin{figure}[!htbp]
\centering
\begin{subfigure}{0.45\textwidth}
  \centering
  \includegraphics[width=1.\linewidth]{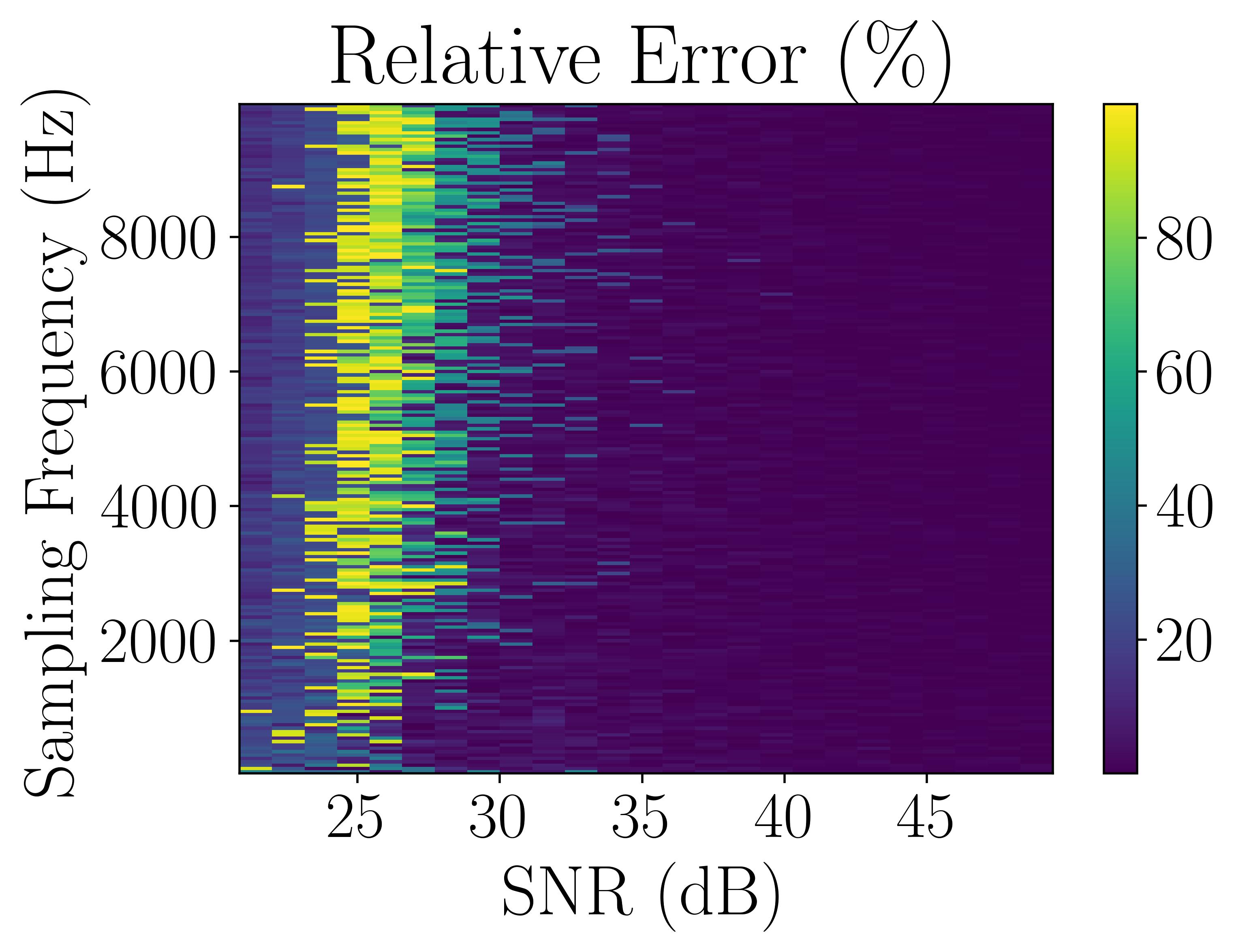}
    \caption{}
\end{subfigure}%
\begin{subfigure}{0.45\textwidth}
  \centering
  \includegraphics[width=1.\linewidth]{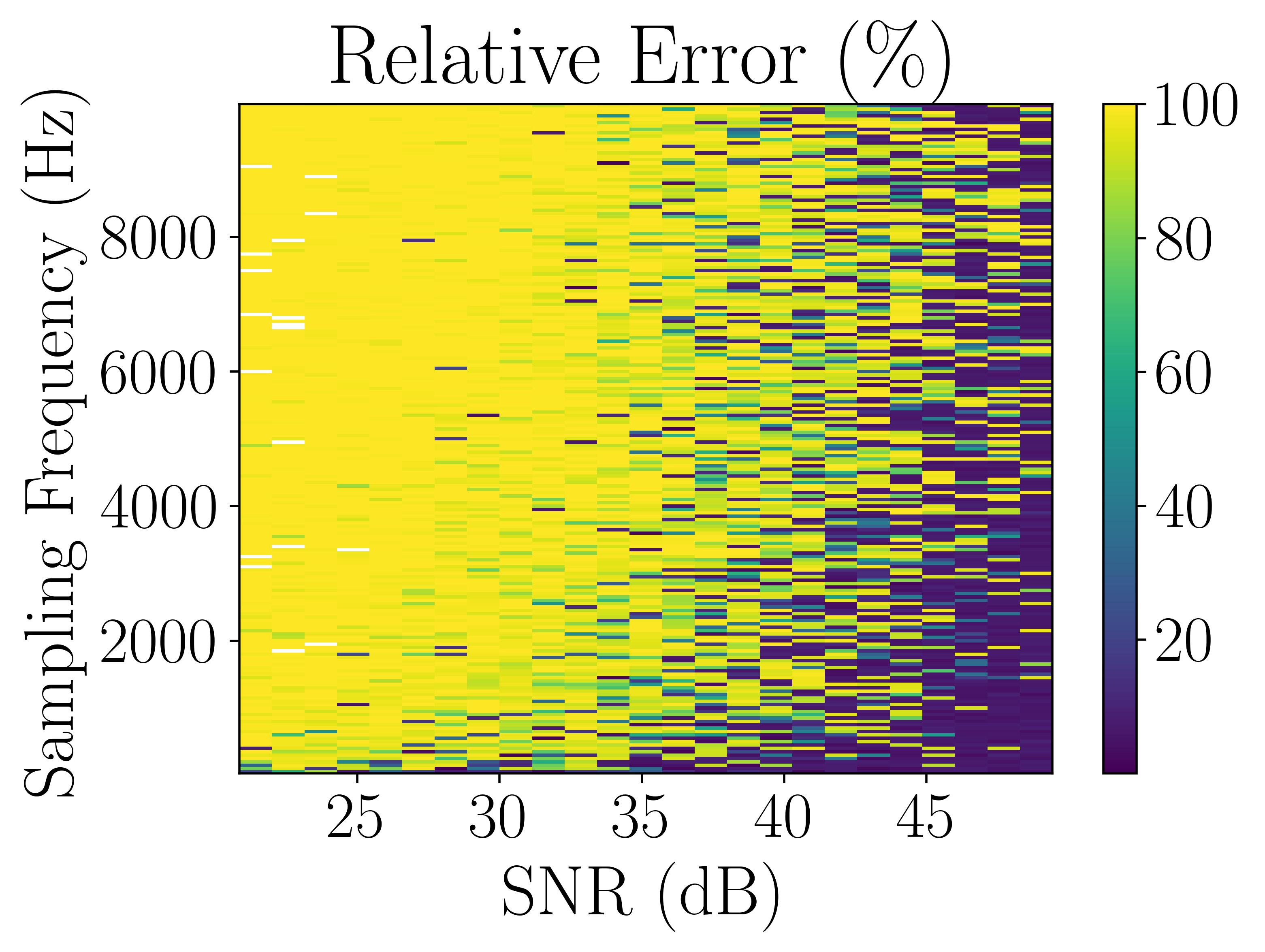}
      \caption{}
\end{subfigure}
\begin{subfigure}{0.45\textwidth}
  \centering
  \includegraphics[width=1.\linewidth]{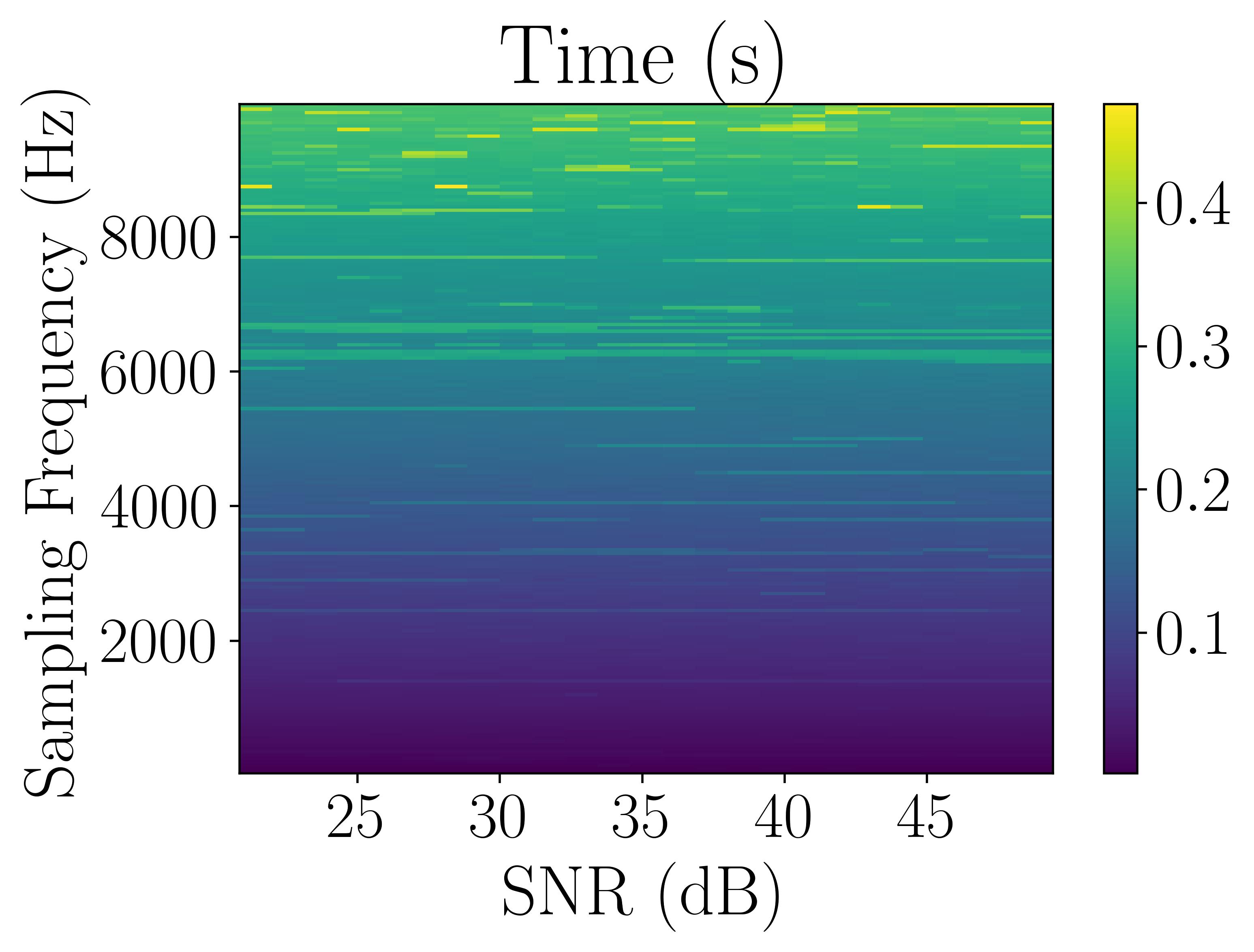}
      \caption{}
\end{subfigure}%
\begin{subfigure}{0.45\textwidth}
  \centering
  \includegraphics[width=1.\linewidth]{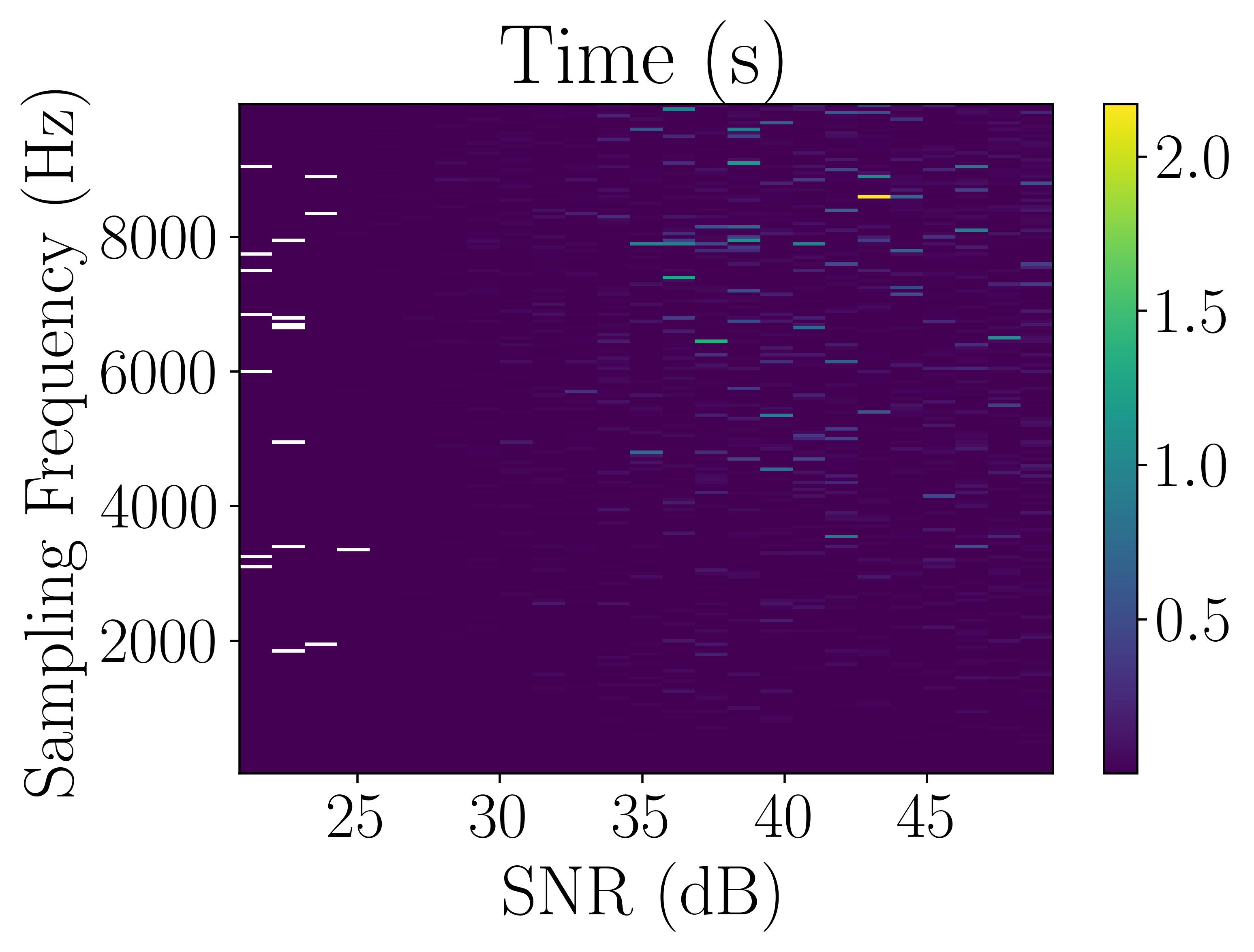}
      \caption{}
\end{subfigure}
\caption{Relative error and time taken for convergence by 0D Persistence (a, c) and Molinaro's algorithm (b, d) for $x_{6}$}
\label{fig:appD_fig5}
\end{figure}
\begin{figure}[!htbp]
\centering
\begin{subfigure}{0.45\textwidth}
  \centering
  \includegraphics[width=1.\linewidth]{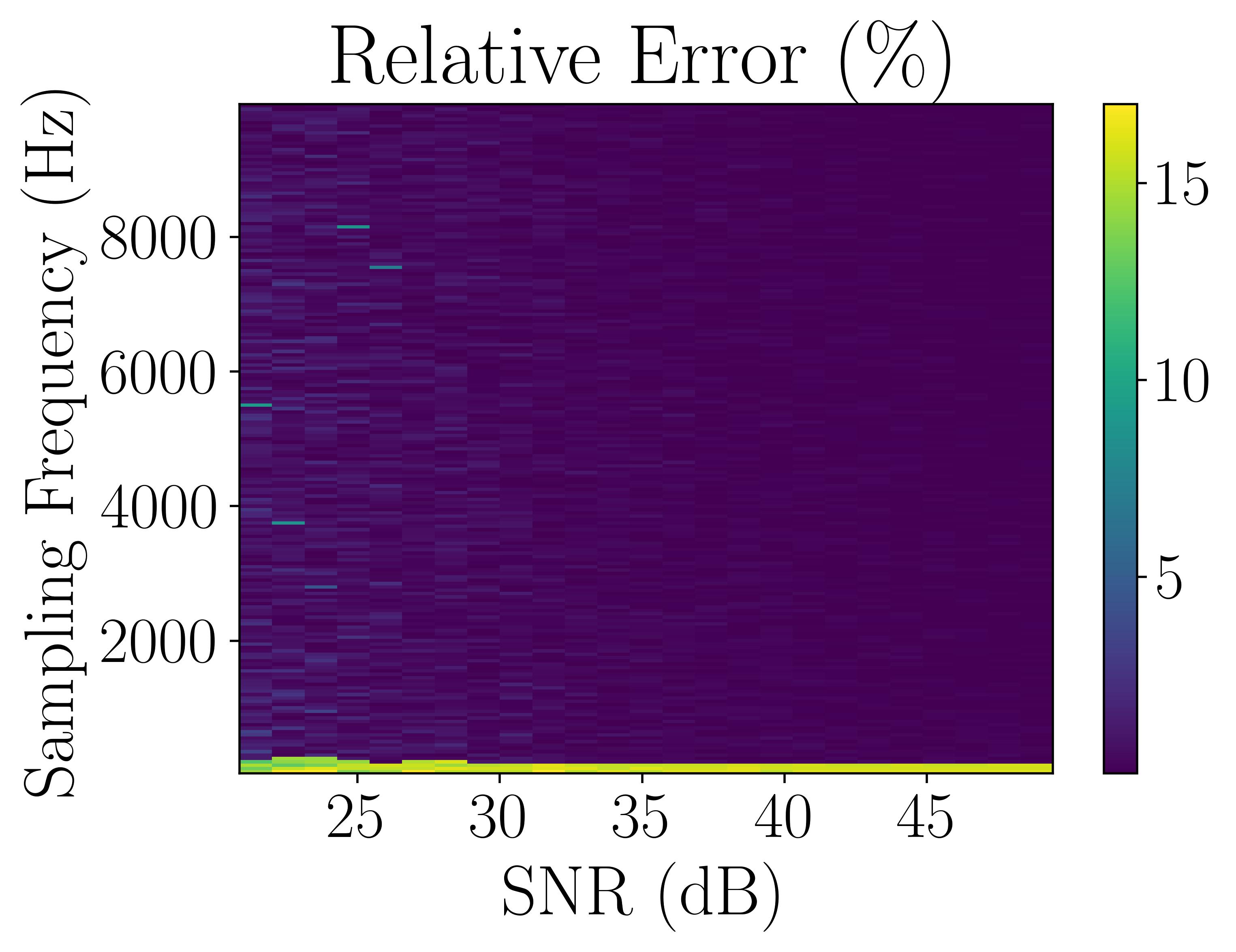}
    \caption{}
\end{subfigure}%
\begin{subfigure}{0.45\textwidth}
  \centering
  \includegraphics[width=1.\linewidth]{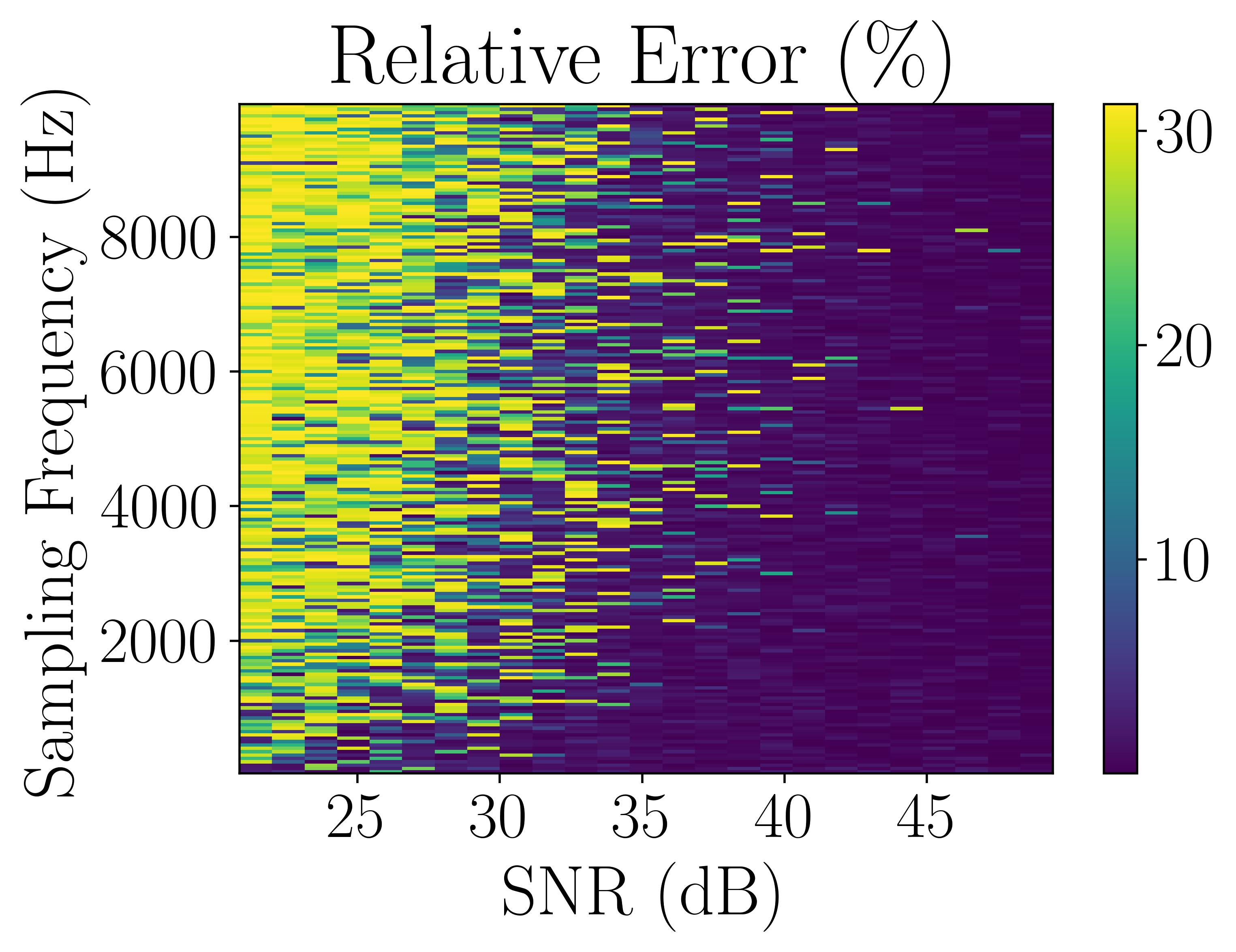}
      \caption{}
\end{subfigure}
\begin{subfigure}{0.45\textwidth}
  \centering
  \includegraphics[width=1.\linewidth]{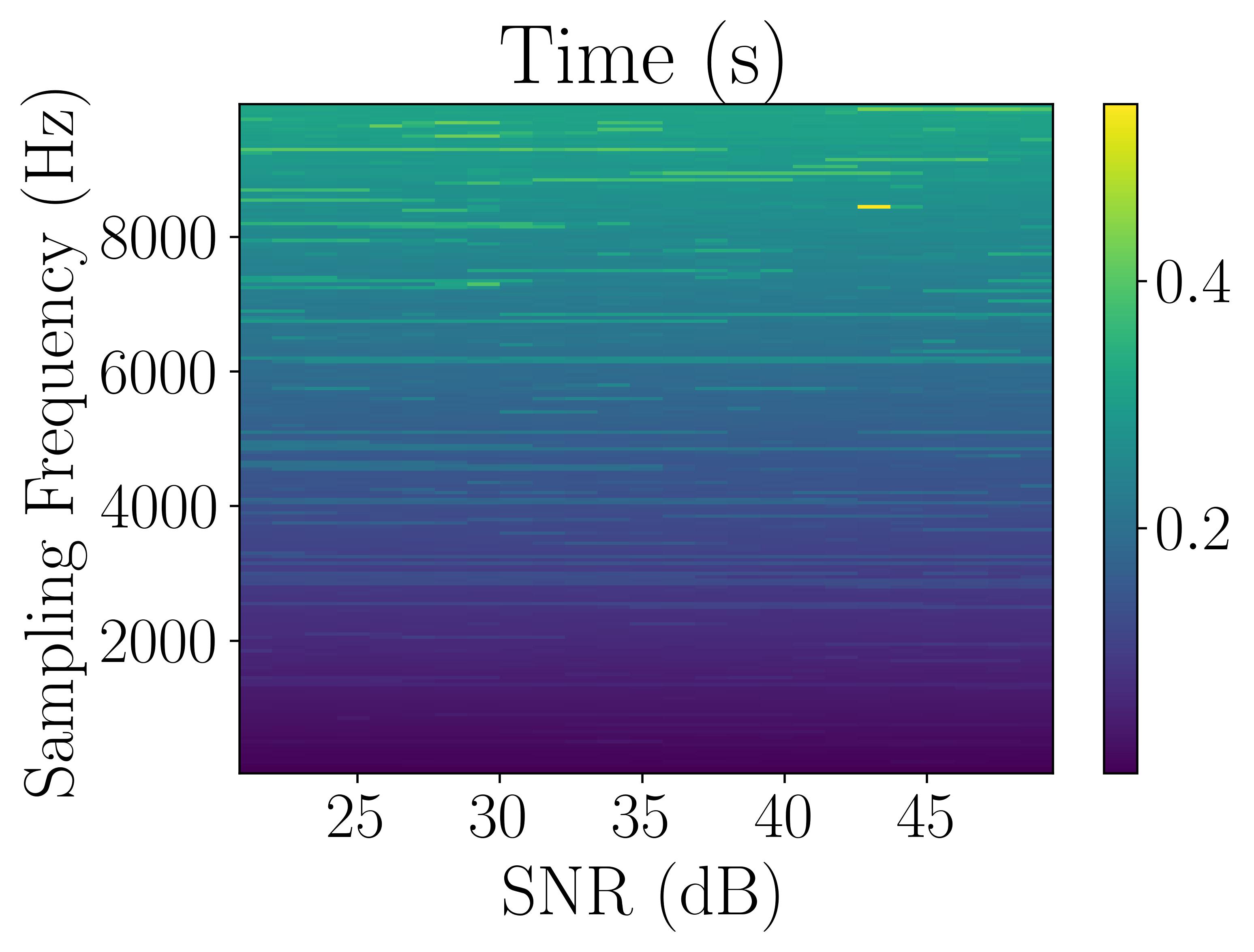}
      \caption{}
\end{subfigure}%
\begin{subfigure}{0.45\textwidth}
  \centering
  \includegraphics[width=1.\linewidth]{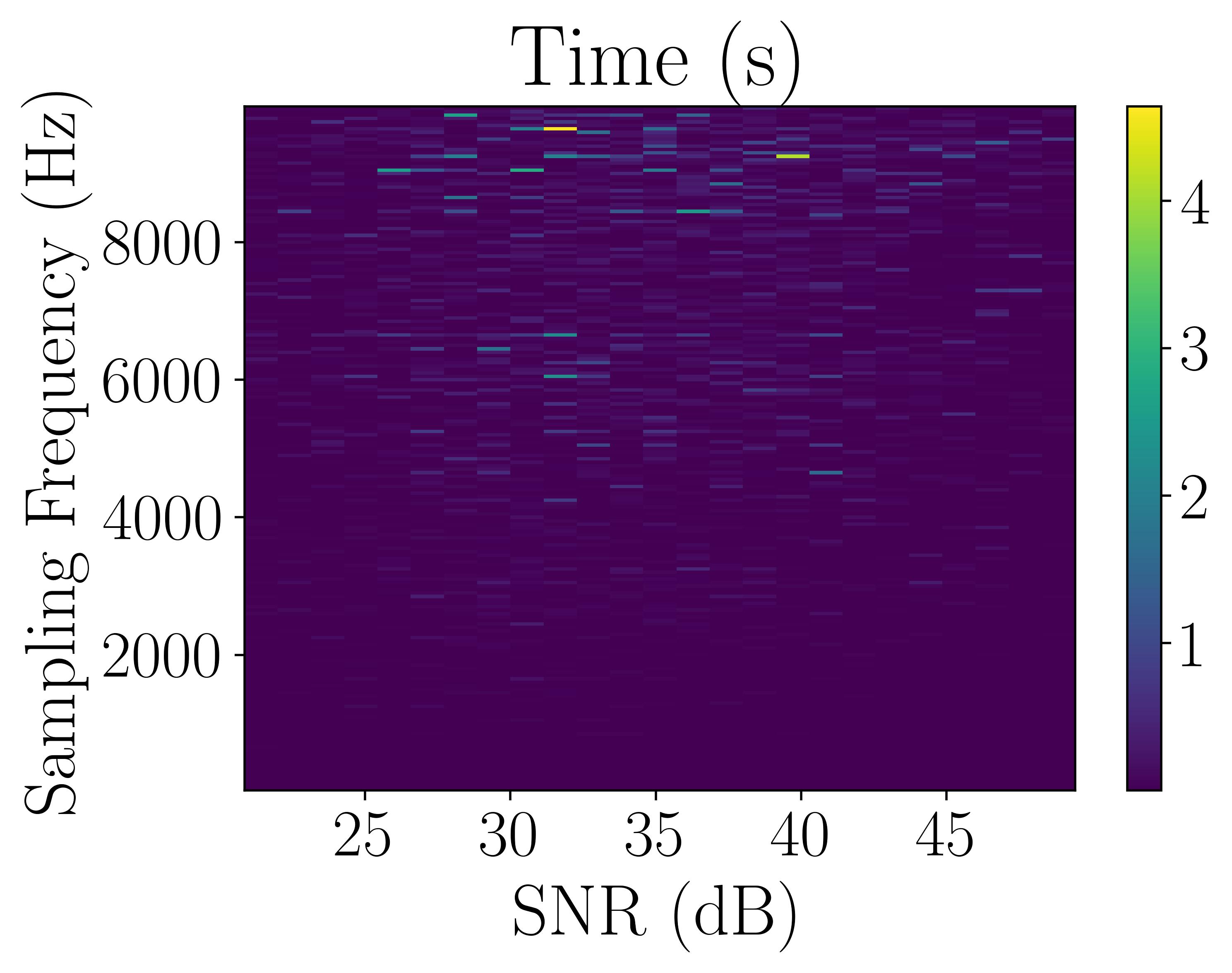}
      \caption{}
\end{subfigure}
\caption{Relative error and time taken for convergence by 0D Persistence (a, c) and Molinaro's algorithm (b, d) for $x_{7}$}
\label{fig:appD_fig6}
\end{figure}
\begin{figure}[!htbp]
\centering
\begin{subfigure}{0.45\textwidth}
  \centering
  \includegraphics[width=1.\linewidth]{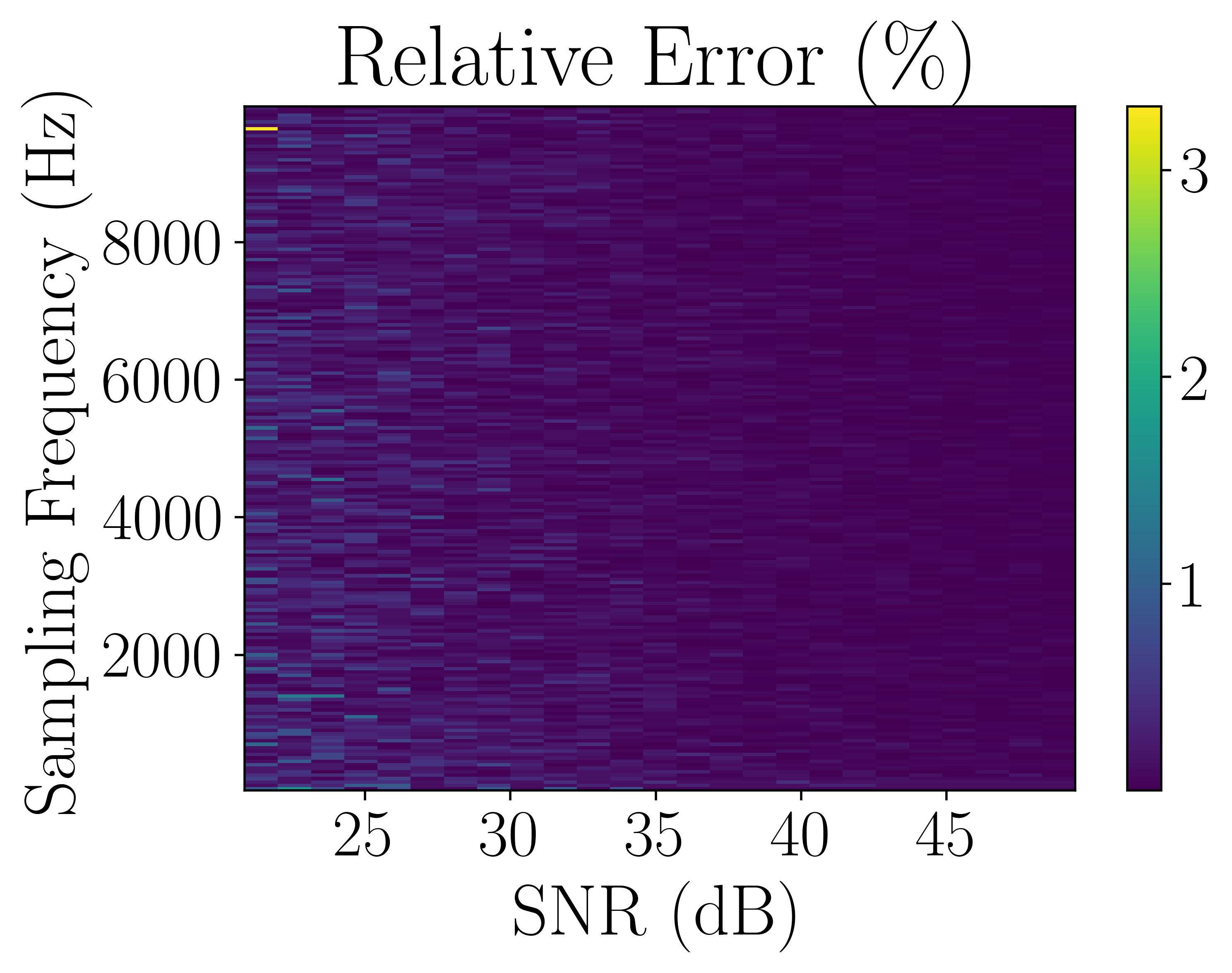}
    \caption{}
\end{subfigure}%
\begin{subfigure}{0.45\textwidth}
  \centering
  \includegraphics[width=1.\linewidth]{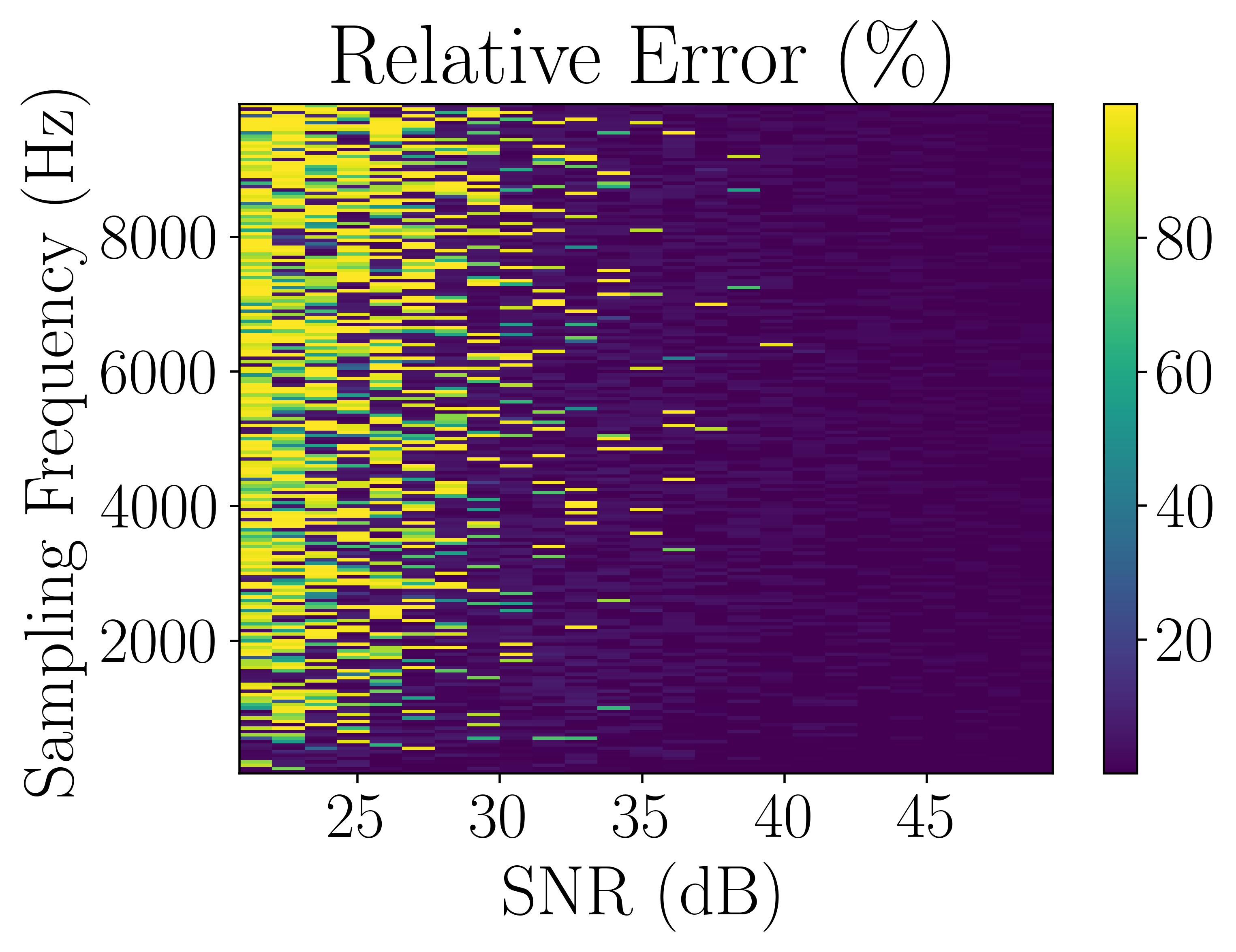}
      \caption{}
\end{subfigure}
\begin{subfigure}{0.45\textwidth}
  \centering
  \includegraphics[width=1.\linewidth]{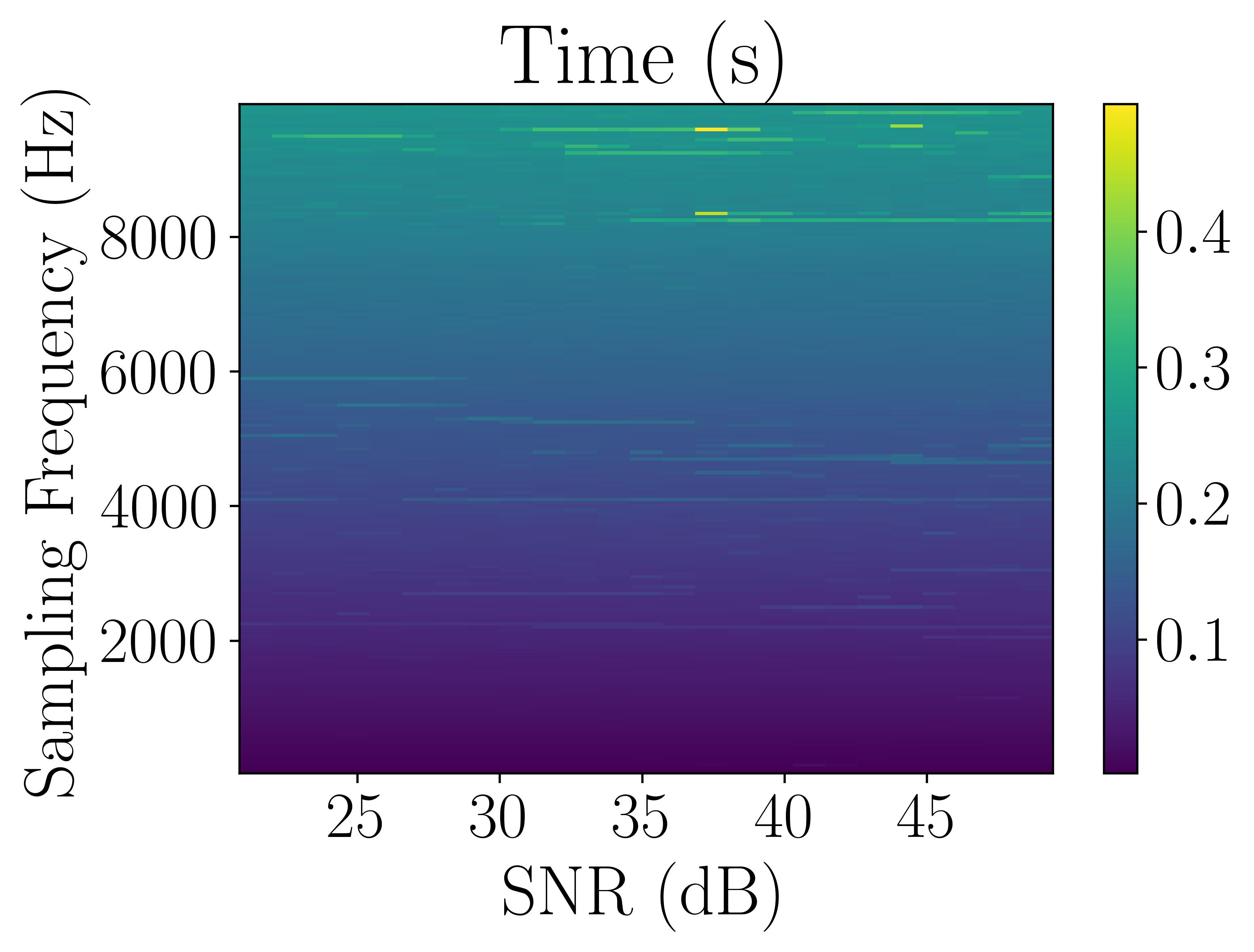}
      \caption{}
\end{subfigure}%
\begin{subfigure}{0.45\textwidth}
  \centering
  \includegraphics[width=1.\linewidth]{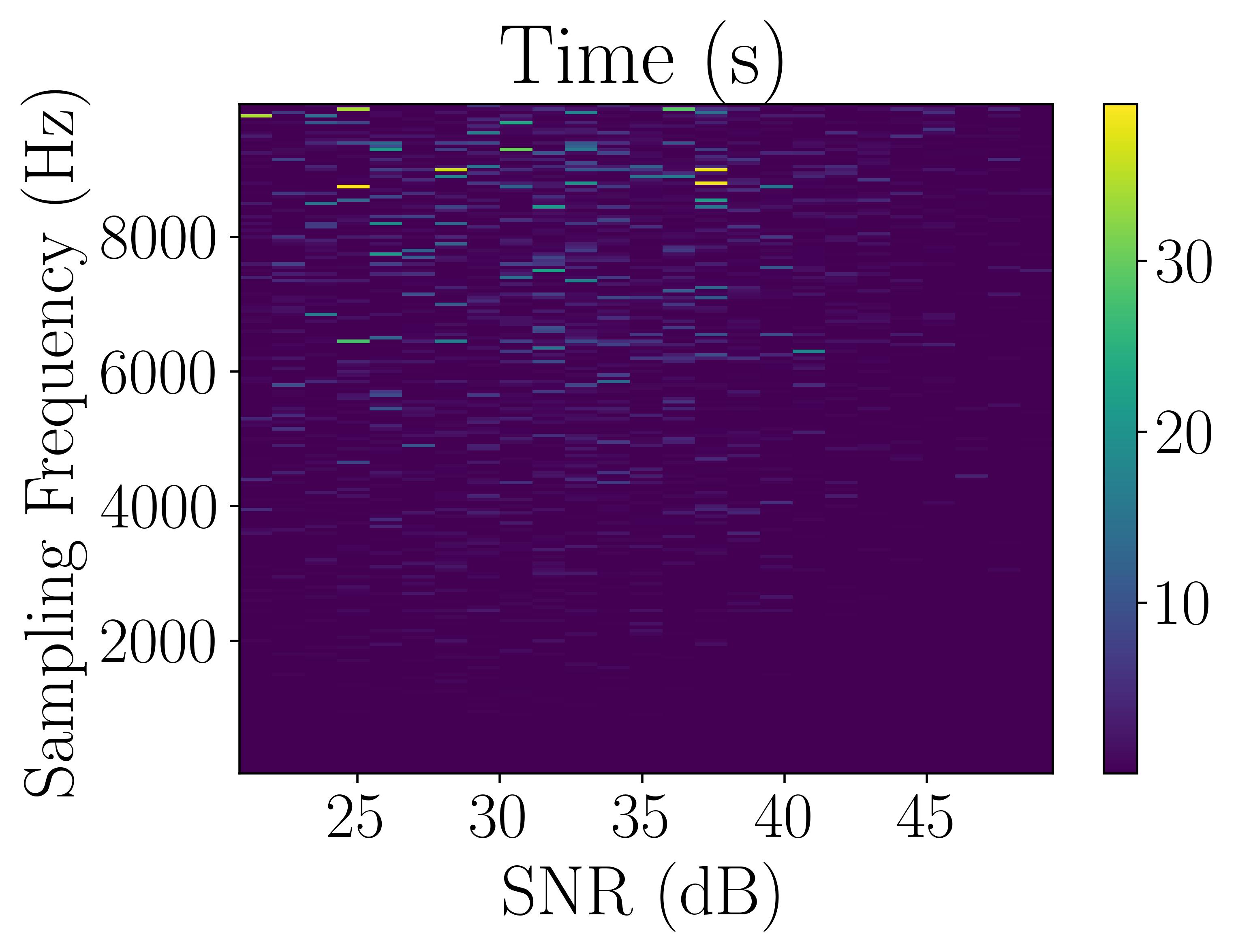}
      \caption{}
\end{subfigure}
\caption{Relative error and time taken for convergence by 0D Persistence (a, c) and Molinaro's algorithm (b, d) for $x_{8}$}
\label{fig:appD_fig7}
\end{figure}
\begin{figure}[!htbp]
\centering
\begin{subfigure}{0.45\textwidth}
  \centering
  \includegraphics[width=1.\linewidth]{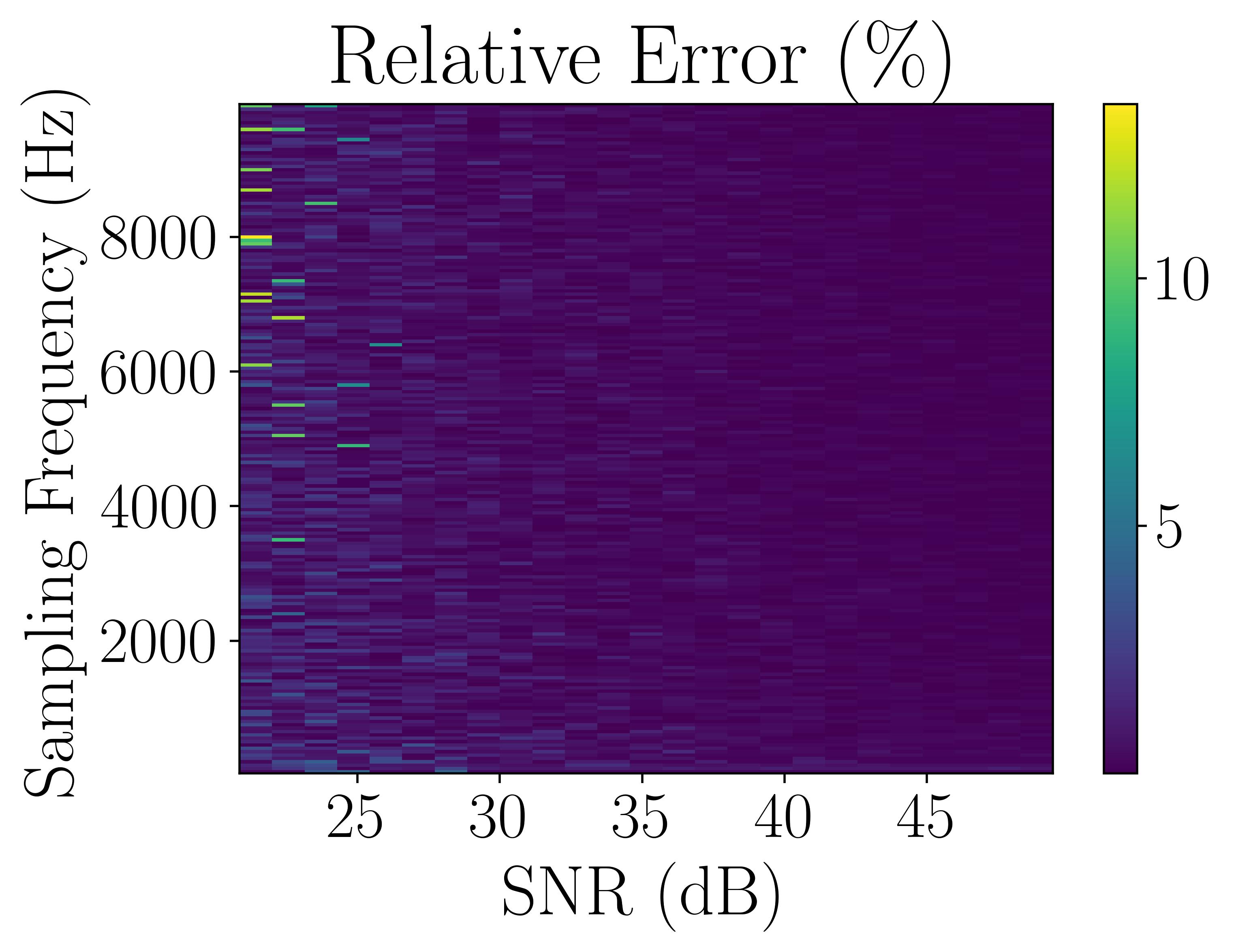}
    \caption{}
\end{subfigure}%
\begin{subfigure}{0.45\textwidth}
  \centering
  \includegraphics[width=1.\linewidth]{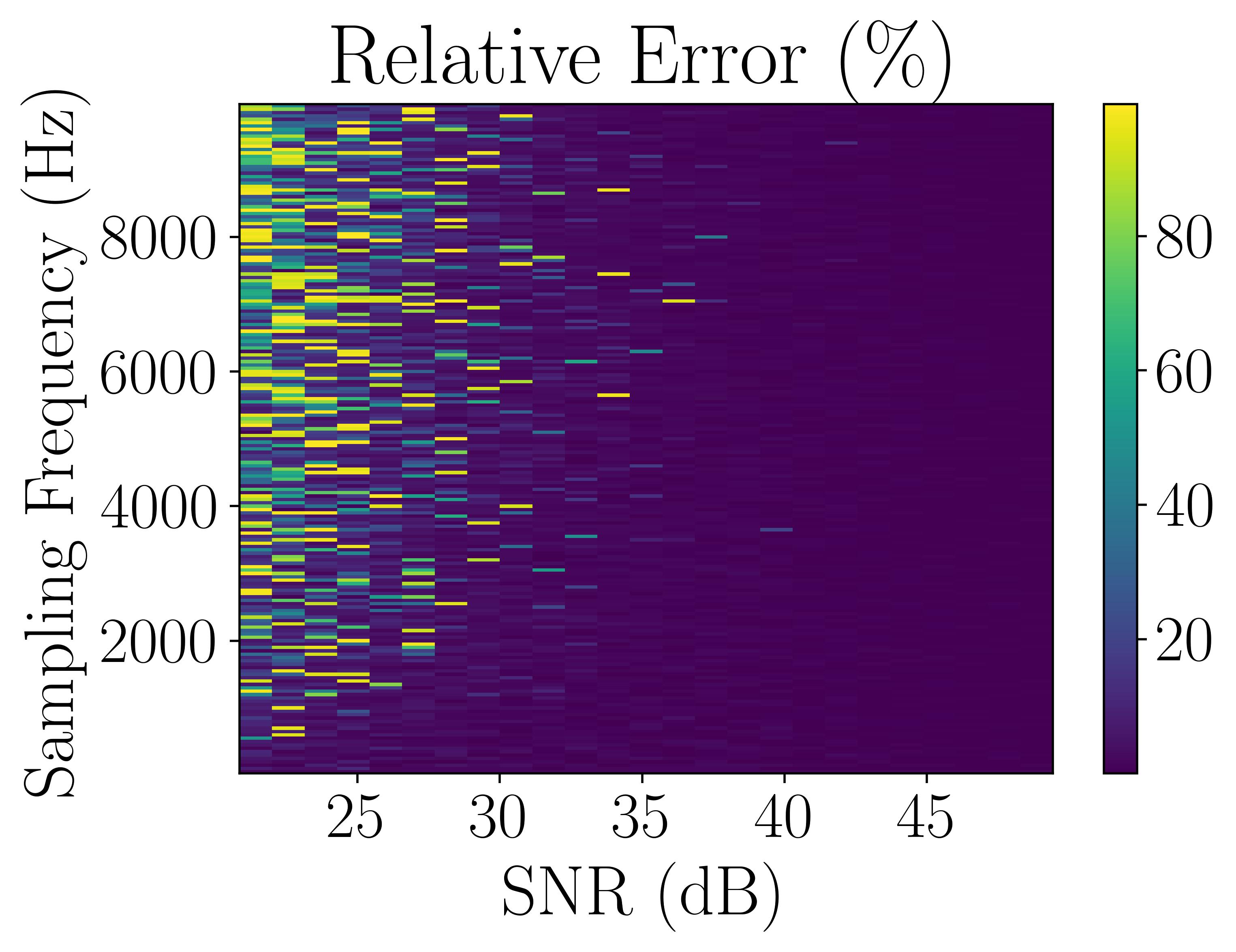}
      \caption{}
\end{subfigure}
\begin{subfigure}{0.45\textwidth}
  \centering
  \includegraphics[width=1.\linewidth]{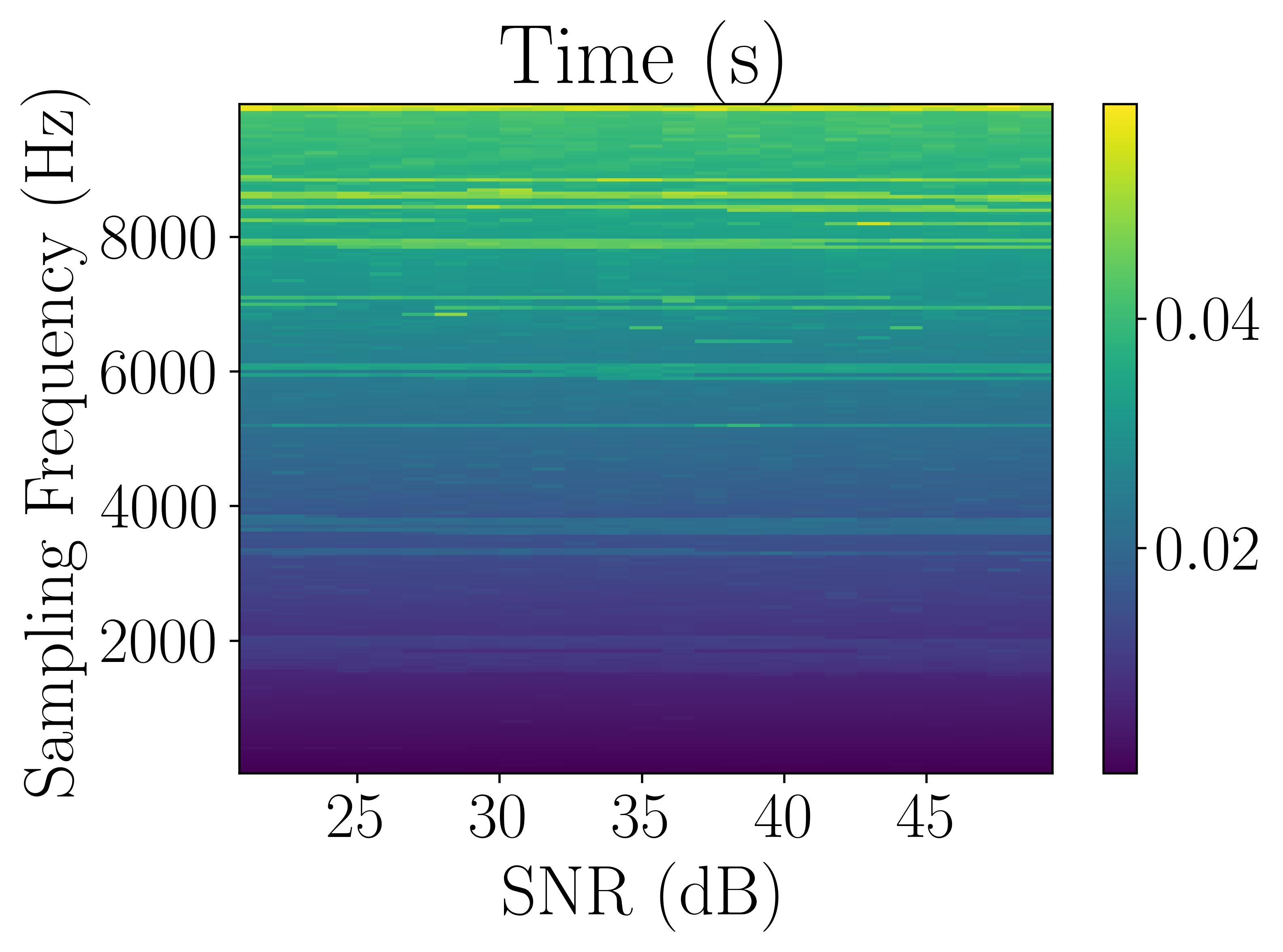}
      \caption{}
\end{subfigure}%
\begin{subfigure}{0.45\textwidth}
  \centering
  \includegraphics[width=1.\linewidth]{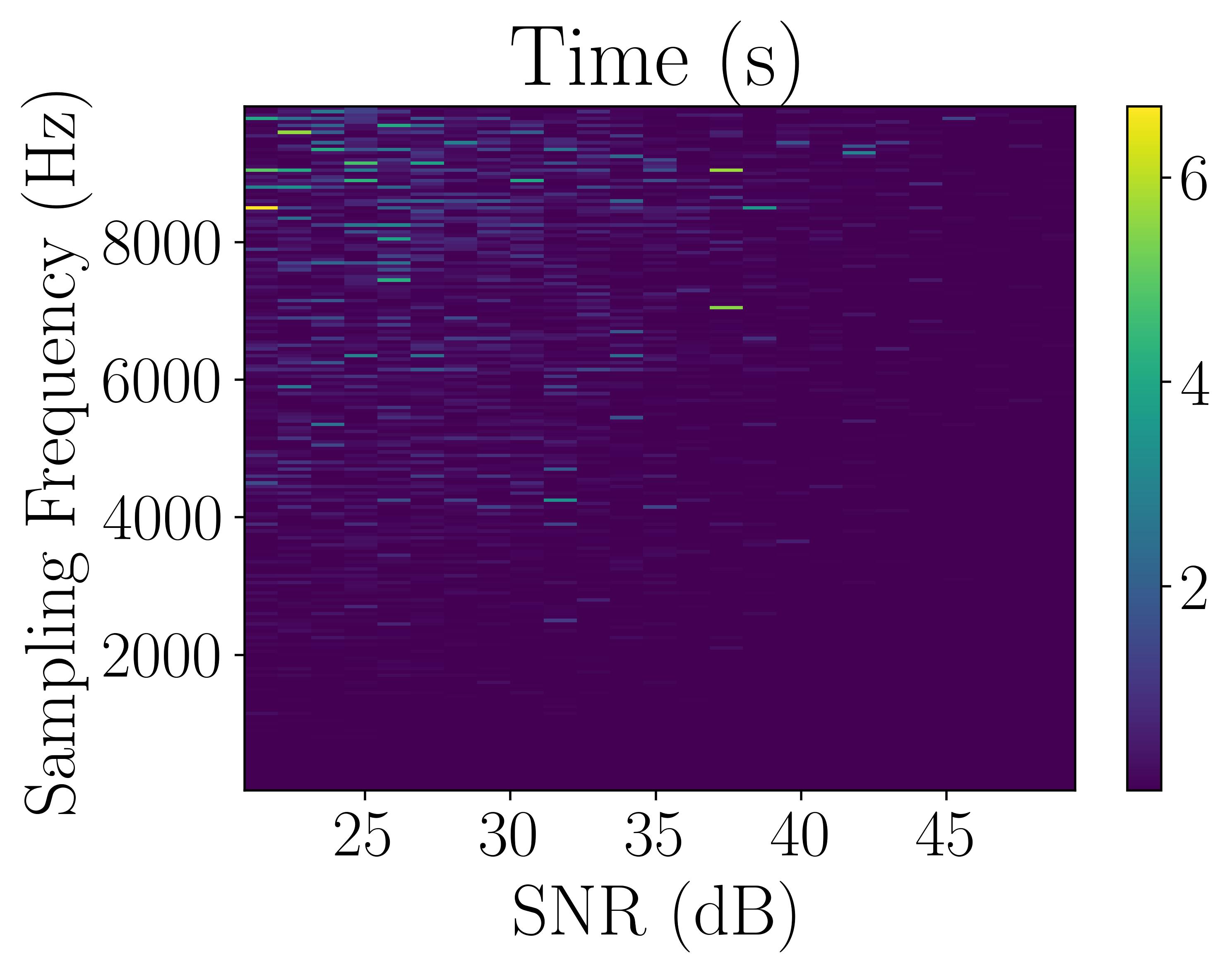}
      \caption{}
\end{subfigure}
\caption{Relative error and time taken for convergence by 0D Persistence (a, c) and Molinaro's algorithm (b, d) for $x_{9}$}
\label{fig:appD_fig8}
\end{figure}
\begin{figure}[!htbp]
\centering
\begin{subfigure}{0.45\textwidth}
  \centering
  \includegraphics[width=1.\linewidth]{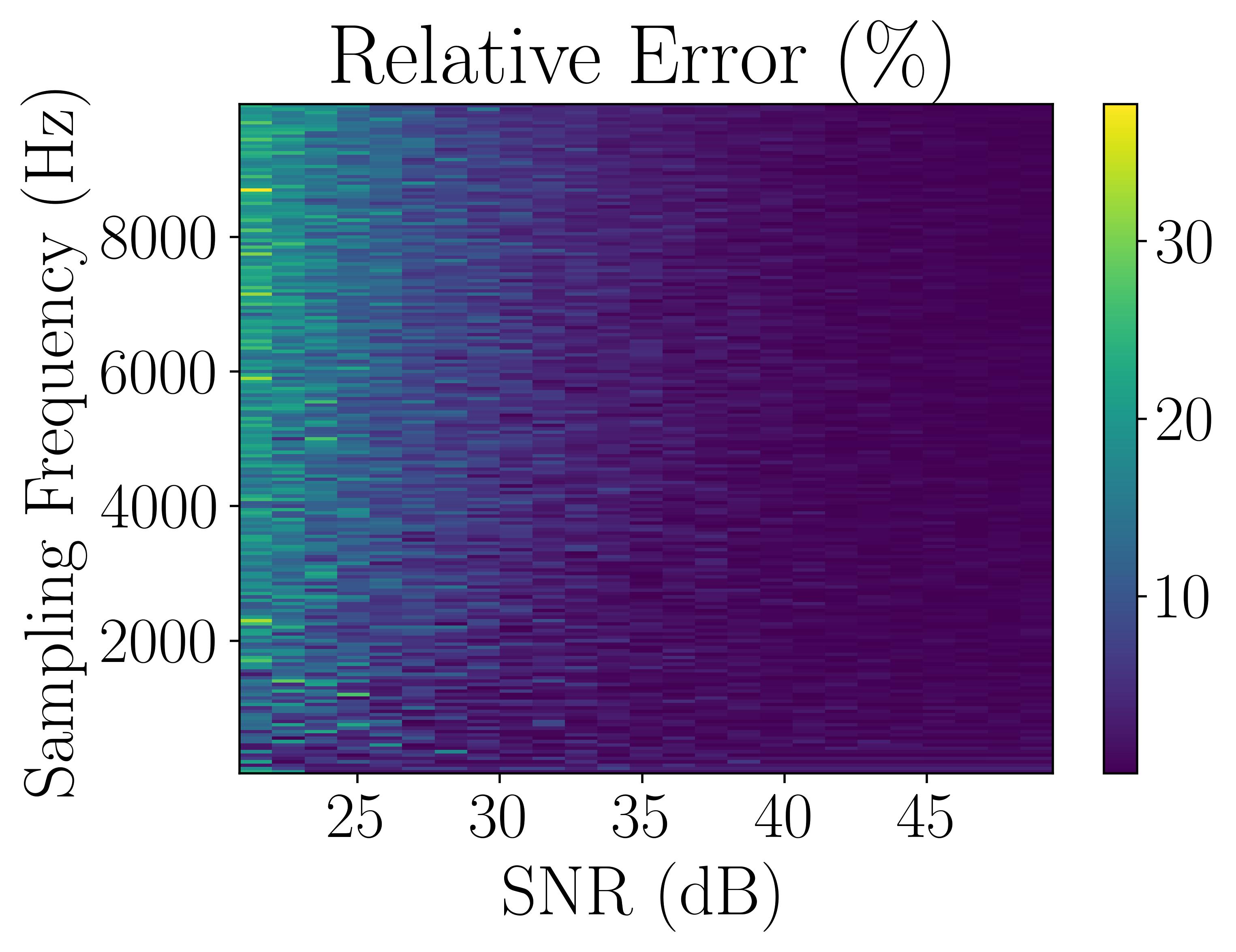}
    \caption{}
\end{subfigure}%
\begin{subfigure}{0.45\textwidth}
  \centering
  \includegraphics[width=1.\linewidth]{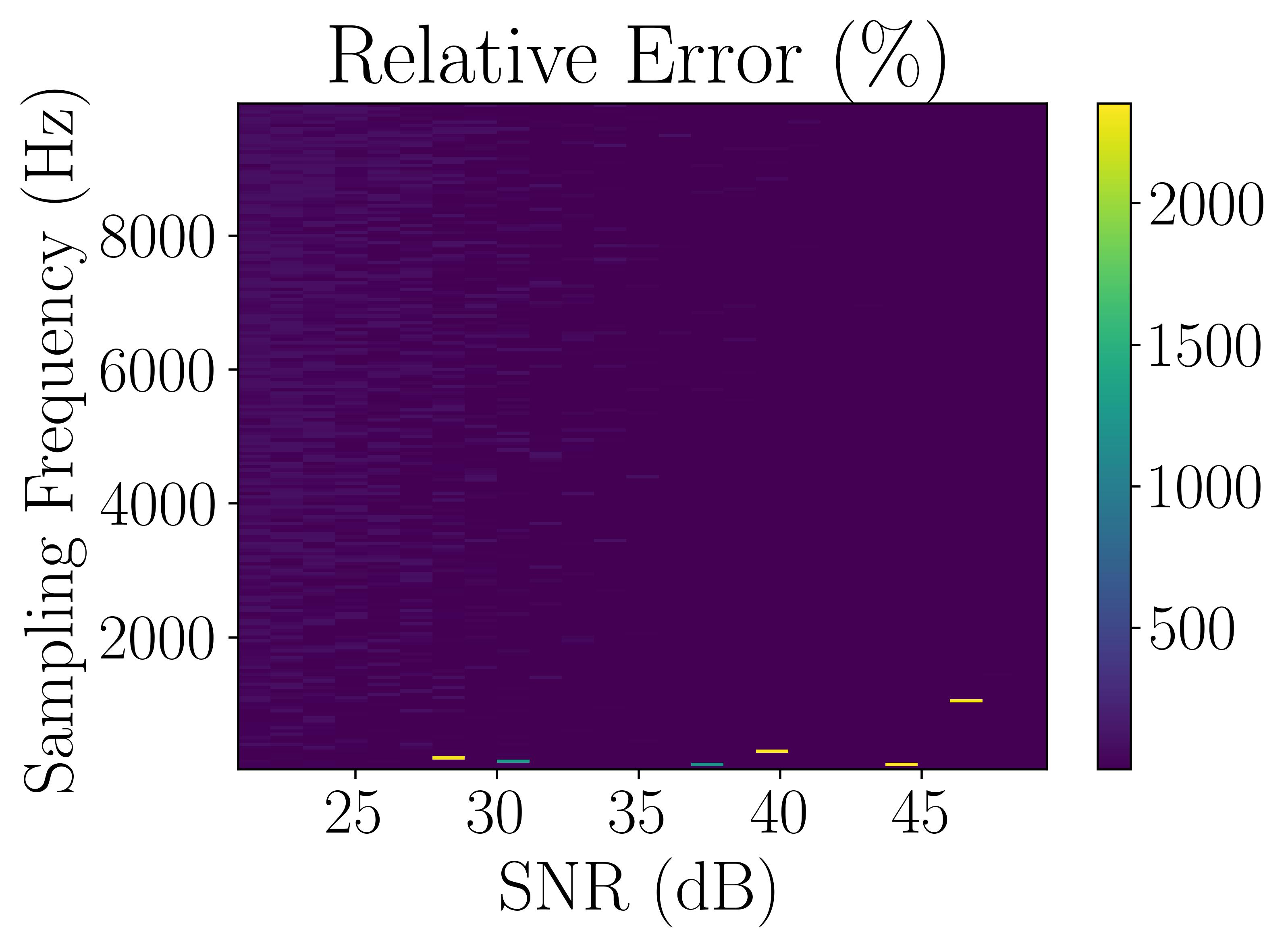}
      \caption{}
\end{subfigure}
\begin{subfigure}{0.45\textwidth}
  \centering
  \includegraphics[width=1.\linewidth]{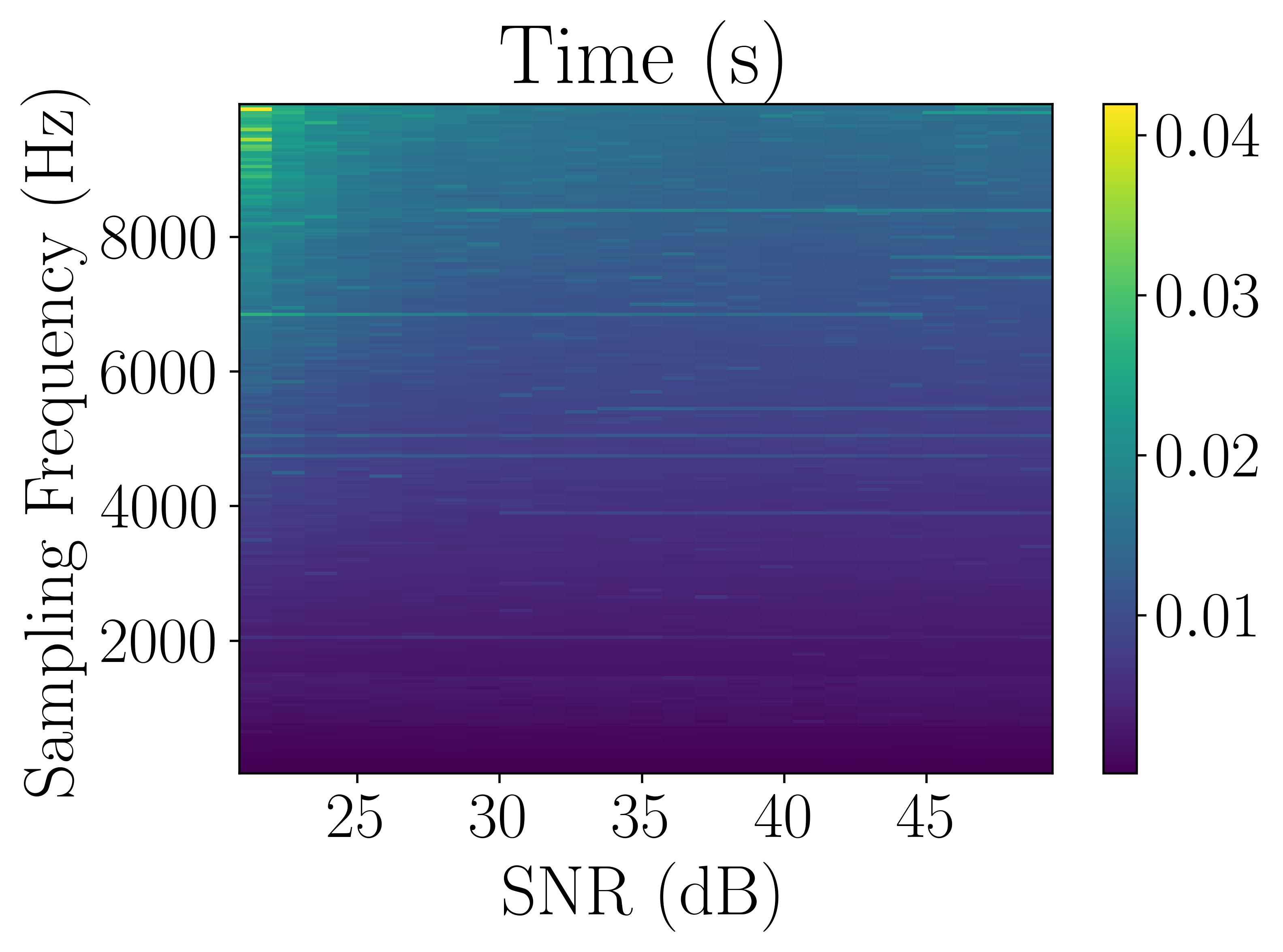}
      \caption{}
\end{subfigure}%
\begin{subfigure}{0.45\textwidth}
  \centering
  \includegraphics[width=1.\linewidth]{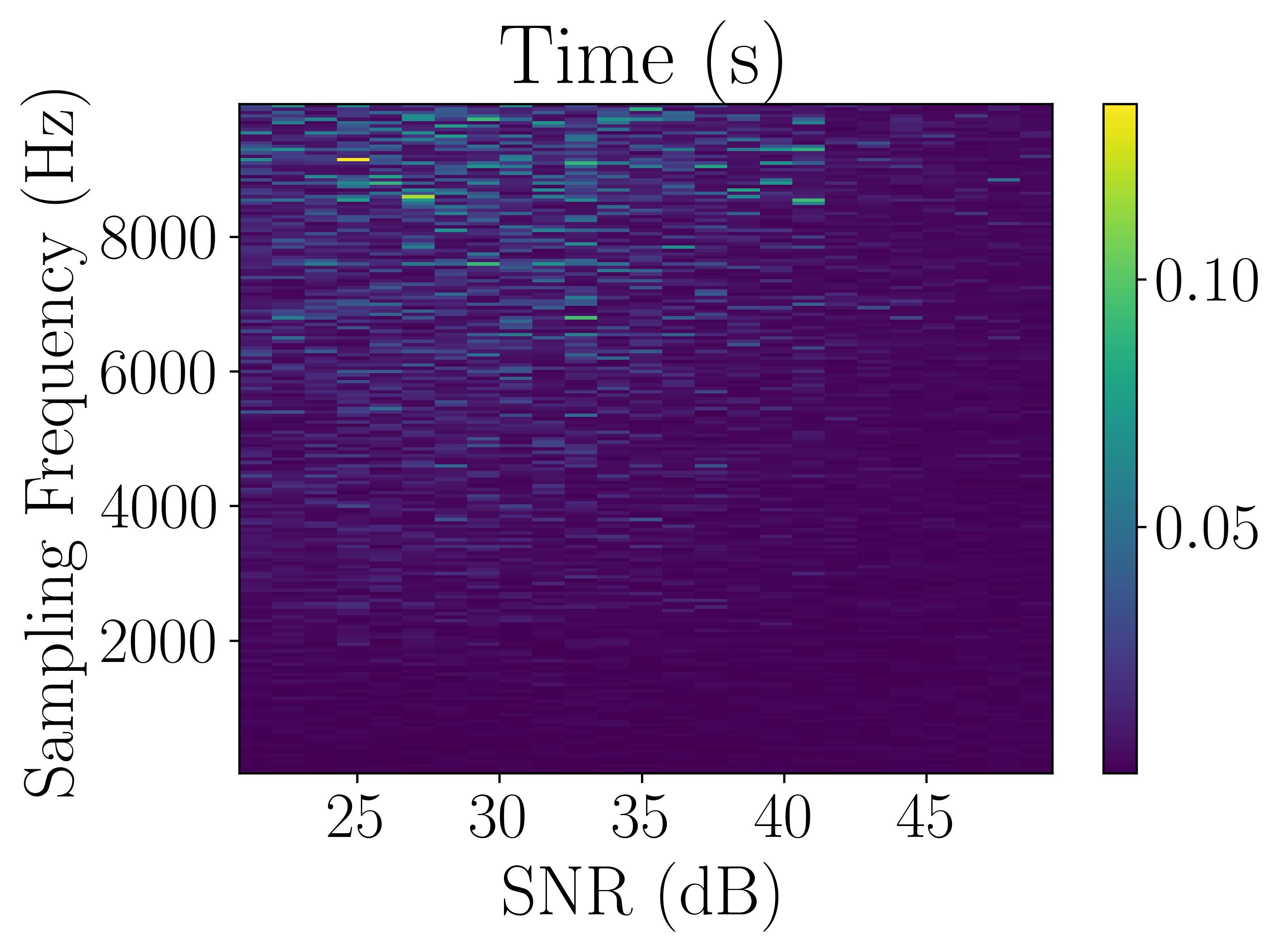}
      \caption{}
\end{subfigure}
\caption{Relative error and time taken for convergence by 0D Persistence (a, c) and Molinaro's algorithm (b, d) for $x_{10}$}
\label{fig:appD_fig9}
\end{figure}
\begin{figure}[!htbp]
\centering
\begin{subfigure}{0.45\textwidth}
  \centering
  \includegraphics[width=1.\linewidth]{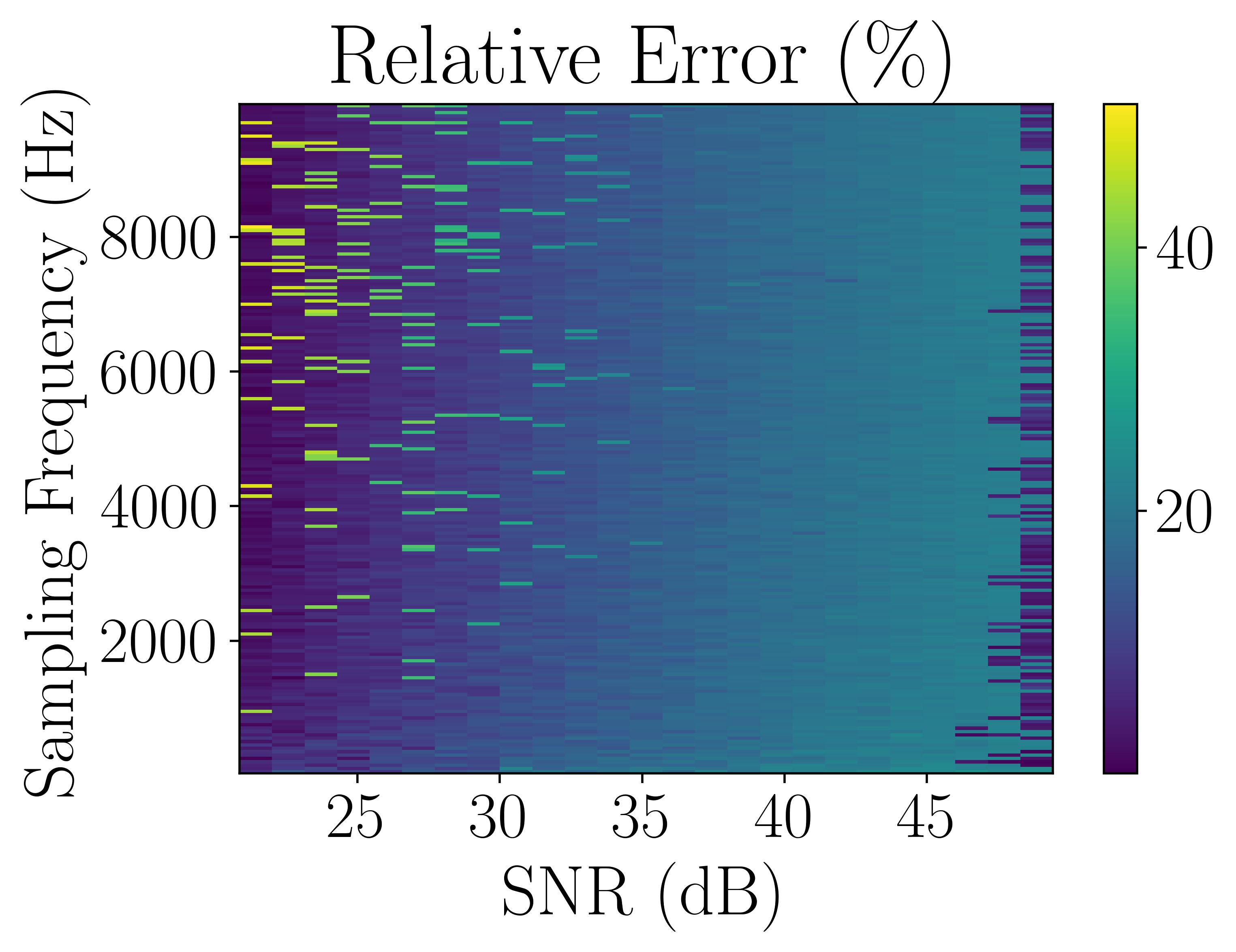}
    \caption{}
\end{subfigure}%
\begin{subfigure}{0.45\textwidth}
  \centering
  \includegraphics[width=1.\linewidth]{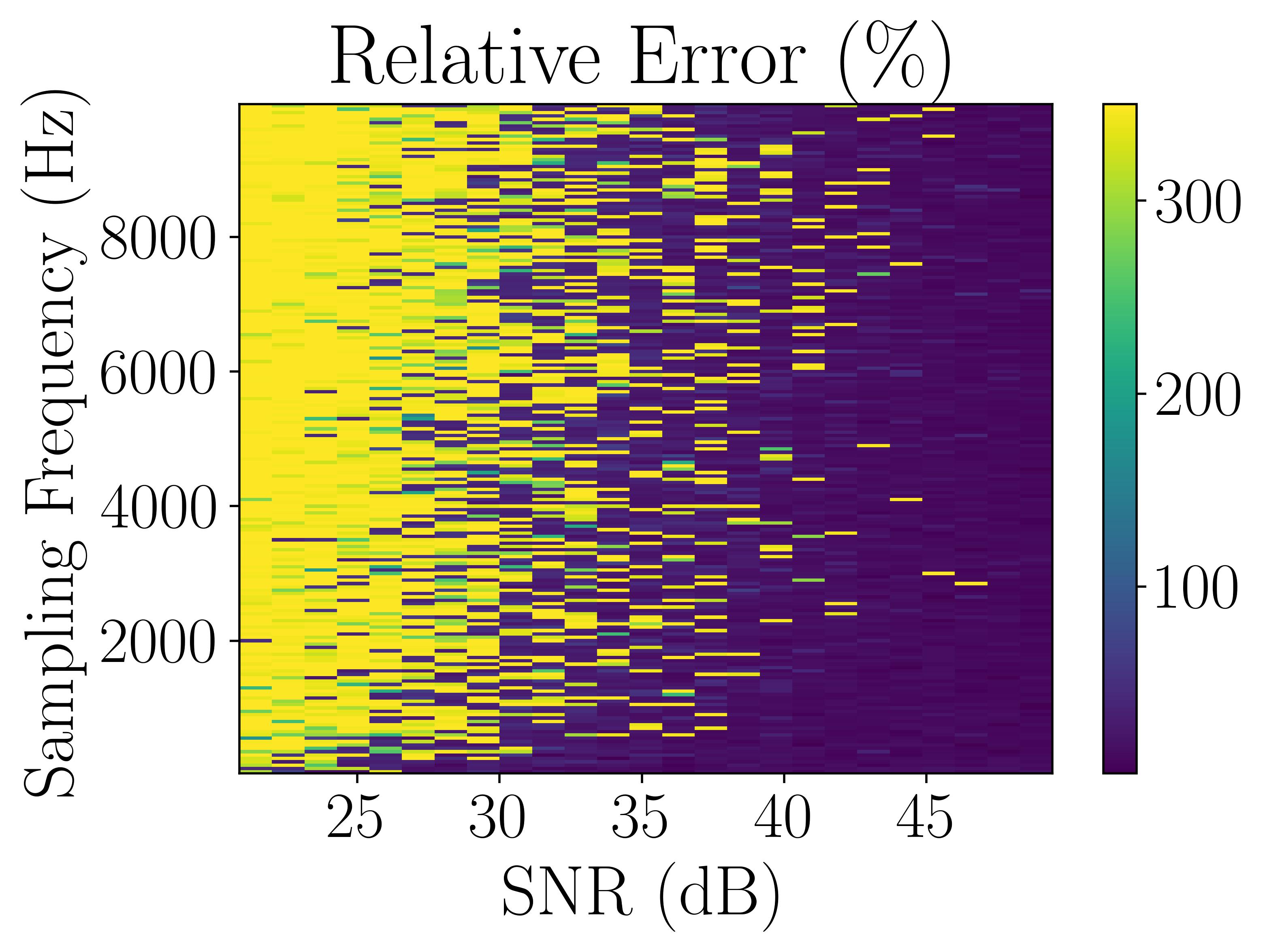}
      \caption{}
\end{subfigure}
\begin{subfigure}{0.45\textwidth}
  \centering
  \includegraphics[width=1.\linewidth]{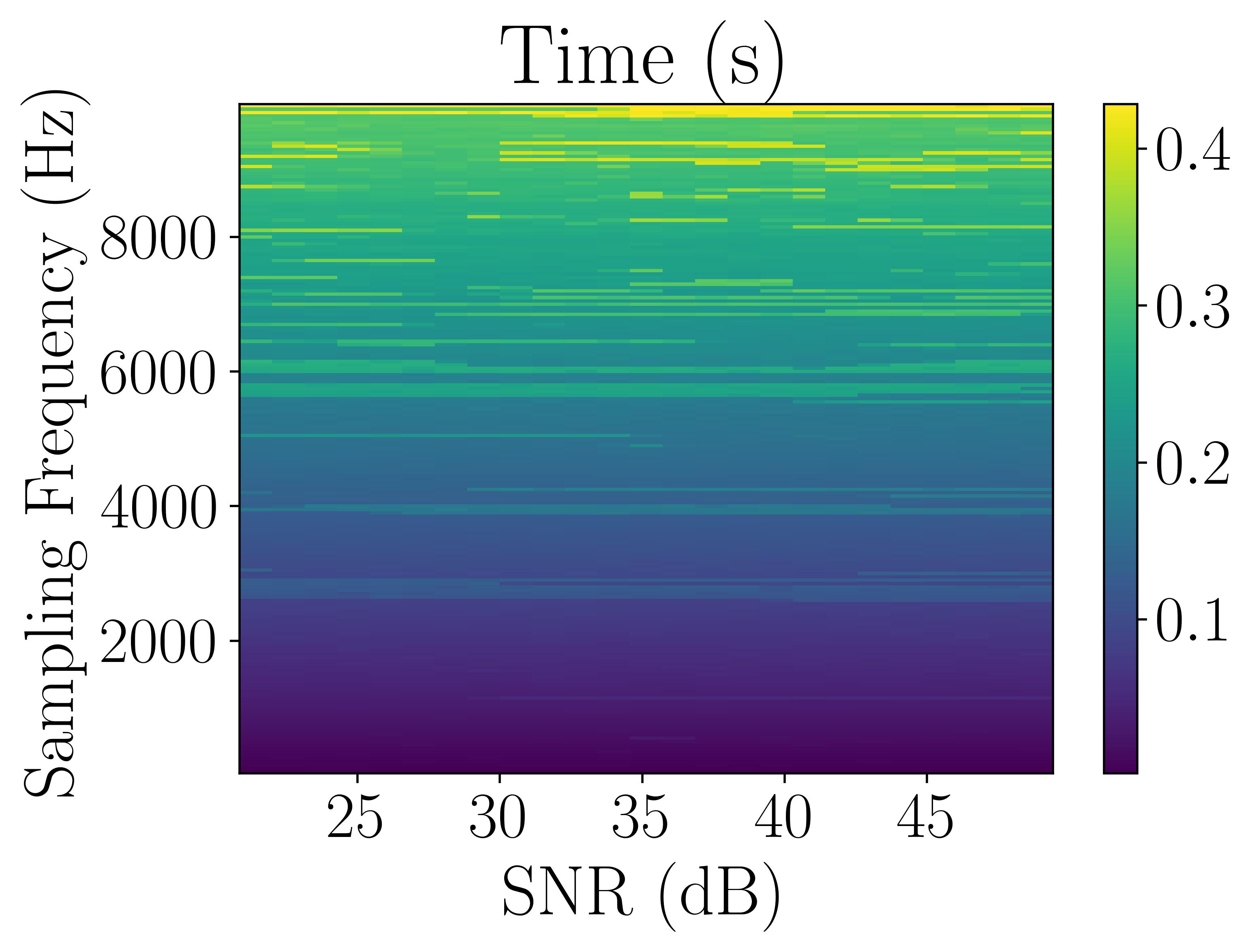}
      \caption{}
\end{subfigure}%
\begin{subfigure}{0.45\textwidth}
  \centering
  \includegraphics[width=1.\linewidth]{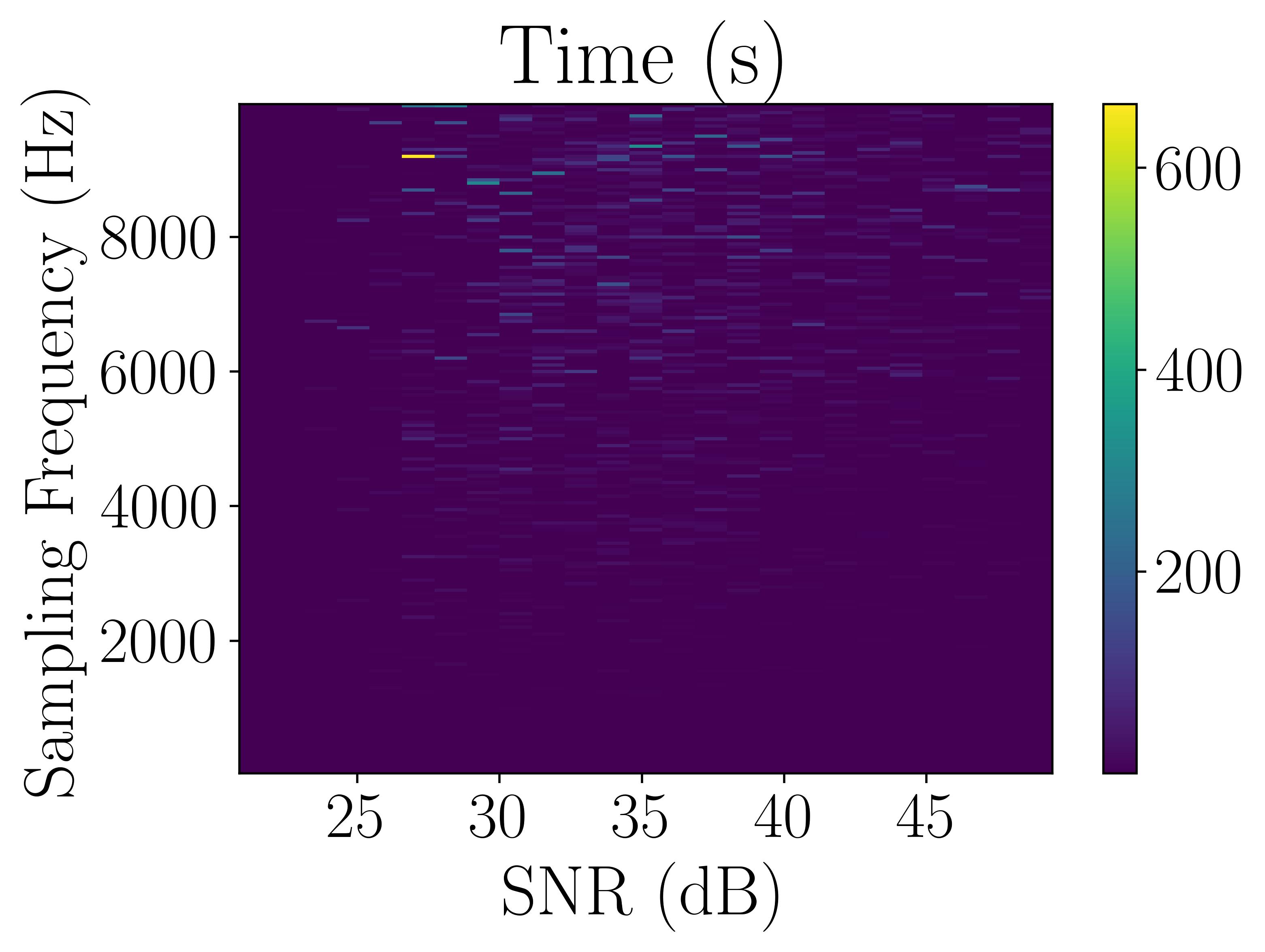}
      \caption{}
\end{subfigure}
\caption{Relative error and time taken for convergence by 0D Persistence (a, c) and Molinaro's algorithm (b, d) for $x_{11}$}
\label{fig:appD_fig10}
\end{figure}
\begin{figure}[!htbp]
\centering
\begin{subfigure}{0.45\textwidth}
  \centering
  \includegraphics[width=1.\linewidth]{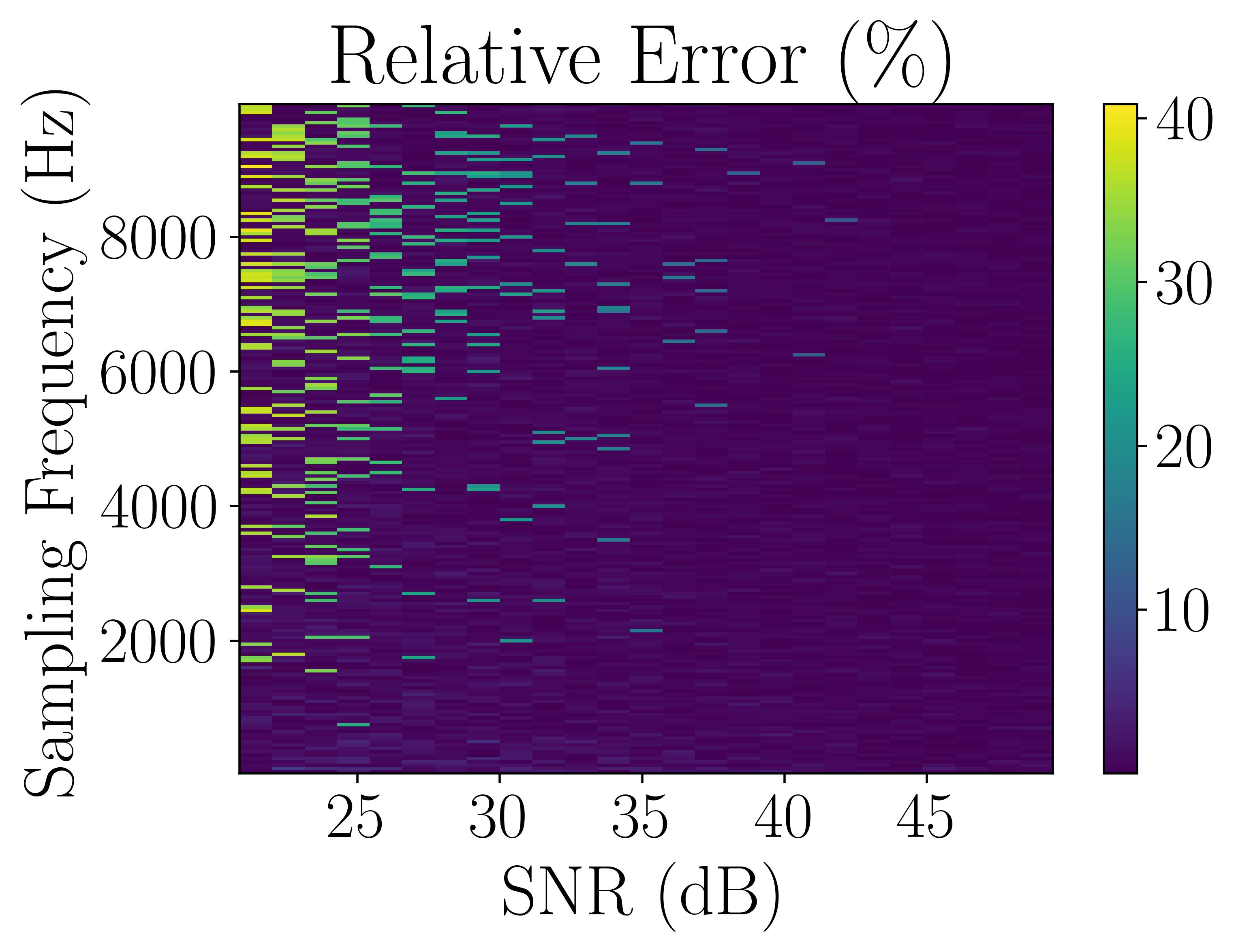}
    \caption{}
\end{subfigure}%
\begin{subfigure}{0.45\textwidth}
  \centering
  \includegraphics[width=1.\linewidth]{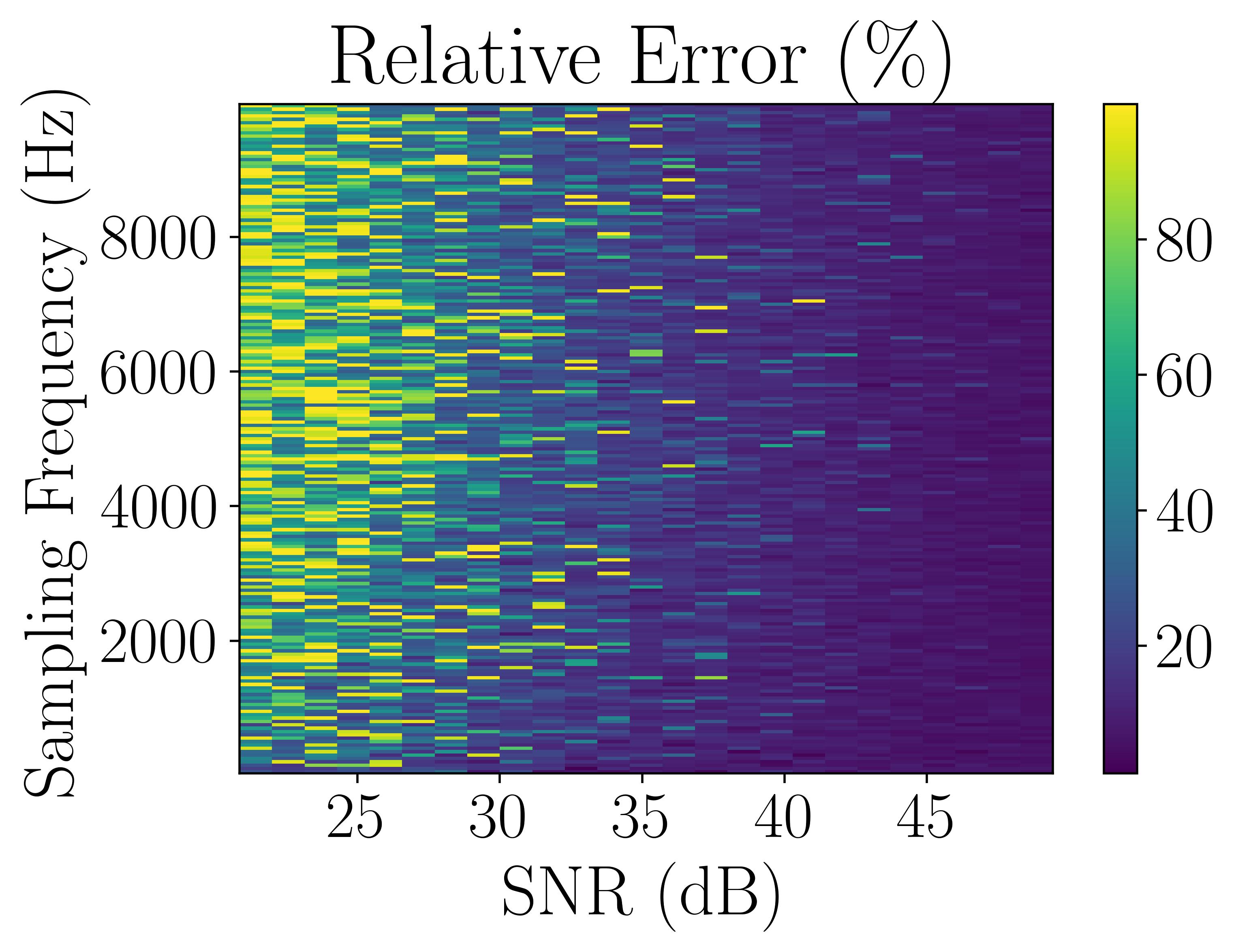}
      \caption{}
\end{subfigure}
\begin{subfigure}{0.45\textwidth}
  \centering
  \includegraphics[width=1.\linewidth]{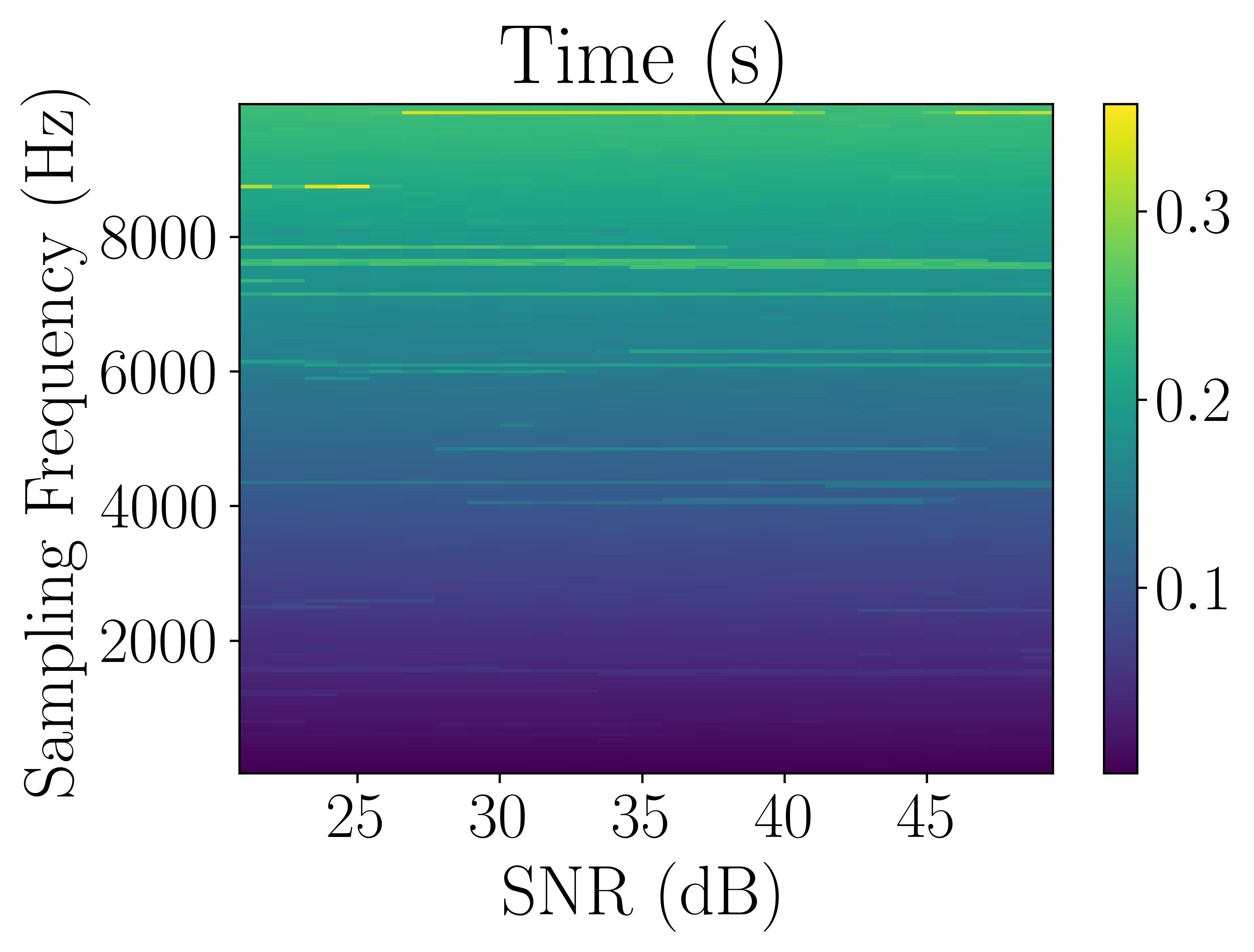}
      \caption{}
\end{subfigure}%
\begin{subfigure}{0.45\textwidth}
  \centering
  \includegraphics[width=1.\linewidth]{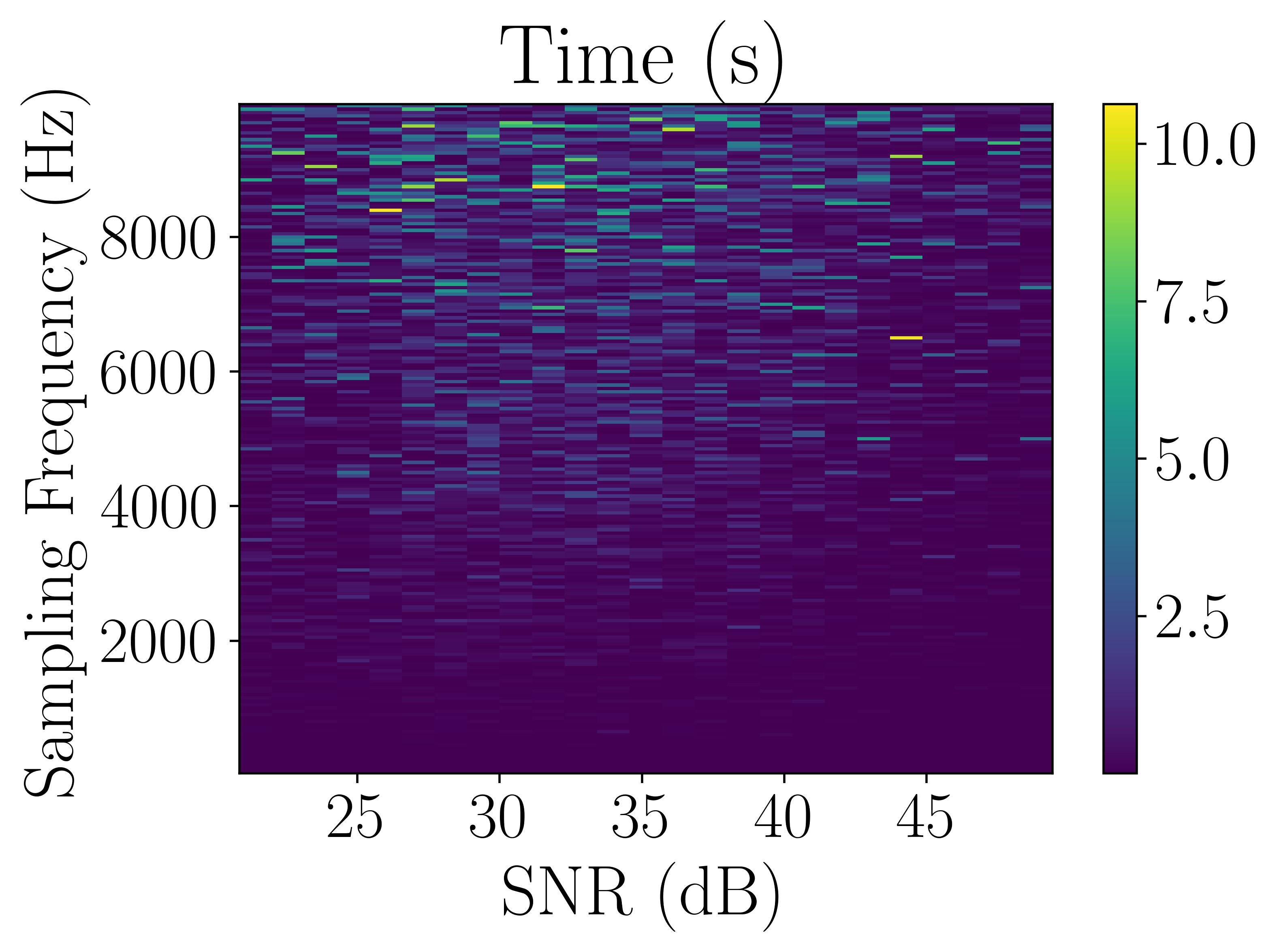}
      \caption{}
\end{subfigure}
\caption{Relative error and time taken for convergence by 0D Persistence (a, c) and Molinaro's algorithm (b, d) for $x_{12}$}
\label{fig:appD_fig11}
\end{figure}
\begin{figure}[!htbp]
\centering
\begin{subfigure}{0.45\textwidth}
  \centering
  \includegraphics[width=1.\linewidth]{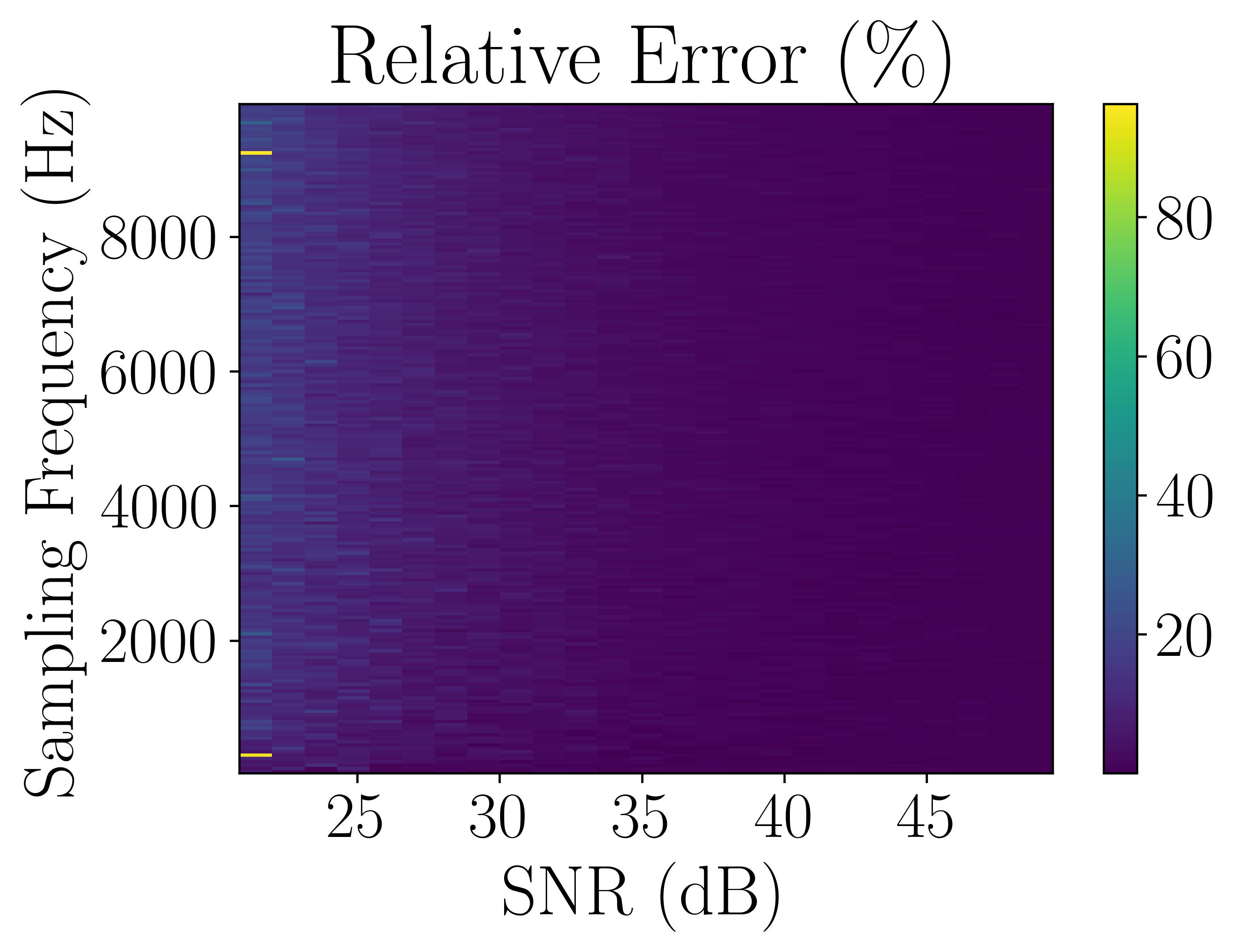}
    \caption{}
\end{subfigure}%
\begin{subfigure}{0.45\textwidth}
  \centering
  \includegraphics[width=1.\linewidth]{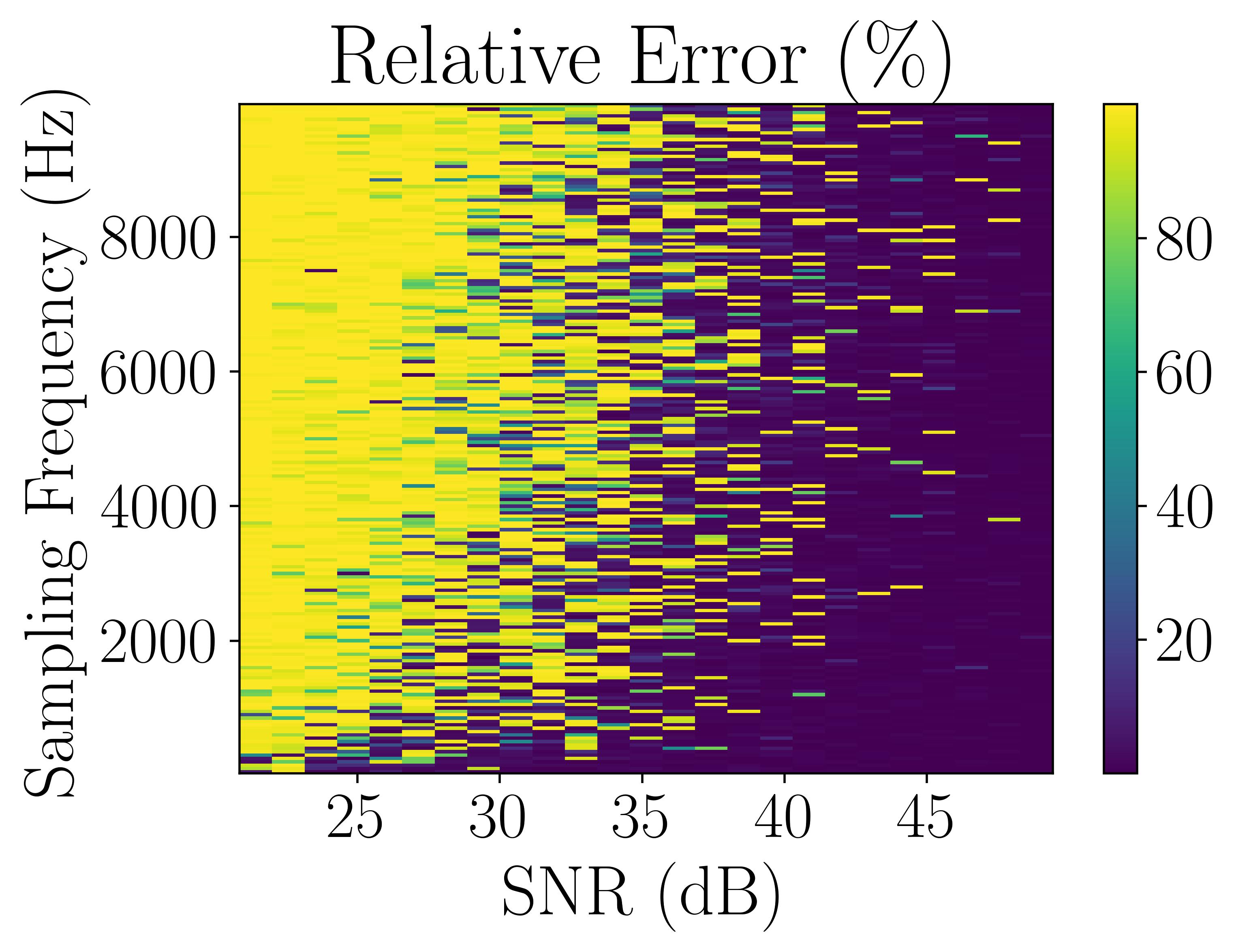}
      \caption{}
\end{subfigure}
\begin{subfigure}{0.45\textwidth}
  \centering
  \includegraphics[width=1.\linewidth]{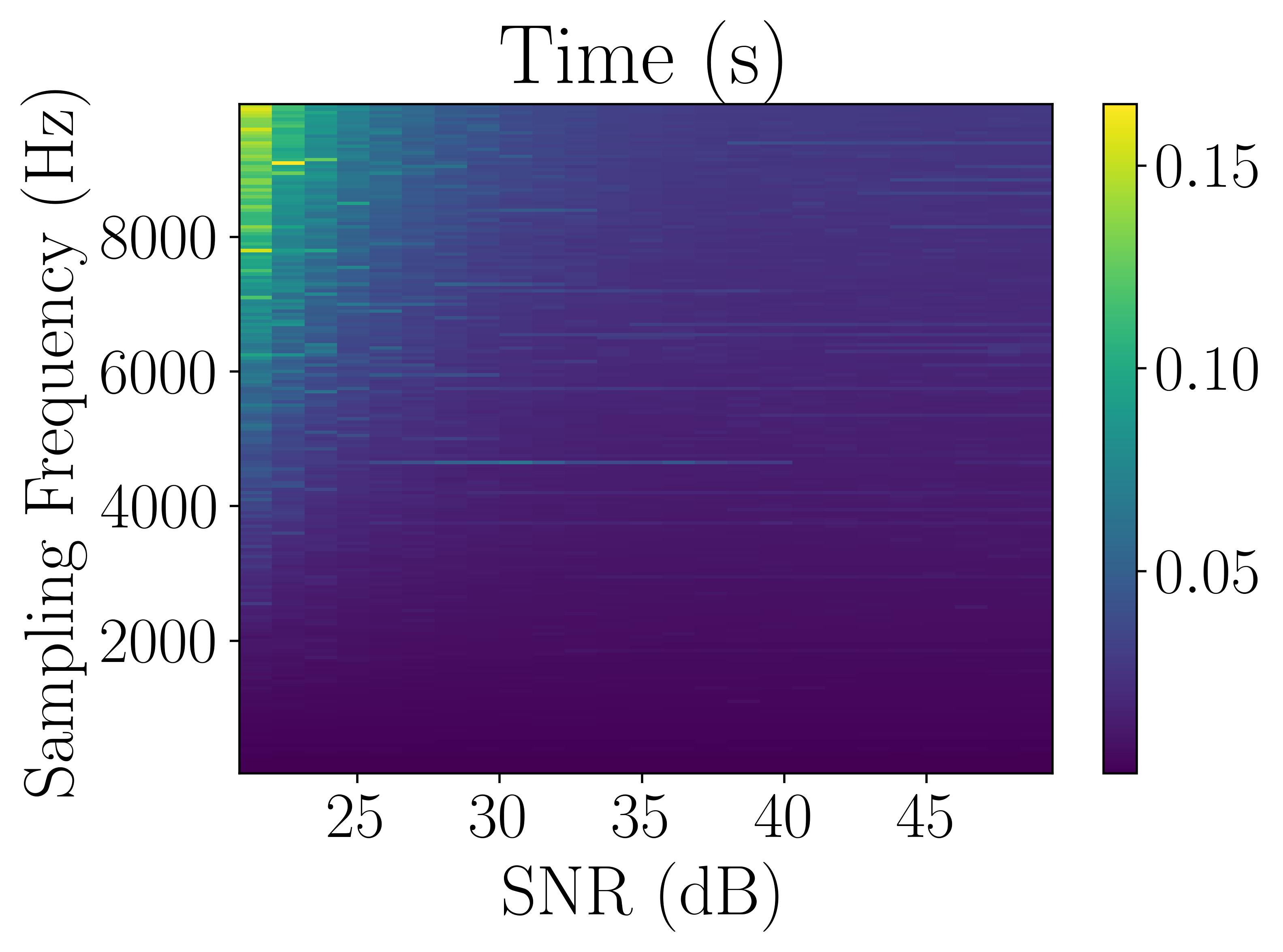}
      \caption{}
\end{subfigure}%
\begin{subfigure}{0.45\textwidth}
  \centering
  \includegraphics[width=1.\linewidth]{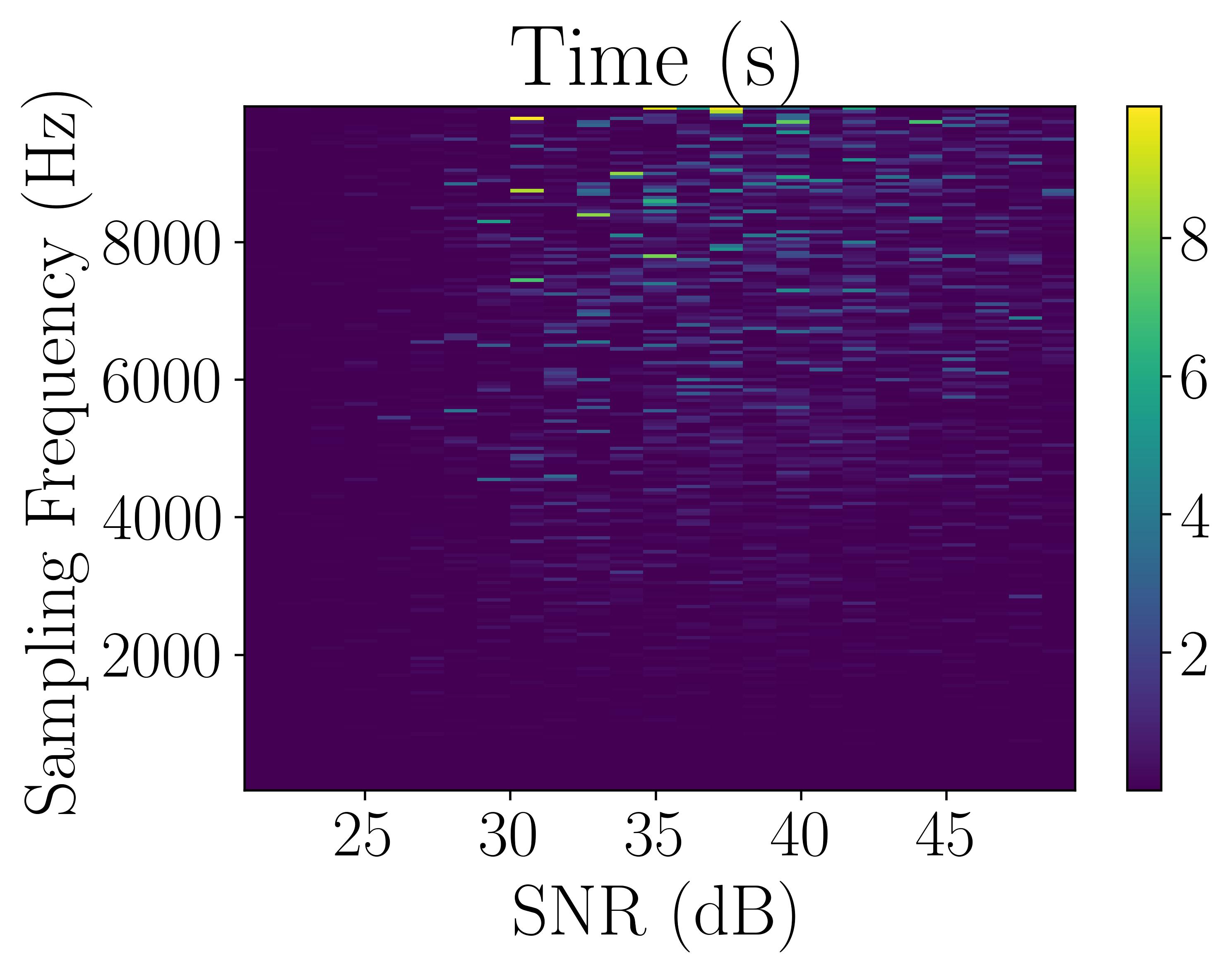}
      \caption{}
\end{subfigure}
\caption{Relative error and time taken for convergence by 0D Persistence (a, c) and Molinaro's algorithm (b, d) for $x_{13}$}
\label{fig:appD_fig12}
\end{figure}
\begin{figure}[!htbp]
\centering
\begin{subfigure}{0.45\textwidth}
  \centering
  \includegraphics[width=1.\linewidth]{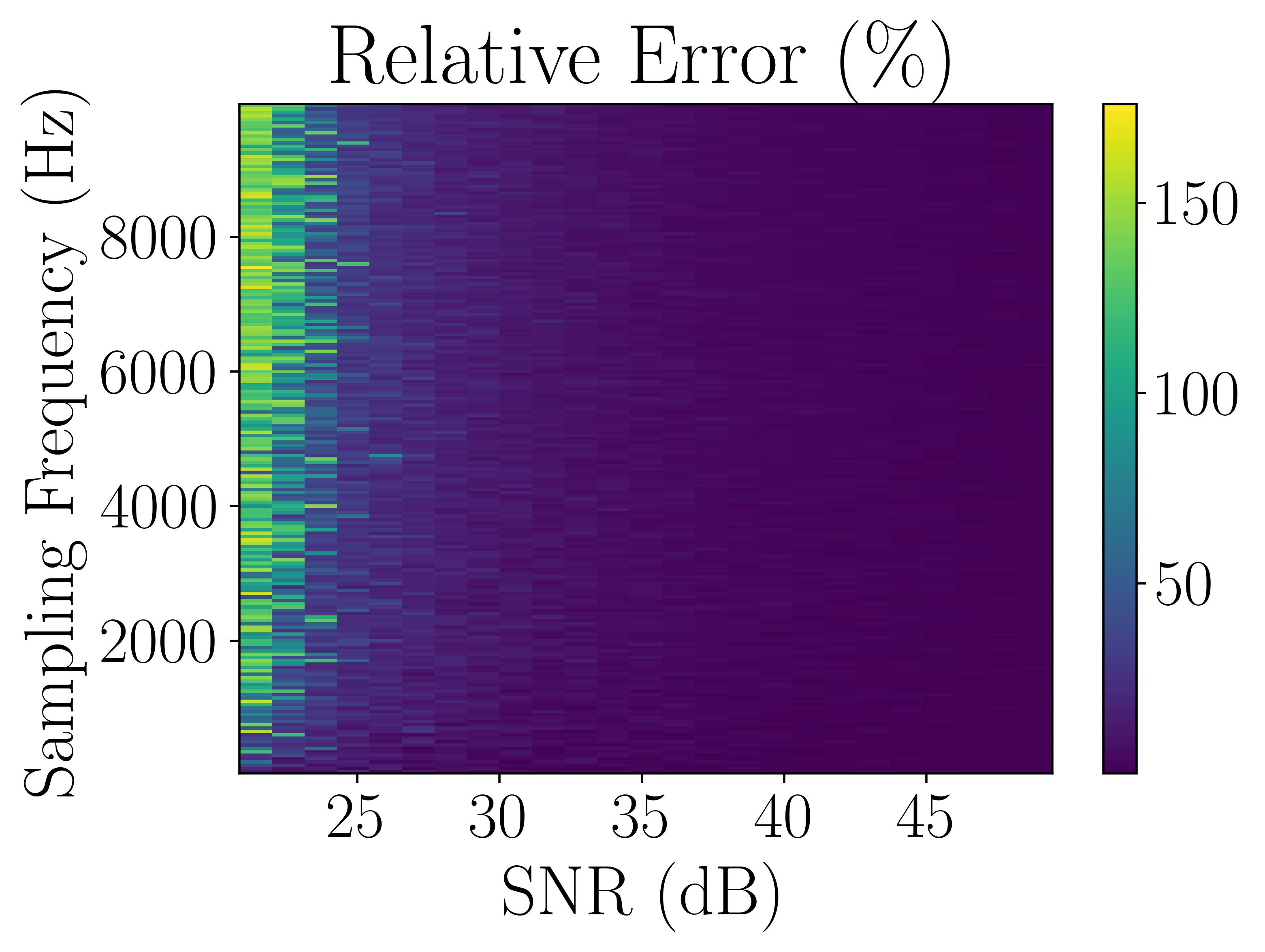}
    \caption{}
\end{subfigure}%
\begin{subfigure}{0.45\textwidth}
  \centering
  \includegraphics[width=1.\linewidth]{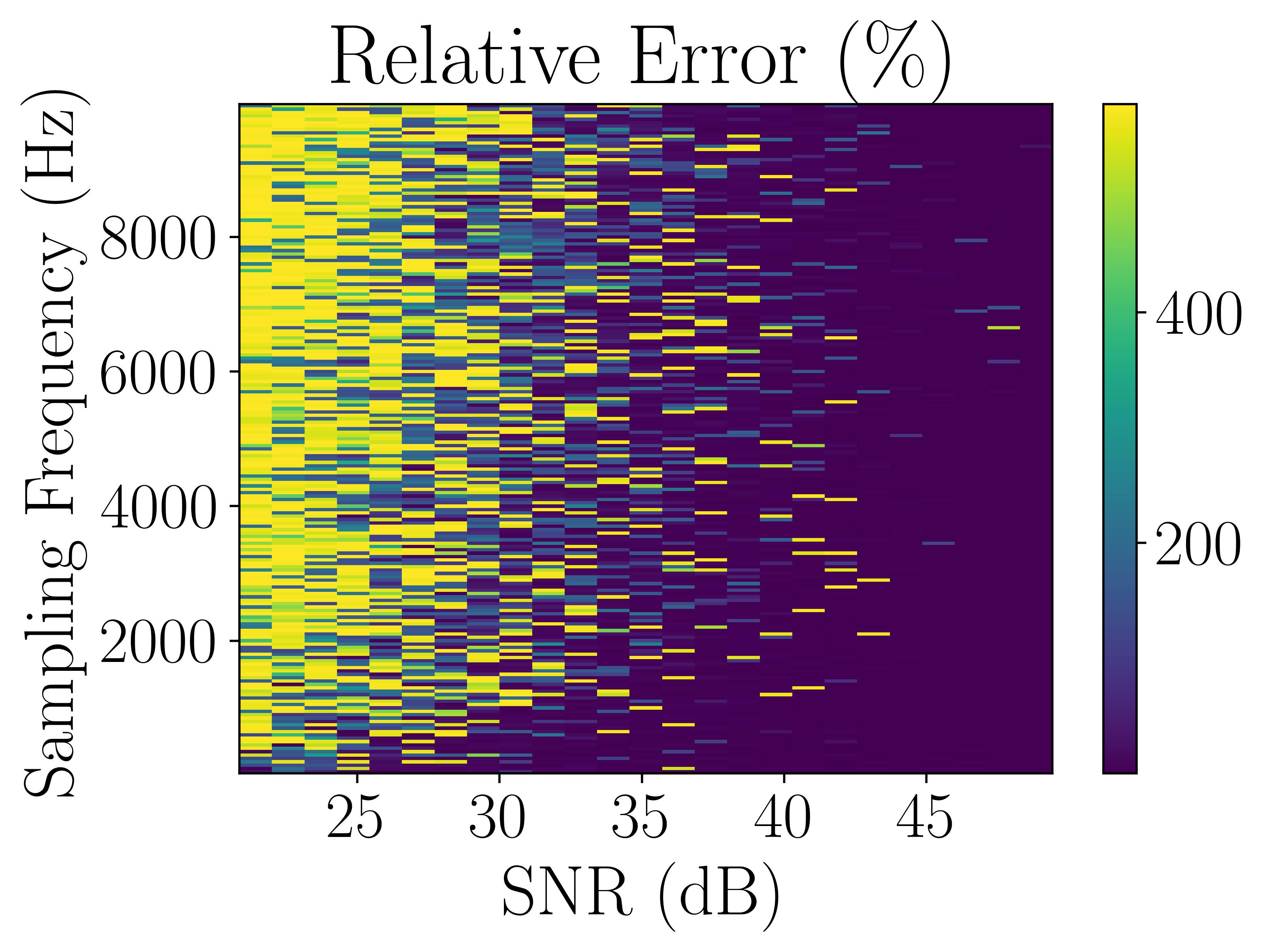}
      \caption{}
\end{subfigure}
\begin{subfigure}{0.45\textwidth}
  \centering
  \includegraphics[width=1.\linewidth]{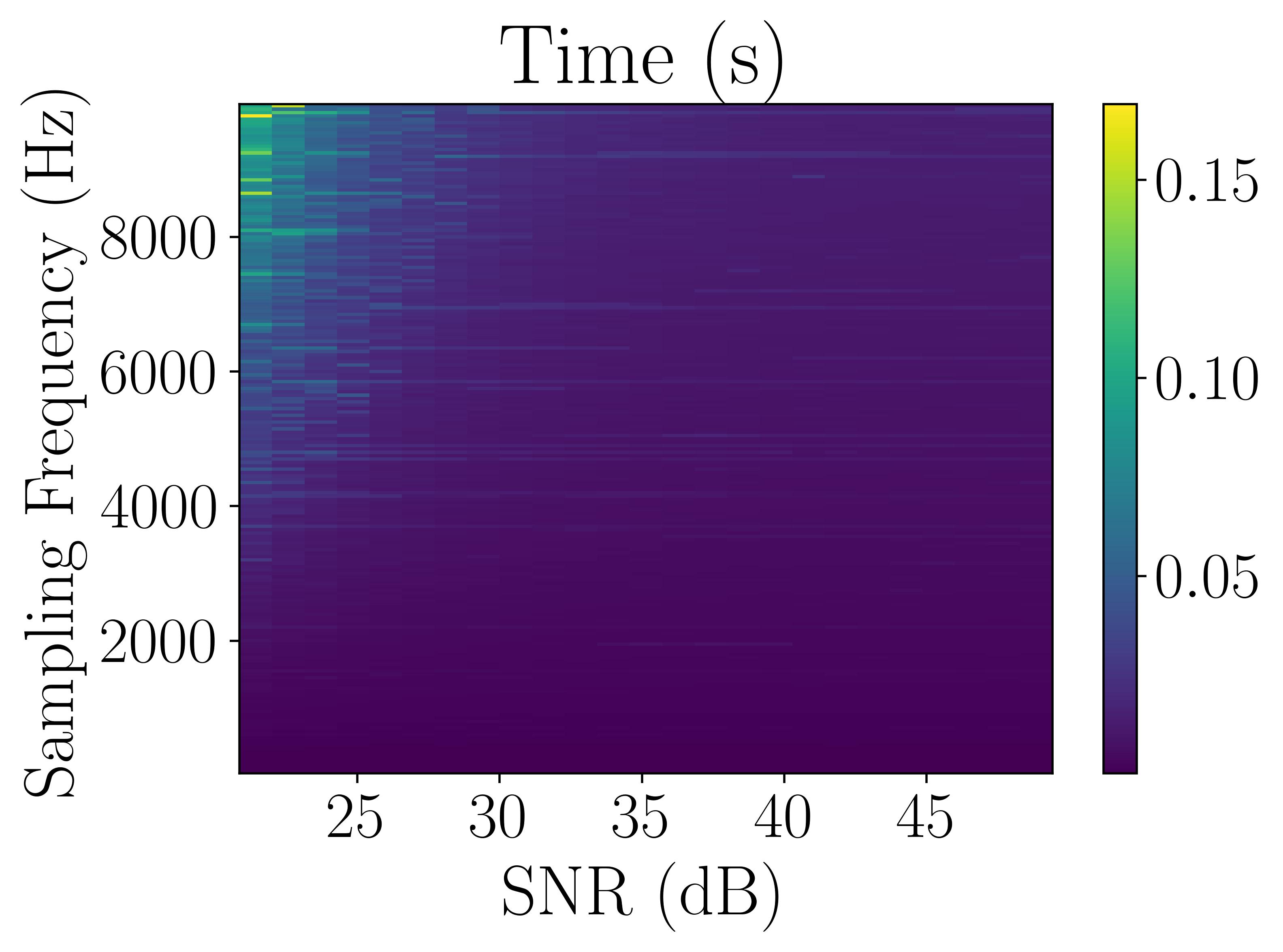}
      \caption{}
\end{subfigure}%
\begin{subfigure}{0.45\textwidth}
  \centering
  \includegraphics[width=1.\linewidth]{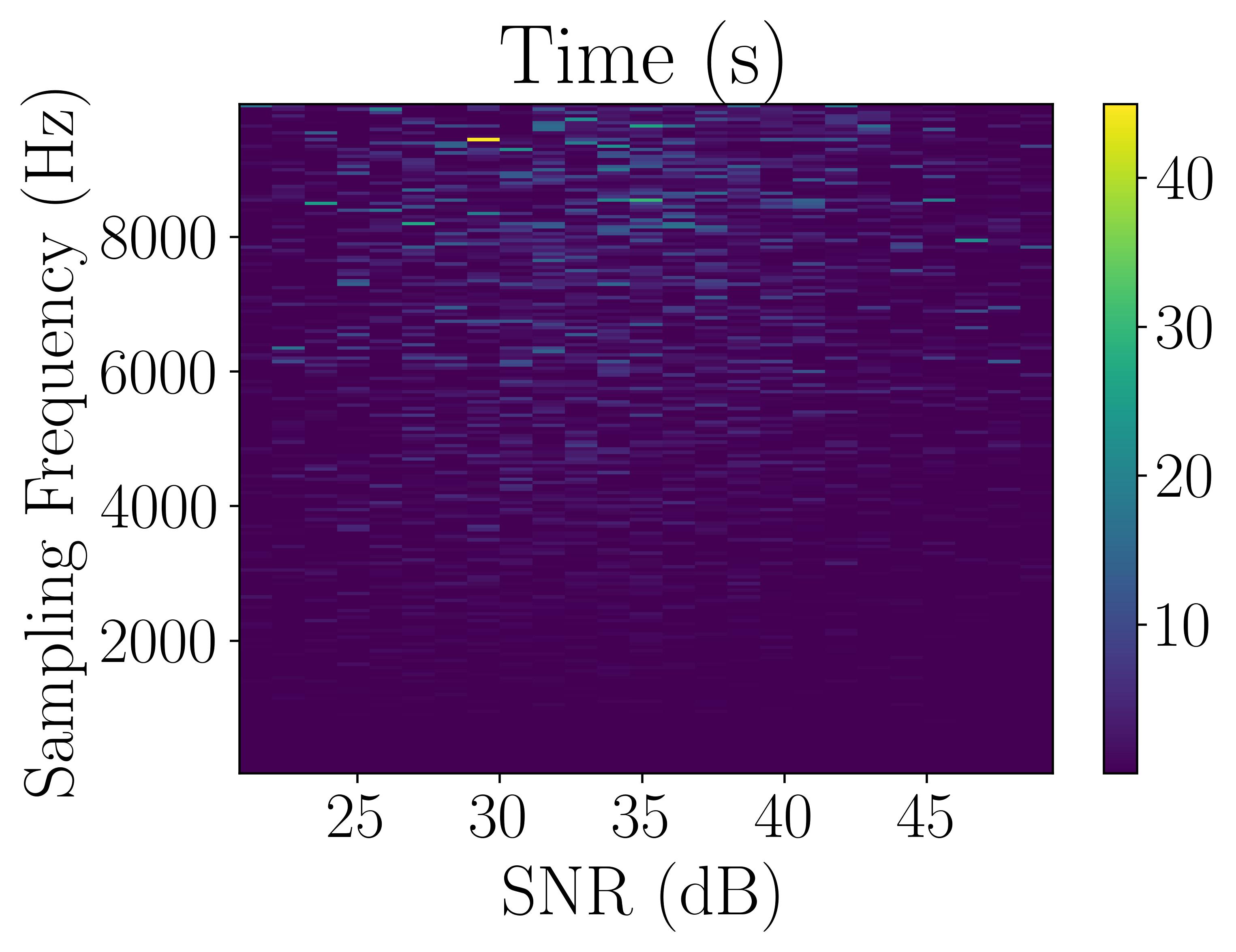}
      \caption{}
\end{subfigure}
\caption{Relative error and time taken for convergence by 0D Persistence (a, c) and Molinaro's algorithm (b, d) for $x_{14}$}
\label{fig:appD_fig13}
\end{figure}